\documentclass[sigconf]{acmart}
\usepackage{amsmath}

\newcommand\numberthis{\addtocounter{equation}{1}\tag{\theequation}}
\usepackage{amsmath,amsfonts}
\usepackage{mathtools}
\usepackage{graphicx}
\usepackage{textcomp}
\usepackage{comment}
\usepackage{color}
\usepackage{enumitem}
\setlist[itemize]{noitemsep, nolistsep}
\newif\ifpaper
\paperfalse   
\usepackage{hyperref}
\usepackage{csquotes}
 \DeclareMathAlphabet{\mathpzc}{OT1}{pzc}{m}{it}  
\usepackage{graphicx}
\usepackage{balance}  
\usepackage{algorithm}
\usepackage[noend]{algpseudocode}
\usepackage{amsthm}
\usepackage[breakable]{tcolorbox}
\usepackage{subcaption}
\usepackage{comment}
\usepackage{xspace}
\theoremstyle{definition}
\newtheorem{definition}{Definition}[section]
\newtheorem{theorem}{Theorem}
\newtheorem{lemma}[theorem]{Lemma}
\newtheorem{exmp}{Example}[section]

  \usepackage{color}

\usepackage{xcolor}
\copyrightyear{2022}
\acmYear{2022}
\setcopyright{acmcopyright}\acmConference[CCS '22]{Proceedings of the 2022 ACM SIGSAC Conference on Computer and Communications Security}{November 7--11, 2022}{Los Angeles, CA, USA}
\acmBooktitle{Proceedings of the 2022 ACM SIGSAC Conference on Computer and Communications Security (CCS '22), November 7--11, 2022, Los Angeles, CA, USA}
\acmPrice{15.00}
\acmDOI{10.1145/3548606.3560610}
\acmISBN{978-1-4503-9450-5/22/11}
\begin{document}
\title{Strengthening Order Preserving Encryption with Differential Privacy}
\author{Amrita Roy Chowdhury}
\affiliation{UW Madison \country{}}
\author{Bolin Ding}
\affiliation{Alibaba \country{}}
\author{Somesh Jha}
\affiliation{UW Madison \country{}}
\author{Weiran Liu}
\affiliation{Alibaba \country{}}
\author{Jingren Zhou}
\affiliation{Alibaba \country{}}
\renewcommand{\shortauthors}{Roy Chowdhury et al.}


\begin{abstract}
Ciphertexts of an order-preserving encryption (\textsf{OPE}) scheme preserve the order of their corresponding plaintexts. 
However, \textsf{OPE}s are vulnerable to inference attacks that exploit this preserved order. 
Differential privacy (\textsf{DP}) has become the de-facto standard for data privacy. One of the most attractive properties of \textsf{DP} is that any  post-processing computation, such as inference attacks,  performed  on  the  noisy  output  of  a \textsf{DP} algorithm does not degrade its privacy guarantee. 
In this work, we propose a novel \textit{differentially private order preserving encryption scheme}, \textsf{OP}$\epsilon$. Under \textsf{OP}$\epsilon$, the leakage of order from the ciphertexts is differentially private. Consequently, in the least, \textsf{OP}$\epsilon$ ensures a formal guarantee (a relaxed \textsf{DP} guarantee) even in the face of inference attacks. To the best of our knowledge, this is the first work to combine \textsf{DP} with a \textsf{OPE}. \textsf{OP}$\epsilon$ is based on a novel \textit{differentially private 
order preserving encoding} scheme, \textsf{OP$\epsilon$c}, that can be of independent interest in the local \textsf{DP} setting.  
We demonstrate \textsf{OP}$\epsilon$'s utility in answering range queries via empirical evaluation on four real-world datasets. For instance, \textsf{OP}$\epsilon$ misses only around $4$ in every $10$K correct records on average for a dataset of size $\sim732$K with an  attribute  of  domain  size $\sim18$K and $\epsilon= 1$. 

\end{abstract}

\maketitle

\newcommand{\OPE}{\textsf{OPE}\xspace}
\newcommand{\OPEs}{\textsf{OPE}s\xspace}
\newcommand{\ldp}{\textsf{LDP}\xspace}
\newcommand{\CDP}{\textsf{CDP}\xspace}
\newcommand{\DP}{\textsf{DP}\xspace}
\newcommand{\ddp}{\textsf{dLDP}\xspace}
\newcommand{\cdp}{\textsf{dDP}\xspace}
\newcommand{\name}{\textsf{OP$\epsilon$c}\xspace}
\newcommand{\arc}[1]{{\color{red} \emph{[[ARC: #1]]}}}
\newcommand{\nam}{\textsf{OP$\epsilon$}\xspace}
\newcommand{\indfaocpa}{\textsf{IND-FA-OCPA}~}
\newcommand{\indfaocpae}{$\epsilon$-\textsf{IND-FA-OCPA}}
\newcommand{\indfaocpaed}{$(\epsilon,\delta)-$\textsf{IND-FA-OCPA}}

\newcommand{\indocpa}{\textsf{IND-OCPA}~}
\newcommand{\Ideal}{$\textsf{IDEAL}_{\textsf{IND-FA-OCPA}}$~}
\newcommand{\advppt}{\mathcal{A}_{\textsf{PPT}}}
\newcommand{\adviit}{\mathcal{A}_{\textsf{IT}}}
\newcommand{\K}{\textsf{K}}
\newcommand{\Pa}{\mathcal{P}}
\newcommand{\E}{\textsf{E}}
\newcommand{\D}{\textsf{D}}
\newcommand{\s}{\textsf{S}}
\newcommand{\squishlist}{
	\begin{list}{$\bullet$}
		{
			\setlength{\itemsep}{0pt}
			\setlength{\parsep}{1pt}
			\setlength{\topsep}{1pt}
			\setlength{\partopsep}{0pt}
			\setlength{\leftmargin}{0.8em}
			\setlength{\labelwidth}{0.5em}
			\setlength{\labelsep}{0.3em} } }
	
\newcommand{\squishend}{
\end{list}  }\newcommand{\squishlistnum}{
	\begin{enumerate}[noitemsep]
		{
			\setlength{\itemsep}{0pt}
			\setlength{\parsep}{1pt}
			\setlength{\topsep}{1pt}
			\setlength{\partopsep}{0pt}
			\setlength{\leftmargin}{0.8em}
			\setlength{\labelwidth}{0.5em}
			\setlength{\labelsep}{0.3em} } }
	
\newcommand{\squishendnum}{
\end{enumerate}  }
\vspace{-0.3cm}
\section{Introduction}\vspace{-0.1cm}
\label{sec:intro}
Frequent mass data breaches \cite{databreach3, databreach5, databreach1, databreach2,  databreach4,  databreach6} of sensitive information have exposed the privacy vulnerability of data storage in practice. This has lead to a rapid development of systems that aim to protect the data while enabling statistical analysis on the dataset, both in academia \cite{academy1,academy2, academy4, academy3, CryptDB} and industry \cite{commercial2, commercial5,commercial4, commercial1, commercial3}. Encrypted database systems that allow query computation over the encrypted data is a popular approach in this regard.  Typically, such systems rely on \textit{property-preserving encryption} 
schemes \cite{DTE,Boldyreva} to enable efficient computation. \textit{Order-preserving encryption} (\OPE) \cite{Db3,Boldyreva,IND-FAOCPA,frequencyHiding2,IND-OCPA} is one such cryptographic primitive that preserves the numerical order of the plaintexts even after encryption. This allows actions like sorting, ranking,  and answering range queries to be performed directly over the encrypted data \cite{Db3,Db5,Db6,Db7,Db1,Db2,Db4}. 
\par 
However, encrypted databases are vulnerable to inference attacks \cite{Tao,Durak,Grubbs2,Grubbs3,Grubbs1,Paterson2,Paterson1,Naveed1,strongattack,strongattack2,strongattack3} that can reveal the plaintexts with good accuracy. Most of these attacks are inherent  to any property-preserving encryption scheme -- they do not leverage any weakness in the cryptographic security guarantee of the schemes but rather exploit just the \textit{preserved property}. For example, the strongest cryptographic guarantee for \OPEs (\textsf{IND-FA-OCPA}, see Sec. \ref{sec:background:OPE}) informally states that \textit{only} the order of the plaintexts will be revealed from the ciphertexts. However, inference attacks \cite{Grubbs2,Grubbs3,Grubbs1} can be carried out by leveraging only this ordering information. The basic principle of these attacks is to use  auxiliary
information to estimate the plaintext distribution and then correlate it with the ciphertexts based on the preserved property \cite{sok}. 
\par Differential privacy (\DP) has emerged as the de-facto standard for data privacy with widespread adoption in practice \cite{Census1, Census2,Vilhuber17Proceedings, Microsoft, Apple, Rappor1,Rappor2, Uber}.  \DP is an information theoretic guarantee that provides a rigorous guarantee of privacy for individuals in a dataset regardless of an adversary's auxiliary knowledge \cite{sok:DP}.  It is characterized by a  parameter $\epsilon>0$ where lower the value of $\epsilon$, greater the privacy guarantee achieved. An additional appealing property of \DP is that any post-processing computation, such as inference attacks, performed on the noisy output of a \DP algorithm does not incur additional privacy loss.  
\par In this work, we 
ask the following question: \vspace{-0.12cm}
\begin{displayquote}\emph{Is it possible to leverage the properties of \DP for providing a formal security guarantee for \OPEs even in the face of inference attacks?\vspace{-0.12cm}}\end{displayquote}
To this end, we propose a novel \textit{differentially private 
order preserving encryption} scheme, \nam. Recall that standard \OPE schemes are designed to reveal nothing but the order of the plaintexts. Our proposed scheme, \nam, ensures that this leakage of order is differentially private. In other words, the cryptographic guarantee of \OPEs is strengthened with a layer of \DP guarantee (specifically, a relaxed definition of \DP as discussed in the following paragraph). As a result, even if the cryptographic security guarantee of standard \OPEs proves to be inadequate (in the face of inference attacks), the \DP guarantee would continue to hold true. 
Intuitively, the reason is that \DP is resilient to post-processing computations as discussed above. To the best of our knowledge, \textit{this is the first work to combine \DP with a property-preserving encryption scheme.
} 
\subsection{Brief Overview of Key Ideas}\label{sec:intro:brief}
 The standard \DP guarantee requires any two pairs of input data to be indistinguishable from each other (see Sec. \ref{sec:background:DP}) and  is generally catered towards answering statistical queries over the entire dataset. However, in our setting we require the output of the \DP mechanism to retain some of the ordinal characteristics of its input -- the standard \DP guarantee is not directly applicable to this case. Hence, we opt for a natural relaxation of \DP -- 
 only pairs of data points that are ``close" to each other should be indistinguishable. Specifically, the privacy guarantee is heterogeneous and degrades linearly with the $\ell_1$-distance between a pair of data points.   
It is denoted by $\epsilon$-\ddp (or $\epsilon$-\cdp in the central model of \DP; see Sec. \ref{sec:background:DP}). This relaxation is along the lines of d$_\chi$-privacy \cite{dx} and is amenable to many practical settings. For instance, consider a dataset of annual sale figures of clothing firms. The information whether a firm is a top selling or a mid-range one is less sensitive than its actual sales figures. Similarly, for an age dataset, whether a person is young or middle-aged is less sensitive than their actual age.  
\par \DP guarantee inherently requires randomization -- this entails an inevitable loss of utility, i.e., some pairs of output might not preserve the correct order of their respective inputs. In order to reduce the instances of such pairs, \nam offers the flexibility of preserving only a partial order of the plaintexts. Specifically, a (user specified) partition is defined on the input domain and the preserved order is expected at the granularity of this partition. The output domain is defined by a numeric encoding over the intervals of the partition and  all the elements belonging to the same interval are mapped to the corresponding encoding for the interval (with high probability).  Due to the linear dependence of the \DP guarantee (and consequently, the ratio of output probabilities) on the distance between the pair of inputs, lower is  the number of intervals in the partition, higher is the probability of outputting the correct encoding in general (see Sec. \ref{sec:construction} and Sec. \ref{sec:evaluation:results}). \nam preserves the order over this encoding. The reason why this results in better utility for encrypted databases is illustrated by the following example. The typical usecase for \OPE encrypted databases is retrieving a set of records from the outsourced database that belong to a queried range. Suppose a querier asks for a range query $[a,b]$ and let $\mathcal{P}$ be a partition that covers the range with $k$ intervals $\{[s_1,e_1],\cdots,[s_k,e_{k}]\}$ such that $s_1< a < e_1$ and $ s_k < b < e_k $. A database system encrypted with \nam and instantiated with the partition $\Pa$ will return all the records that are noisily mapped to the range $[s_1,e_k]$ (since the order is preserved at the granularity of $\Pa$). Thus, the querier has to pay a processing overhead of fetching extra records, i.e., the records belonging to the ranges $\{[s_1,a-1],[b+1,e_k]\}$. However, if $k< b-a$, then with high probability it would receive all the correct records in $[a,b]$ which can be decrypted and verified (Sec. \ref{sec:application}). 
To this end, we first propose a new primitive, \name, that enables \textit{
order preserving encoding} with $\epsilon$-\ddp. The encryption scheme, \nam, is then constructed using the $\name$ primitive and a \OPE (Sec. \ref{sec:encryption}).

Our work is along the lines of a growing area of research exploring the association between \DP and cryptography \cite{CryptE,capc,kairouz2021distributed, Shrinkwrap} (see Sec. \ref{sec:relatedwork}). Beyond \OPEs, the \name primitive can be used as a building block for other secure computation that require ordering, such as order-revealing encryptions (see App. \ref{app:discussion}). Additionally, \name can be of independent interest for the \ldp setting in answering a variety of queries, such as ordinal queries, frequency and mean estimation (see Sec. \ref{sec:opec:ldp}).
\vspace{-0.4cm}
\subsection{Discussion}
In this section, we answer some key questions pertinent to our work that the readers might have. \vspace{-0.2cm}\\\\
\textbf{Q1.} \textit{Why should we care about \OPEs?}\\\textbf{A.}  Range query constitutes an extremely important class of queries for data analytics. For instance, about half of the queries of the TPC-H benchmark \cite{TPCH}, which is the standard benchmark for OLAP queries, have range predicates \cite{TPCH20}. Additionally, range query is a fundamental operation in DBMS with its applications ranging from B+tree indices \cite{B+Tree,Btree} to band-joins \cite{bandjoins}.  Thus, efficient support for range queries on encrypted databases is a fundamental task for secure data analytics. One of the main challenges of practical deployment of cryptographic protocols is the associated performance overhead especially with large realistic datasets.  The advantage of \OPEs in this regard is that it allows range queries to be performed directly over the encrypted data thereby matching the \textit{optimal} performance of plaintext computation. Hence, given this immense performance advantage, \OPEs are a key building block for encrypted databases~\cite{sok} and exploring secure implementations of \OPEs is still an important problem.\vspace{-0.2cm}
\\\\
\textbf{Q2.} \textit{What is the advantage of a \nam scheme over just  \name primitive or a \OPE scheme?}\\
\textbf{A.} \nam satisfies a new security guarantee, \indfaocpae, (see Sec. \ref{sec:security}) that enhances the cryptographic guarantee of a \OPE scheme (\textsf{IND-FA-OCPA}) with a layer of $\epsilon$-\cdp guarantee. As a result, \nam enjoys \textit{strictly stronger security} than both \name primitive ($\epsilon$-\cdp) and \OPE  (\textsf{IND-FA-OCPA}). 
\vspace{-0.2cm}\\\\
\textbf{Q3.} {\textit{What are the security implications of \nam  in the face of inference attacks?}\\
\textbf{A.} 
In the very least,  \nam  \textit{rigorously limits the accuracy of inference attacks for every record for all adversaries} (Thm. \ref{thm:attack:record}, Sec. \ref{sec:attack}). In other words, \nam guarantees that none of the attacks can infer the value of any record beyond a certain accuracy that is allowed by the \ddp guarantee. For instance, for an age dataset and an adversary with real-world auxiliary knowledge, no inference attack in the snapshot model can distinguish between two age values $(x,x')$ such that $|x-x'|\leq 8$ for $\epsilon=0.1$  (Sec. \ref{sec:evaluation:results}).
\vspace{-0.2cm}\\\\
\textbf{Q4.} \textit{How is \nam's utility (accuracy of range queries)?}\\
\textbf{A.} We present a construction for the $\name$ primitive (and subsequently, \nam) and our experimental results on four real-world datasets demonstrate its practicality for real-world usage (Sec. \ref{sec:evaluation}). Specifically, \nam  misses only $4$ in every $10$K correct records on average for a dataset of size $\sim 732$K with an attribute of domain size $18$K and $\epsilon=1$. The overhead of processing extra  records  is also low -- the average number of extra records returned is just $0.3\%$ of the total dataset size.
\vspace{-0.2cm}\\\\
\textbf{Q5.} \textit{When to use \nam?}\\
\textbf{A.} As discussed above, \nam gives a strictly stronger guarantee than any \OPE scheme (even in the face of inference attacks) with almost no extra performance overhead (Sec. \ref{sec:application}). Additionally, it is backward compatible with any encrypted database that is already using a \OPE scheme (satisfying \textsf{IND-FA-OCPA}, see Sec. \ref{sec:application}). Hence, \nam could be used for secure data analytics in settings where (1) the $\epsilon$-\cdp guarantee is acceptable, i.e., the main security concern is preventing the distinction between input values close to each other (such as the examples discussed above) and $(2)$ the application can tolerate a small loss in utility. 
Specifically in such settings, replacing encrypted databases with \nam  would give a strictly stronger security guarantee against all attacks with \textit{nominal change in infrastructure or performance} -- a win-win situation.
\vspace{-0.25cm}
\section{Background}\label{sec:background}
\subsection{Differential Privacy}\label{sec:background:DP}
Differential privacy is a quantifiable measure of the stability of the output of a randomized mechanism to changes to its input. There are two popular models of differential privacy, \textit{local} and \textit{central}. The local model consists of a set of individual data owners and an untrusted data aggregator; each individual perturbs their data using a (local) \DP algorithm and sends it to the aggregator which uses these noisy data to infer some statistics on the entire dataset. Thus, the \ldp model allows gleaning of useful information from the dataset without requiring the data owners to trust any third-party entity. 
The \ldp guarantee is formally defined as follows:\vspace{-0.1cm}
\begin{definition}[Local Differential Privacy, \ldp]
 A randomized algorithm 
 $\mathcal{M} : \mathcal{X} \rightarrow \mathcal{Y}$ is $\epsilon$-\ldp if for any pair of private values $x, x' \in \mathcal{X}$ and any subset of output, $\mathcal{T} \subseteq \mathcal{Y}$ \begin{gather}
 \mathrm{Pr}\big[\mathcal{M}(x) \in \mathcal{T}\big] \leq e^{\epsilon} \cdot \mathrm{Pr}\big[\mathcal{M}(x') \in \mathcal{T}  \big]\end{gather}
 \end{definition}
 $\epsilon$-\ldp guarantees the same level of protection for all pairs of
private values.  However, as discussed in the preceding section,  in this paper we use an
extension of \ldp which uses the $\ell_1$-distance  between a pair of values to customize heterogeneous (different levels of)
privacy guarantees for different pairs of private values.
 \begin{definition}
 [Distance-based Local Differential Privacy, \ddp] A randomized algorithm $\mathcal{M} : \mathcal{X} \rightarrow \mathcal{Y}$ is $\epsilon$-distance based locally differentially private (or $\epsilon$-\ddp), if for any pair of private values $x, x' \in \mathcal{X}$ and any subset of output $\mathcal{T} \subseteq \mathcal{Y}$, 
 \begin{gather*} 
 \mathrm{Pr}[\mathcal{M}(x)\in \mathcal{T}] \leq e^{\epsilon|x-x'|}\cdot \mathrm{Pr}[\mathcal{M}(x')\in \mathcal{T}]\numberthis\end{gather*}\vspace{-0.1cm}
 \label{def:ddp}\vspace{-0.1cm}\end{definition}\vspace{-0.4cm}
 The above definition is equivalent to the notion of metric-based  \ldp \cite{metricDP,dx} where the metric used is $\ell_1$-distance. 

In the central differential privacy (\CDP) model, a trusted data curator collates data from all the individuals  and stores it in the clear in a centrally held dataset. The curator mediates upon every query posed by a mistrustful analyst and enforces privacy by adding noise to the answers of the analyst's queries before releasing them.
\begin{definition}[Central Differential Privacy, \CDP]
  A randomized algorithm $\mathcal{M}: \mathcal{X}^n \mapsto \mathcal{Y}$ satisfies $\epsilon$-differential privacy ($\epsilon$-\DP) if for all $\mathcal{T} \subseteq \mathcal{Y}$ and for all adjacent datasets $X, X' \in \mathcal{X}^n$ it holds that\begin{gather}
\mathrm{Pr}[\mathcal{M}(X)\in \mathcal{T}] \leq e^{\epsilon}\cdot \mathrm{Pr}[\mathcal{M}(X') \in \mathcal{T}]
\end{gather}
\end{definition}\vspace{-0.2cm}
The notion of adjacent inputs is application-dependent, and typically means that $X$ and $X'$
differ in a single element (corresponding to a single individual). 
Particularly in our setting, the equivalent definition of the distance based relaxation of differential privacy in the \CDP model is given as follows:
 \begin{definition}[Distance-based Central Differential Privacy, \cdp] A randomized algorithm $\mathcal{M} : \mathcal{X}^n \rightarrow \mathcal{Y}$ is $\epsilon$-distance based centrally differentially private (or $\epsilon$-\cdp), if for any pair of datasets $X$ and $X'$ such that they differ in a single element, $x_i$ and $x_i'$, 
 and any subset of output $\mathcal{T} \subseteq \mathcal{Y}$,  
 \begin{gather}
 \mathrm{Pr}\big[\mathcal{M}(X)\in \mathcal{T}\big] \leq e^{\epsilon|x_i-x_i'|}\cdot \mathrm{Pr}\big[\mathcal{M}(X')\in \mathcal{T}\big]\end{gather}\label{def:cdp} \vspace{-0.3cm}
 \end{definition}
 We define $X $ and $X'$, as described above, to be \textbf{$\mathbf{t}$-adjacent} where $t \geq |x_i-x_i'|$, i.e., the differing elements differ by \textit{at most}  $t$. Trivially, any pair of $t$-adjacent datasets is also $t'$-adjacent for $t'>t$.  
\\Next, we formalize the resilience of \ddp (and \cdp) to post-processing computations. 
\begin{theorem}[Post-Processing~\cite{Dwork}]\label{thm:post}
Let $\mathcal{M}: \mathcal{X} \mapsto \mathcal{Y}$ ($\mathcal{M}: \mathcal{X}^n \mapsto \mathcal{Y}$) be a $\epsilon$-\ddp (\cdp) algorithm. Let $g : \mathcal{Y}\mapsto 
\mathcal{Y}'$ be any randomized mapping. Then $g \circ \mathcal{M}$ is also $\epsilon$-\ddp (\cdp).  \end{theorem}

\subsection{Order Preserving Encryption}\label{sec:background:OPE}
 In this section, we discuss the necessary definitions for \OPEs. 
\begin{definition}[Order Preserving Encryption \cite{frequencyHiding2}\footnote{See App.\ref{app:background} for additional notes.}] An order preserving encryption (\OPE) scheme $\mathcal{E}=\langle\K,\E,\D\rangle$ is a tuple of probabilistic polynomial time (\textsf{PPT}) algorithms:
\squishlist\item \textit{Key Generation} (\K).  The key generation algorithm takes as input a security parameter $\kappa$ and outputs
a secret key (or state) $\s$ as $\s\leftarrow \K(1^\kappa)$.   \item \textit{Encryption} (\E). Let  $X=\langle x_1,\cdots,x_n\rangle$ be an input dataset.  The encryption algorithm  takes as input a secret key $\s$, a
plaintext $x \in X$, and an order $\Gamma$ (any permutation of $\{1,\cdots,n\}$). It outputs a new key $\s'$ and a ciphertext $y$ as $(\s',y)\leftarrow \E(\s,x,\Gamma)$. \item \textit{Decryption} (\D). Decryption recovers the plaintext $x$ from the ciphertext $y$ using the secret key $\s$, $x\leftarrow \D(\s,y)$ . \squishend Additionally, we have \squishlist   \item \textit{Correctness Property.}  $x\leftarrow \D\big(\E(\s,x,\Gamma)\big),  \thinspace \forall \s, \forall x, \forall \Gamma$ 
\item 
  \textit{Order Preserving Property.} $\thinspace x > x' \implies y > y', \thinspace \forall x,x'$ where $y~(y')$ is the ciphertext corresponding to the plaintext $x~(x')$
  \squishend\label{def:OPE}
\end{definition}
The role of $\Gamma$ in the above definition is discussed later in this section.
The strongest formal guarantee for a \OPE scheme is \textit{indistinguishability against frequency-analyzing ordered chosen plaintext attacks} (\textsf{IND-FA-OCPA}). We present two definitions in connection to this starting with the notion of \textit{randomized orders} as defined by Kerschbaum \cite{IND-FAOCPA}.
\begin{definition}(Randomized Order~\cite{IND-FAOCPA}) 
Let $X=\langle x_1,\cdots,x_n\rangle$ be a dataset. An order
 $\Gamma = \langle\gamma_1, \cdots, \gamma_n\rangle$, where $ \gamma_i \in [n]$ and $i \neq j \implies \gamma_i \neq \gamma_j $,
for all $i, j,$ of dataset $X$, is defined to be a \textit{randomized order} if it holds that \vspace{-0.2cm}
\begin{gather*}\vspace{-0.4cm}
\forall i, j~ (x_i > x_j \implies \gamma_i > \gamma_j) \wedge (\gamma_i > \gamma_j \implies x_i \geq  x_j ) \vspace{-0.2cm}
\end{gather*}
\label{def:RO}\vspace{-0.4cm}
\end{definition}
For a plaintext dataset $X$ of size $n$, a randomized order, $\Gamma$, is a permutation of the plaintext indices $\{1, \cdots , n\}$ such that its inverse, $\Gamma^{-1}$, gives a sorted version of $X$. This is best explained by an example -- let  $X=\langle 9,40,15,76,15,76\rangle$ be a dataset of size $6$. A randomized order for $X$ can be either of 
 $\Gamma_1=\langle1,4,2,5,3,6\rangle$, $\Gamma_2=\langle 1,4,3,5,2,6\rangle$, $\Gamma_3=\langle 1,4,2,6,3,5\rangle$  and $\Gamma_4=\langle 1,4,3,6,2,5\rangle$. This is because the order of the two instances of $76$ and $15$ does not matter for a sorted version of $X$. 
\begin{definition}[\indfaocpa~\cite{frequencyHiding2,IND-FAOCPA}] An order-preserving encryption scheme $\mathcal{E} = (\K, \E, \D)$ has \textit{indistinguishable ciphertexts under frequency-analyzing ordered chosen plaintext attacks} if for any \textsf{PPT} adversary $\advppt$:
\begin{gather}\small \Big|\mathrm{Pr}[\mathcal{G}^{\advppt}_\textsf{FA-OCPA}(\kappa, 1) = 1] - \mathrm{Pr}[\mathcal{G}^{\advppt}_\textsf{FA-OCPA} (\kappa, 0) = 1]|\leq \textsf{negl}(\kappa)\label{eq:indocpa}\end{gather}
where $\kappa$ is a security parameter, \textsf{negl}$(\cdot)$ denotes a negligible function and $\mathcal{G}^{\advppt}_\textsf{FA-OCPA}(\kappa, b)$ is the random variable denoting $\advppt$'s output for the following game:\\\\
\textbf{Game} $\mathcal{G}^{\advppt}_\textsf{FA-OCPA} (\kappa, b)$
\squishlistnum
\item $(X_0,X_1)\leftarrow \advppt $ where $|X_0|=|X_1|=n$ and $X_0$ and $X_1$
have at least one common randomized order
\item Select $\Gamma^*$ uniformly at random from the common randomized orders of $X_0,X_1$
\item $ S_0 \leftarrow \K(1^\kappa)$
\item For  $\forall i \in [n]$, run $(S_i,y_{b,i}) \leftarrow \E(S_{i-1},x_{b,i},\Gamma^*)$
\item $b' \leftarrow \advppt(y_{b,1},\cdots,y_{b,n})$ where $b'$ is $\advppt$'s guess for $b$
\squishendnum
$\advppt$ is said to win the above game iff $b=b'$.
\label{def:faocpa}
\end{definition}

Informally, this guarantee implies that nothing other than the order of the plaintexts, not even the frequency, is revealed from the ciphertexts. Stated otherwise, the ciphertexts
only leak a randomized order of the plaintexts (randomized orders do not contain any frequency
information since each value always occurs exactly once) which is determined by the input order $\Gamma$ in Defn. \ref{def:OPE}. In fact, if $\Gamma$ itself happens to be a randomized order of the input $X$ then, the randomized order leaked by the corresponding ciphertexts is guaranteed to be $\Gamma$. For example, for $X=\langle9,40,15,76,15,76\rangle$ and $\Gamma=\langle 1,4,2,5,3,6\rangle$, we have $y_1<y_3<y_5<y_2<y_4<y_6$ ($y_i$ denotes the corresponding ciphertext for $x_i$ and $\Gamma^{-1}=\langle 1, 3, 5, 2, 4, 6\rangle$). 
Thus, the $\indfaocpa$ guarantee ensures that two
datasets with a common randomized order – but different
plaintext frequencies – are indistinguishable. For example, in the aforementioned game $\mathcal{G}^{\advppt}_\textsf{FA-OCPA} (\cdot)$, $\advppt$ would fail to distinguish between the plaintext datasets $X_0=\langle 9,40,15,76,15,76\rangle$ and  $X_1=\langle 22,94,23,94,36,94\rangle$  both of which share the randomized order $\Gamma^*=\langle 1,4,2,5,3,6\rangle$.

\section{$\epsilon$-\texorpdfstring{\MakeLowercase{d}}{d}LDP 
Order Preserving Encoding (\textsf{OP}$\epsilon$\texorpdfstring{\MakeLowercase{c}}{c})}
\label{sec:encoding}
In this section, we discuss our proposed primitive -- $\epsilon$-\ddp \textit{
order preserving encoding}, \name. 
\vspace{-0.3cm}\\\\
\textbf{Notations.} 
$[n], n \in \mathbb{N}$ denotes the set $\{1,2,\cdots,n-1,n\}$. If $\mathcal{X}=[s,e]$ is an input domain, then a $k$-partition $\mathcal{P}$ on $\mathcal{X}$ denotes a set of $k$ non-overlapping intervals $\mathcal{X}_i=(s_i,e_i]$
\footnote{The first interval,  $\mathcal{X}_1=[s_1,e_1]$, is a closed interval.}
, $s_{j+1}=e_j, i \in [k], j \in [k-1]$ such that $\bigcup_{i=1}^k \mathcal{X}_i=\mathcal{X}$. For example, for $\mathcal{X}=[1,100]$, $\Pa=\{[1,10],$ $(10,20],\cdots,(90,100]\}$ denotes a $10$-partition. Let $\hat{\mathcal{X}}$ denote the domain of partitions defined over $\mathcal{X}$. 
Additionally, let $\mathcal{O}=\{o_1, \cdots, o_k\}, o_i<o_{i+1}, i \in [k-1]$ represent the output domain where $o_i$ is the corresponding encoding for the interval $\mathcal{X}_i$ and let $\mathcal{P}(x)=o_i$ denote that $x \in \mathcal{X}_i$. Referring back to our example, if $\mathcal{O}=\{1,2,\cdots,10\}$, then  $\Pa(45)=5$. 

\subsection{Definition of \name}\label{sec:OPec:def}
\name is a randomised mechanism that encodes its input while maintaining some of its ordinality. 
\vspace{-0.1cm}
\begin{definition}[$\epsilon$-\ddp 
Order Preserving Encoding, \name] For a given $k$-partition \scalebox{0.9}{$\Pa\in \hat{\mathcal{X}}$}, a \textit{$\epsilon$-\ddp 
order preserving encoding} scheme, \scalebox{0.9}{$\name:\mathcal{X}\times \hat{\mathcal{X}}\times \mathbb{R}_{>0}\mapsto \mathcal{O}$} is a randomized mechanism such that 
\squishlistnum \vspace{-0.05cm}
\item $k=|\mathcal{O}|, k\leq |\mathcal{X}|$ \vspace{-0.1cm}\item For all $ x \in \mathcal{X}$ and $o'\in\mathcal{O}\setminus\mathcal{T}_x$ where \scalebox{0.9}{$\mathcal{T}_x = \begin{cases}\mbox{\scalebox{0.9}{$\{o_1,o_2\}$}} & \hspace{-0.5cm}\mbox{if } \Pa(x)=o_1 \\ \mbox{\scalebox{0.9}{$\{o_{k-1},o_k\}$}} & \hspace{-0.5cm}\mbox{if } \Pa(x)=o_k\\\mbox{\scalebox{0.9}{$\{o_{i-1},o_i,o_{i+1}\}$}} & \hspace{-0.2cm}\mbox{otherwise}  \end{cases}$} \\\vspace{-0.1cm}
$\exists o \in \mathcal{T}_x$ such that, 
\begin{gather*}\vspace{-0.1cm}\mathrm{Pr}\big[\name(x,\Pa,\epsilon) = o\big]> \mathrm{Pr}\big[\name(x,\Pa,\epsilon)= o'\big]\numberthis\label{eq:popec:order}\vspace{-0.4cm}\end{gather*}  \vspace{-0.4cm}
\item For all $x,x' \in \mathcal{X}, o \in \mathcal{O}$, 
we have 
\vspace{-0.1cm}
\begin{gather*}
\hspace{-0.4cm}\mathrm{Pr}\big[\name(x,\Pa,\epsilon)=o\big]\leq e^{\epsilon|x-x'|}\cdot \mathrm{Pr}\big[\name(x',\Pa,\epsilon)=o\big]\label{eq:popec:dp}\vspace{-0.4cm}\end{gather*} \squishendnum
\label{def:popec}
\end{definition} \vspace{-0.2cm} 
The first property in the above definition signifies the flexibility of the \name primitive to provide only a \textit{partial} ordering guarantee. For instance, in our above example $k=10 < |\mathcal{X}|=100$. Thus, $\Pa$ acts as a utility parameter -- it determines the granularity at which the ordering information is maintained by the encoding (this is independent of the privacy-accuracy trade-off arising from the choice of $\epsilon$). For example, for the same value of $\epsilon$ and $\mathcal{X}=[1,100]$, $\Pa=\{[1,10],(10,20],\cdots,(90,100]\}$ gives better utility than $\Pa'=\{[1,33],(33,66],(66,100]\}$ since the former preserves the ordering information at a finer granularity.  $\Pa=\mathcal{O}=\mathcal{X}$ denotes the default case where effectively no partition is defined on the input domain and  $\Pa(x)=x, x \in \mathcal{X}$ trivially. We discuss the significance of the parameter $\Pa$ in Sec. \ref{sec:application}.
 \par Due to randomization (required for the \ddp guarantee), \name is bound to incur some errors in the resulting numerical ordering of its outputs. To this end, the second property guarantees that the noisy output is most likely to be either the correct one or the ones immediately next to it. For instance, for the aforementioned example, $\name(45,\Pa,\epsilon)$ is most likely to fall in $\{4,5,6\}$. This ensures that the noisy outputs still retain sufficient ordinal characteristics of the corresponding inputs. Note that the actual value of the encodings in $\mathcal{O}$ does not matter at all as long as the ordinal constraint $o_i<o_{i+1}, i \in [k-1]$ is maintained. For instance for $\Pa=\{[1,10],(10,20],\cdots,(90,100]\}$, $\mathcal{O}=\{1, 2,3,4,5,6,7$ $,8,9,10\}$, $\mathcal{O}'=\{5, 15, 25, 35, 45, 55, 65, 75, 85, 95\}$ and $\mathcal{O}''=\{81,99,120,150,234,345,$  $400,432,536,637\}$ are all valid.  \par 
Finally, the third property ensures that the primitive satisfies $\epsilon$-\ddp. Note that \scalebox{0.9}{$\epsilon=\infty$} represents the trivial case \scalebox{0.9}{$\name(X,\Pa,\infty)=\Pa(X)$}.
\vspace{-0.4cm}
\subsection{Construction of \name}\label{sec:construction}
In this section, we describe a construction for the \name primitive (Alg. \ref{alg:1}).
\setlength{\textfloatsep}{2pt}
\begin{algorithm}
\small \caption{Construction of \name}\label{alg:1}
\begin{algorithmic}[1]
\Statex   \textbf{Setup Parameters:} $\mathcal{D}$ - Prior input distribution over $\mathcal{X}$, its 
\Statex \hspace{2.8cm}  default value is the uniform distribution;
\Statex  \hspace{2.4cm}$\mathcal{O}$ - Output domain $\{o_1,\cdots,o_k\}$;

\Statex   \textbf{Input:} $x$ - Number to be encoded via \name; $\epsilon$ - Privacy budget;
\Statex \hspace{0.8cm}$\Pa$ - A $k$-partition $\{[s_1,e_1] ,\cdots ,( s_k,e_k]\}$ over $\mathcal{X}$ 
\Statex \textbf{Output:} $o$ - Output encoding;
\Statex \textbf{Stage I:} Computation of central tendency for each interval
\State \textbf{for} $i \in [k]$ 
\State \hfill $d_i = $ Weighted median of the interval $(s_i,e_i]$ where  $\mathcal{D}$ gives \Statex \hspace{1.6cm} the  corresponding weights
\State \textbf{end for} 
\Statex\textbf{Stage II:} Computation of the output probability distributions
\State \textbf{for} $\mathbf{x} \in \mathcal{X}$:
\State \hspace{0.3cm}\textbf{for} $i \in [k]$ 
\State
\begin{equation}
\small\hspace{0.5cm}
p_{\mathbf{x},i} = \frac{e^{-|\mathbf{x}-d_i|\cdot \epsilon/2}}{\overset{k}{\underset{j=1}{\sum}}
e^{-|\mathbf{x}-d_j|\cdot \epsilon/2}} \hspace{0.3cm}\mbox{\textcolor{blue}{$\rhd$} \scalebox{0.9}{$p_{\mathbf{x},i}=\mathrm{Pr}\big[\name(\mathbf{x},\Pa,\epsilon)=o_i\big]$}} \label{eq:pivot1}
\end{equation}
\State \hspace{0.3cm}\textbf{end for}
\State \hspace{0.3cm} $p_\mathbf{x}=\{p_{\mathbf{x},1},\cdots, p_{\mathbf{x},k}\}$ \hfill \textcolor{blue}{$\rhd$} Encoding (output)  probability \Statex \hfill distribution for $\mathbf{x}$ 
\State \textbf{end for}
\State $o\sim p_x$ \textcolor{blue}{$\rhd$} Encoding drawn at random from the distribution $p_{x}$
\State \textbf{Return} $o $ 
\end{algorithmic}
\end{algorithm}
  The algorithm is divided into two stages.  
  In Stage I (Steps $1$-$3$),  it computes the central tendency (a typical value for a distribution) ~\cite{central:tendency}, $d_i, i \in [k]$,  of each of the intervals of the given $k$-partition $\Pa$\footnote{$\Pa$ \textit{cannot} be computed from the private dataset (without apportioning a separate privacy budget). Here, we assume that $\Pa$ is chosen from the (non-private) prior, $\mathcal{D}$.}. Specifically, we use weighted median \cite{wm} as our measure for the central tendency where the weights are determined by a prior on the input data distribution, $\mathcal{D}$.  This maximizes the expected number of inputs that are mapped to the correct encoding, i.e., $x$ is mapped to $\Pa(x)$. 
  $\mathcal{D}$ can be estimated from domain knowledge or (non-private) auxiliary datasets. In the event such a prior is not available, $\mathcal{D}$ is assumed to be the uniform distribution ($d_i$ is the median).   

\par In Stage II (Steps $4$-$9$), the encoding probability distributions are computed such that the probability of $x$ outputting the $i$-th encoding, $o_i$, is inversely proportional to its distance from the $i$-th central tendency, $d_i$. Specifically, we use a variant of the classic exponential mechanism \cite{Ordinal1,Dwork} (Eq. \eqref{eq:pivot1}). An illustration of the algorithm is in App. \ref{app:alg:1}.
\begin{theorem}Alg. \ref{alg:1} gives a construction for \name (Def. \ref{def:popec}). \label{thm:opec}\end{theorem}\vspace{-0.2cm}
The proof of the above theorem follows directly from two facts. First, Alg. \ref{alg:1} satisfies the ordinal constraint of Eq. \ref{eq:popec:order} (Lemma \ref{lemma:1} in App. \ref{app:thm:construction}) as depicted in Fig. \ref{fig:partition}.  Second, it is straightforward from Eq. \ref{eq:pivot1} that Alg. \ref{alg:1} satisfies $\epsilon$-\ddp (Lemma \ref{lemma:2} in App. \ref{app:thm:construction}). 
\\\textbf{Size of partition $|\Pa|$.} From Eq. \ref{eq:pivot1}, we observe that for every input $x$, the encoding probability distribution $p_x$ is an exponential distribution centered at $\Pa(x)$ -- its correct encoding. Moreover, the smaller is the size of $\Pa$ (number of intervals in $\Pa$),  the larger is the probability of outputting $\Pa(x)$ (or its immediate neighbors). This is demonstrated in Fig. \ref{fig:partition} which plots $p_x$ for \scalebox{0.9}{$x=50$} and $\epsilon=0.1$ under varying equi-length partitioning of the input domain $[100]$.
 \begin{figure}[ht]
 \centering
      \includegraphics[width=0.5\linewidth]{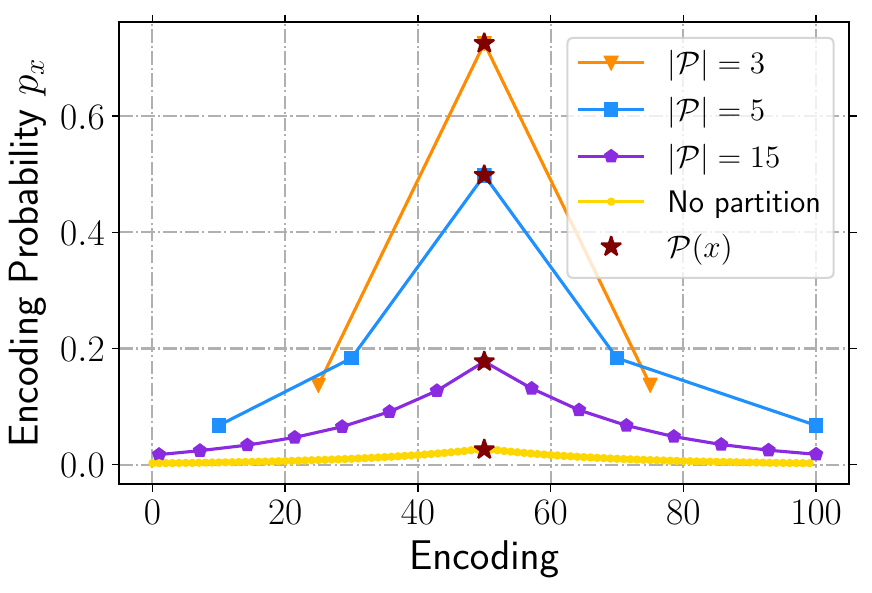}\vspace{-0.4cm}
        \caption{Encoding probability distribution for different partition sizes for $x=50$, $\epsilon=0.1$ and $\mathcal{X}=[100]$} \label{fig:partition}
    \end{figure}
\begin{tcolorbox}[sharp corners]\vspace{-0.2cm}\textbf{Remark 1.} The \ddp guarantee of \name (Thm. \ref{thm:opec}) does \textit{not} depend on the partition $\Pa$. Thus, the partition size\footnote{$k=1$ is a trivial case which destroys all ordinal information.} 
could range from $k=|\mathcal{O}|=|\mathcal{X}|$ (no effective partitioning at all) to $k=2$. Additionally, the \ddp guarantee (and utility) is also independent of the encoding domain, $\mathcal{O}$, as long as the appropriate ordering constraint is valid. Also note that the partition can be completely \textit{arbitrary}. \vspace{-0.2cm} \end{tcolorbox}
\noindent\textbf{Design Choices.} 
Note that we intend to use the \name primitive as a building block for our differentially private \nam scheme (details in Sec. \ref{sec:encryption}). Now for a standard \OPE, the utility (accuracy) remains exactly the same as that of the plaintext. However, the inherent randomization in \DP results in an inevitable loss in utility. Hence, our primary motivation for designing \name is to $(1)$ provide a meaningful guarantee against inference attacks $(2)$ with high utility. 
\vspace{-0.2cm}\\\\\noindent\textit{Why \ddp?}  The standard (local) \DP definition requires every input pair to be indistinguishable -- this requires the addition of a large amount of noise resulting in low utility, especially for large data domains. With \ddp, only input pairs that are close to each other are indistinguishable which still results in a meaningful guarantee in practice (Sec. \ref{sec:intro:brief}). Essentially, this heterogeneous guarantee reveals some controlled information about the $\ell_1$-distance between input pairs. Note that utility in our context implies how well is the order of the plaintexts preserved. Observe that the order of two values is determined by their $\ell_1$-distance. Thus intuitively, the output under \ddp retains some ordinal information about its input, thereby improving utility.
\vspace{-0.2cm}\\\\\noindent \textit{Why partitioning?} The typical approach to achieve \DP is via addition of noise that is proportional to a domain-dependent term called sensitivity ~\cite{Dwork}. Common approaches to  mitigate the cost of high sensitivity (for large data domains), such as propose-test-release~\cite{Dwork}, result in approximate DP. Partitioning bypasses the need for noise addition by design and provides a clean way to improve utility with pure DP (see Sec. \ref{sec:application} for details). Additionally, partitioning provides flexibility. A use case is demonstrated below where one can plug-and-play with different values of $\Pa$ enabling \name to answer different types of queries in the LDP setting. 

\vspace{-0.2cm}\subsection{\ldp Mechanisms using \name} \label{sec:opec:ldp}

The \name primitive can be of independent interest in the \ldp setting. 
Depending on the choice of the partition $\mathcal{P}$ over the input domain $\mathcal{X}$, \name can be used to answer different types of queries with high utility. 
In this section, we describe how to use \name to answer two such queries. \\\textbf{Problem Setting.} We assume the standard \ldp setting with $n$ data owners, $\textsf{DO}_i, i \in [n]$ each with a private data $x_i$. 
\vspace{-0.2cm}
\subsection*{Ordinal Queries}
\name can be used to answer queries in the \ldp setting that require the individual noisy outputs to retain some of the ordinal characteristics of their corresponding inputs. 
 One class of such queries include identifying which $q$-quantile 
 does each data point belong to. This constitutes a popular class of queries for domains such as annual employee salaries, annual sales figures of commercial firms and student test scores. For example, suppose the dataset consists of the annual sales figures of different clothing firms and the goal is to group them according to their respective deciles. 
Here, partition $\Pa$ is defined by dividing the input domain into $q=10$ equi-depth intervals using an estimate of the input distribution, $\mathcal{D}$. If such an estimate is not available, a part of the privacy budget can be first used to compute this directly from the data \cite{numestimate}. For another class of queries,  the partition can be defined directly on the input domain based on its semantics. Consider an example where the goal is to group a dataset of audiences of TV shows based on their age demographic -- the domain of age can be divided into intervals $\{[1,20], [21,40], [41,60], [61,100]\}$ based on categories like ``youth'', ``senior citizens''. Once the partition is defined, each data owner uses \name to report their noisy encoding. Note that the \ddp guarantee is amenable to these cases, as one would want to report the intervals correctly but the adversary should not be able to distinguish between values belonging to the same interval. 
\vspace{-0.2cm}
\subsection*{Frequency Estimation} \label{sec:freqest}
Here, we discuss the default case of the \name primitive where the partition is same as the input domain, i.e., $\Pa=\mathcal{O}=\mathcal{X}$. Under this assumption, we can obtain a frequency oracle in the \ldp setting under the \ddp guarantee. We describe the mechanism below (see Alg. \ref{alg:freqestimation} in \ifpaper the full paper~\cite{anom} \else App. \ref{app:freq:utility} \fi for full algorithm). Given a privacy parameter, $\epsilon$, each data owner, $\textsf{DO}_i, i \in [n]$, reports  $\tilde{o}_i=\name(x_i,\mathcal{X},\epsilon)$  to the untrusted data aggregator. Next, the data aggregator  performs non-negative least squares (\textsf{NNLS})  as a post-processing inferencing step on the noisy data to compute the final frequency estimations.  \textsf{NNLS} is a type of constrained least squares optimizations problem where the coefficients are not allowed to become negative. That is, given a matrix $\mathbf{A}$ and a (column) vector of response variables $\mathbf{Y}$, the goal is to find $\mathbf{X}$ such that \begin{gather*}
\arg \min_\mathbf{X} \|\mathbf {A\cdot X} -\mathbf {Y} \|_{2}, \mbox{subject to } \mathbf{X} \geq 0
\end{gather*}
where $||\cdot||_2$ denotes Euclidean norm. 
The rationale behind this inferencing step is discussed below. 
\begin{lemma}  W.l.o.g let $\mathcal{X}=\{1,\cdots,m\}$ and let $\mathbf{Y}$ be the vector such that $\mathbf{Y}(i), i \in [m]$ indicates the count of value $i$ in the set $\{\tilde{o}_1,\cdots,\tilde{o}_n\}$ where $\tilde{o}_i=\name(i,\mathcal{X},\epsilon)$.  Given,  \begin{gather}
\mathbf{A}(i,j)=\mathrm{Pr}\big[\name(i,\mathcal{X},\epsilon)=j\big], i,j \in [m]\end{gather}   the solution $\mathbf{X}$ of $\mathbf{A}\cdot \mathbf{X} = \mathbf{Y}$ gives an unbiased frequency estimator ($\mathbf{X}(i)$ is the unbiased estimator for value $i$).\label{lem:unbiased}
\end{lemma}
The proof of the above lemma is presented in \ifpaper the full paper~\cite{anom}\else App. \ref{app:lem:unbiased}\fi. Thus by the above lemma, $\mathbf{X}$ is an unbiased frequency estimator. However, it is important to note that the solution $\mathbf{X}$ is not guaranteed to be non-negative.  But, given our problem setting, the count estimates are constrained to be non-negative. Hence, we opt for an $\textsf{NNLS}$ inferencing. When the exact solution $\mathbf{X}=\mathbf{A}^{-1}\cdot \mathbf{Y}$ is itself non-negative, the estimator obtained from the \textsf{NNLS} optimization is identical to the exact solution. Otherwise, the \textsf{NNLS} optimization gives a biased non-negative estimator that results in minimal least square error. 
The resulting frequency oracle can be used to answer other queries like mean estimation and range queries\footnote{In the \ldp setting, this refers to statistical range query, i.e., the count of the records that belong to a queried range.}. A formal utility analysis is in \ifpaper the full paper~\cite{anom}\else App. \ref{app:freq:utility}\fi.

\section{$\epsilon$-\texorpdfstring{\MakeLowercase{d}}{d}DP 
Order Preserving Encryption (\nam)}\label{sec:encryption}
In this section, we describe our proposed $\epsilon$-\cdp order preserving encryption scheme, \nam.  

\vspace{-0.2cm}
\subsection{Definition of \nam}\label{sec:OPe:def}
The \textit{$\epsilon$-\cdp 
order preserving  encryption} (\nam) scheme is an encryption scheme that bolsters the cryptographic guarantee of a \OPE scheme with an additional \cdp guarantee. 
Here, we detail how our proposed primitive \name can be used in conjunction with a \OPE scheme (Def. \ref{def:OPE}) to form a \nam scheme. 
\begin{definition}[$\epsilon$-\cdp 
Order Preserving Encryption, \nam] A \textit{$\epsilon$-\cdp 
order preserving encryption} scheme, \nam, is composed of a \OPE scheme, $\mathcal{E}$, that satisfies the $\indfaocpa$ guarantee (Def. \ref{def:indfaocpae}), and the \name primitive and is defined by the following algorithms: \\\textbf{\nam Scheme}
\squishlist
\item \textit{Key Generation} ($\textsf{K}_{\epsilon}$). Uses $\K$ from the \OPE scheme to generate a secret key $\s$.
\item \textit{Encryption} ($\textsf{E}_{\epsilon}$). The encryption algorithm inputs a plaintext  $x \in \mathcal{X}$, an  order $\Gamma$, a partition $\Pa \in \hat{\mathcal{X}}$, and the privacy parameter $\epsilon$. It outputs $(\s',y)\leftarrow \E(\s,\tilde{o},\Gamma)$ where $\tilde{o}\leftarrow \name(x,\Pa,\epsilon/2)$.
\item \textit{Decryption} ($\D_\epsilon$). The decryption algorithm uses \textsf{D} to get back $\tilde{o}\leftarrow \D(\s,y)$.
\squishend\label{def:pope}
\end{definition}
Following the above definition, the encryption of a dataset $X \in \mathcal{X}^n, X =\langle x_1,\cdots,x_n\rangle$ is carried out as follows:
\squishlistnum 
\item Set $\s_0\leftarrow\K(1^\kappa)$
\item For $\forall i \in [n]$, compute $(\s_{i},y_i)\leftarrow \E(\s_{i-1}, \tilde{o}_i,\Gamma)$ where $\tilde{o}_i\leftarrow\name(x_i,\Pa,\epsilon/2)$\squishendnum
\begin{tcolorbox}[sharp corners]\vspace{-0.2cm}\textbf{Key Idea.}  A \nam scheme works as follows: 
\squishlist \item First, obtain an (randomized) encoding for the input using the \name primitive (one possible construction is given by Alg. \ref{alg:1} for any given $\epsilon$ and partition $\Pa$). \item Encrypt the above encoding with a \OPE scheme.\squishend \vspace{-0.2cm}\end{tcolorbox}
Thus, ciphertexts encrypted with \nam preserve the order of the corresponding encodings as output by the \name primitive. Referring back to our example, if $X=\langle76,9,9,40,15,76,77\rangle$ and its corresponding encodings are $\widetilde{O}=\{8,1,2,4,2,8,8\}$, then the ciphertext of $X$ under \nam preserves the order of $\widetilde{O}$. \par In other words, since a \OPE scheme preserves the exact order of its input dataset by definition, the utility of \nam (in terms of the preserved ordering information) is determined by the underlying \name primitive. This is formalized by the following theorem.
\begin{theorem}\vspace{-0.1cm}[Utility Theorem] If, for a given partition $\Pa\in \hat{\mathcal{X}}$  and for all $x,x' \in \mathcal{X}$ such that $x > x'$ we have \begin{gather}\vspace{-0.4cm}\mathrm{Pr}\big[\name(x,\Pa,\epsilon)\geq \name(x',\Pa,\epsilon)\big]\geq \alpha, \alpha \in [0,1] \vspace{-0.2cm}\end{gather}  then for a \nam scheme instantiated on such a \name primitive, 
 \begin{gather}\vspace{-0.4cm}\mathrm{Pr}\big[\E_{\epsilon}(x,\s,\Gamma,\Pa,\epsilon)\geq \E_{\epsilon}(x',\s,\Gamma,\Pa,\epsilon) \big]\geq \alpha \vspace{-0.2cm}\end{gather} where $\s\leftarrow \K_{\epsilon}(1^\kappa)$  and any $\Gamma$. \label{thm:utility} \end{theorem} \vspace{-0.1cm}
 The proof follows directly from Defs. \ref{def:OPE} and \ref{def:pope}.
\vspace{-0.15cm}
\begin{lemma}\nam satisfies $\frac{\epsilon}{2}$-\ddp. \label{lem:ddp} \end{lemma}\vspace{-0.1cm}
The proof of the above lemma follows trivially from the post-processing guarantee of \ddp (Thm. \ref{thm:post}).

\subsection{New Security Definition for \nam}\label{sec:security}
Here, we present a novel security guarantee for \nam, namely \textit{indistinguishable ciphertexts under frequency-analyzing $\epsilon$-\cdp 
ordered chosen plaintext attacks} (\indfaocpae, Def. \ref{def:indfaocpae}). 

 The \indfaocpae~guarantee is associated with a security game,   $\mathcal{G}_{\textsf{IND-FA-OCPA}_{\epsilon}}^{\advppt}$, where the adversary, $\advppt$, first chooses four input dataset of equal length, $X_{00},X_{01},X_{10}$ and $X_{11}$, such that  $\Pa_0(X_{00})$ and $\Pa_1(X_{10})$ share at least one randomized order where  $X_{00},X_{01} \in \mathcal{X}^n_0, ~X_{10},X_{11} \in \mathcal{X}^n_1, ~\Pa_0 \in \hat{\mathcal{X}}_0$ and $\Pa_1 \in \hat{\mathcal{X}}_1$. Additionally, $\{X_{00},X_{01}\}$  and $\{X_{10},X_{11}\}$ are $t$-adjacent (Def. \ref{def:cdp}). The challenger then selects two bits $\{b_1,b_2\}$ uniformly at random and returns the corresponding ciphertext for the dataset $X_{b_1b_2}$. $\advppt$ then outputs their guess for the bits and wins the game if they are able to guess either of the bits successfully. The \indfaocpae~guarantee states that $\advppt$ cannot distinguish among the four datasets. In what follows, we first present its formal definition and then, illustrate it using an example. 
\begin{definition}[\indfaocpae] An encryption scheme $\mathcal{E}_\epsilon = (\K_{\epsilon}, \E_{\epsilon}, \D_{\epsilon})$ has indistinguishable ciphertexts under \textit{frequency-analyzing $\epsilon$-\cdp 
 ordered chosen plaintext attacks} if for any \textsf{PPT} adversary, $\advppt$, and security parameter, $\kappa$:
\begin{gather*}\hspace{-3.5cm}\mathrm{Pr}[\mathcal{G}^{\advppt}_{\textsf{FA-OCPA}_\epsilon}(\kappa, b_1,b_2) = (c_1,c_2)] \leq \\e^{t\epsilon}\cdot \mathrm{Pr}\big[\mathcal{G}^{\advppt}_{\textsf{FA-OCPA}_\epsilon} (\kappa, b_1',b_2') = (c_1,c_2)]+ \textsf{negl}(\kappa)\numberthis\label{Eq:main} \end{gather*}
where $b_1, b_2, b_1', b_2',c_1,c_2  \in \{0,1)\}$ and $\mathcal{G}^{\advppt}_{\textsf{FA-OCPA}_\epsilon} (\kappa, b_1,b_2)$ is the random variable indicating the adversary $\advppt$'s output for following security game:\\\\
\textbf{Game} $\mathcal{G}^{\advppt}_{\textsf{FA-OCPA}_\epsilon} (\kappa, b_1,b_2)$\squishlistnum\item  $(X_{00},X_{01},X_{10},X_{11})\leftarrow \advppt$ where \begin{enumerate}\item $X_{00},X_{01}\in \mathcal{X}_0^n$ and $X_{10},X_{11}\in \mathcal{X}_1^n$   \item $\mathcal{P}_0(X_{00})$ and $\mathcal{P}_1(X_{10})$ have at least one common randomized order where $\mathcal{P}_0 \in \hat{\mathcal{X}_0}$ and $\mathcal{P}_1 \in \hat{\mathcal{X}_1}$ \item   $\{X_{00},X_{01}\}$  and $\{X_{10},X_{11}\}$ are $t$-adjacent (Def. \ref{def:cdp}) \end{enumerate}
\item  $\s\leftarrow\K(1^\kappa)$
\item Compute $\widetilde{O}_0\leftarrow\name(X_{00},\mathcal{P}_{0},\frac{\epsilon}{2})$ and $\widetilde{O}_1\leftarrow\name(X_{10},\mathcal{P}_{1},\frac{\epsilon}{2})$.
\item If $\widetilde{O}_0$ and $\widetilde{O}_1$ do not have any common randomized order, then return $\perp$. Else \squishlistnum \item Select two uniform bits $b_1$ and $b_2$ and a randomized order $\Gamma^*$ common to both $\widetilde{O}_0$ and $\widetilde{O}_1$. \item If \scalebox{0.9}{$b_2=0$}, compute $Y_{b_1b_2}\leftarrow\E_{\epsilon}( \widetilde{O}_{b_1},\textsf{S},\Gamma^*,\mathcal{O}_{b_1}, \scalebox{0.7}{$\infty$})\footnote{equivalent to running \scalebox{0.9}{$\E_{\epsilon}( \textsf{S},X_{b_10},\Gamma^*,\Pa_{b_1}, \epsilon):=\langle \widetilde{O}_{b_1}\leftarrow \name(X_{b_10},\Pa_{b_1},\epsilon/2),$}  \scalebox{0.9}{$\E(\widetilde{O}_{b_1},\s,\Gamma^*)\rangle$} where $\mathcal{O}_{b_1}$ is the corresponding encoding domain for $\Pa_{b_1}$ (see App. \ref{app:thm:main})}$. Else, compute $Y_{b_1b_2}\leftarrow\E_{\epsilon}( X_{b_11},\textsf{S},\Gamma^*,\Pa_{b_1}, \frac{\epsilon}{2})$. 
\squishendnum
\item \scalebox{0.95}{$(c_1,c_2)\leftarrow\advppt(Y_{b_1,b_2})$} where $c_1(c_2)$ is $\advppt$'s guess for $b_1(b_2)$ \squishendnum\label{def:indfaocpae}\vspace{-0.1cm}
\end{definition}$\advppt$ is said to win the above game if $b_1=c_1$ or $b_2=c_2$. 
\\\\\textbf{Example 7.} We illustrate the above definition using the following example. 
Consider  $X_{00}=\langle 22, 94, 23,94,36,95\rangle$, $X_{10}=\langle 9,40,11,76,15,76\rangle$,
$X_{01}=\langle 24, 94, 23 , 94 , 36, 95 \rangle$  and $X_{11}=\langle 9, 40, 8,76, 15,76\rangle$ where $\{X_{00},X_{10}\}$ share a randomized order,  $\langle 1,4,2,5,3,6\rangle$, and  $\{X_{00},X_{01}\}$ and $\{X_{10},X_{11}\}$ are $3$-adjacent. For the ease of understanding, we consider the default case of $\Pa_0=\mathcal{O}_0=\mathcal{X}_0$ and $\Pa_1=\mathcal{O}_1=\mathcal{X}_1$. This means that  $\Pa_0(X_{00})=X_{00}$ and so on.
\begin{tcolorbox}[sharp corners]\vspace{-0.2cm}\textbf{\name.} If only \name were to be used to encode the above datasets, then only the pairs $\{X_{00},X_{01}\}$ and $\{X_{10},X_{11}\}$ would be indistinguishable to the adversary (albeit an information theoretic one) because of the $\epsilon$-\cdp guarantee (Defn. \ref{def:cdp}). However, there would be no formal guarantee on the pairs $\{X_{01},X_{11}\},\{X_{01},X_{10}\},\{X_{00},X_{11}\},\{X_{00}, X_{10}\}$.\vspace{-0.2cm}\end{tcolorbox}
\begin{tcolorbox}[sharp corners]\vspace{-0.2cm}\textbf{\OPE.} If we were to use just the \OPE scheme, then only the pair $\{X_{00},X_{10}\}$ would be indistinguishable for $\advppt$ as the rest of the pairs do not share any randomized order. \vspace{-0.2cm}\end{tcolorbox}\vspace{-0.1cm}
\begin{tcolorbox}[sharp corners]\vspace{-0.2cm}\textbf{\nam.}  Using \nam makes all 6 pairs \scalebox{0.9}{$\{X_{00},X_{01}\},$} \scalebox{0.9}{$\{X_{00},X_{11}\},\{X_{00},X_{10}\}, \{X_{01},X_{11}\},\{X_{01},X_{10}\},\{X_{11},X_{10}\}$} indistinguishable for $\advppt$. This is because \nam essentially preserves the order of a $\epsilon$-\cdp scheme.\vspace{-0.2cm}\end{tcolorbox}\vspace{-0.1cm}
Hence, \nam enjoys \textit{strictly stronger security} than both \name and \OPE. \vspace{-0.2cm}
\begin{theorem}The proposed encryption scheme, $\nam$  satisfies  ${\epsilon}-\textsf{IND}$-$\textsf{FA-OCPA}$  security guarantee.\label{thm:faocpae}\end{theorem}\vspace{-0.2cm}
\textit{Proof Sketch}. 
The proof of the above theorem follows directly from the $\indfaocpa$ guarantee of the \OPE scheme and the fact that  \nam  satisfies $\epsilon/2$-\cdp guarantee (Lemma \ref{lem:cdp}, App. \ref{app:thm:main}). The full proof is in App. \ref{app:thm:main}.
\par
Let $\mathcal{N}_{\textsf{G}}(X)=\{X'| X' \in \mathcal{X}^n \mbox{ and } \{X,X'\} \mbox{ are }$ indistinguishable to  $\advppt \mbox{ under guarantee  \textsf{G}}\}$. Additionally, we assume $\mathcal{P}=\mathcal{X}$ for the ease of understanding. Thus, in a nutshell, the $\epsilon$-\cdp guarantee allows a pair of datasets $\{X,X'\}$
to be indistinguishable\footnote{the ratio of their output distributions are bounded by $e^{t\epsilon}$, holds against an information theoretic adversary as well} only if they are $t$-adjacent (for relatively small values of $t$). Referring back to our example, we have $X_{01} \in \mathcal{N}_{\epsilon\text{-}\cdp}(X_{00})$  and $X_{11} \in \mathcal{N}_{\epsilon\text{-}\cdp}(X_{10})$. \par On the other hand, under the \indfaocpa guarantee, $\{X,X'\}$ is indistinguishable\footnote{computational indistinguishability \cite{Oded}} to $\advppt$ only if they share a common randomized order. For instance, $X_{10} \in \mathcal{N}_{\textsf{IND-FA-OCPA}}(X_{00})$. 
\par In addition to the above cases, the \indfaocpae~guarantee allows a pair of datasets $\{X,X'\}$ to be indistinguishable\footnote{Formally given by Eq. \ref{Eq:main} which is structurally similar to that of the $\textsf{IND-CDP}$ guarantee \cite{CDP} which is a computational differential privacy guarantee.}  for $\advppt$ if $\{X,X'\}$ 
\squishlist
\item do not share a randomized order \item are not adjacent, \squishend 
but there exists another dataset $X''$ such that 
\squishlist
\item $\{X',X''\}$ are  adjacent, i.e. $X' \in \mathcal{N}_{\epsilon\text{-}\cdp}(X'')$
\item $\{X,X''\}$ share a randomized order, i.e., $X'' \in \mathcal{N}_{\textsf{IND-FA-OCPA}}(X)$. 
\squishend
From our aforementioned example, we have $X_{11} \not \in \mathcal{N}_{\textsf{IND-FA-OCPA}}(X_{00})$ and $ X_{11} \not \in \mathcal{N}_{\epsilon\text{-}\cdp}(X_{00})$. But still,  $X_{11} \in \mathcal{N}_{\textsf{IND-FA-OCPA}_{\epsilon}}(X_{00})$ since $X_{11} \in \mathcal{N}_{\epsilon\text{-}\cdp}(X_{10})$ and $X_{10} \in \mathcal{N}_{\textsf{IND-FA-OCPA}}(X_{00})$.   
 Thus, formally \begin{gather}\vspace{-0.4cm}\mathcal{N}_{\textsf{IND-FA-OCPA}_{\epsilon}}(X) =\bigcup_{X''\in \mathcal{N}_{\textsf{IND-FA-OCPA}}(X)}  \mathcal{N}_{\epsilon\text{-}\cdp}(X'')\vspace{-0.4cm}\end{gather}  
 Since, trivially \scalebox{0.9}{$X\in\vspace{-0.1cm} \mathcal{N}_{\textsf{IND-FA-OCPA}}(X)$} and \scalebox{0.9}{$X \in \mathcal{N}_{\epsilon\text{-}\cdp}(X)$}, we have \scalebox{0.9}{$\mathcal{N}_{\textsf{IND-FA-OCPA}_{\epsilon}}(X) \supseteq \mathcal{N}_{\textsf{IND-FA-OCPA}}(X)$} and \scalebox{0.9}{$\mathcal{N}_{\textsf{IND-FA-OCPA}_{\epsilon}}(X) \supseteq \mathcal{N}_{\epsilon\text{-}\cdp}(X)$}.
\begin{tcolorbox}[sharp corners]\vspace{-0.2cm}\textbf{Key Insight.}
The key insight of the \indfaocpae~ security guarantee is that the \OPE scheme preserves the order of the outputs of a $\epsilon$-\cdp mechanism. As a result, the adversary is now restricted to only an $\epsilon$-\cdp order leakage from the ciphertexts. Hence, even if the security guarantee of the \OPE layer is completely broken, the outputs of \nam would still satisfy $\epsilon$-\cdp due to Thm.~\ref{thm:post}. 
Referring to Example 1, in the very least input pairs $\{X_{00},X_{01}\}$ and $\{X_{10},X_{11}\}$ will remain indistinguishable under all inference attacks. Thus, \nam  is the first encryption scheme to satisfy a formal security guarantee against all possible inference attacks and still provide some ordering information about the inputs.\vspace{-0.2cm} \end{tcolorbox}
\begin{tcolorbox}[sharp corners]\vspace{-0.2cm}\textbf{Remark 2.} The \indfaocpae~ guarantee of the \nam scheme is \textit{strictly stronger than both \cdp (\ddp) and \indfaocpa (the strongest possible guarantee for any \OPE)}. Further, it depends only on the \ddp guarantee of the underlying \name primitive which is independent of the partition $\Pa$ used (as discussed in Sec. \ref{sec:construction}). We discuss the role of $\Pa$ in Sec. \ref{sec:application}. \vspace{-0.2cm}\end{tcolorbox}

\section{\nam and Inference Attacks} \label{sec:attack} \vspace{-0.3cm}
  \begin{figure}[ht]
      \includegraphics[width=\linewidth]{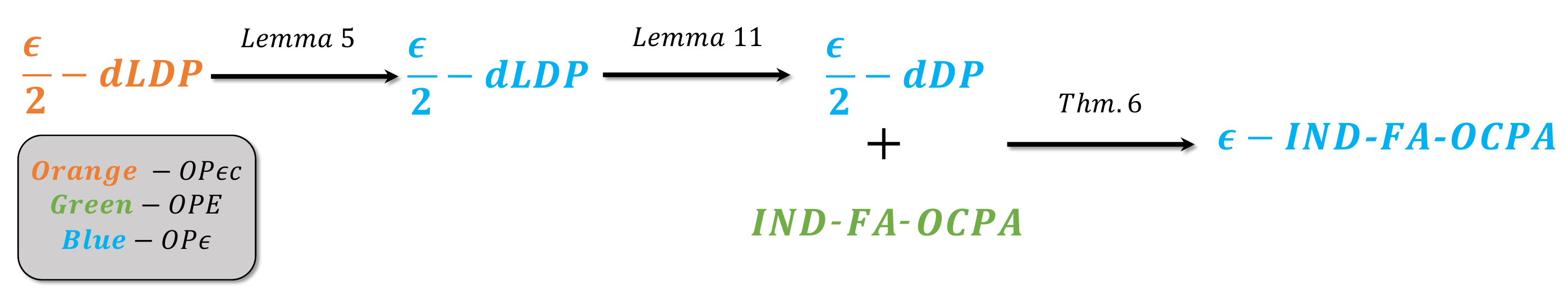}\vspace{-0.4cm}
        \caption{Relationships between \ddp, \cdp and \indfaocpa guarantees.} \label{fig:relation}
\vspace{-0.45cm}    \end{figure}
   In this section, we discuss the implications of \nam's security guarantee  in the face of inference attacks. Specifically, we formalize the protection provided by \nam's (relaxed) \DP guarantee -- this is the worst case guarantee provided by \nam. 
   \par 
 Recall that the \indfaocpae~ guarantee of a \nam bolsters the cryptographic guarantee of a \OPE (\indfaocpa) with an additional layer of a (relaxed) \DP guarantee. For the rest of the discussion, we focus on the worst case scenario where the \OPE scheme provides no protection at all and study what formal guarantee we can achieve from just the (relaxed) \DP guarantee. As discussed in Sec. \ref{sec:background:DP}, our proposed distance-based relaxation of \DP comes in two flavors -- local (\ddp, Def. \ref{def:ddp})  and central (\cdp, Def. \ref{def:cdp}). Intuitively, \ddp is a guarantee for each individual data point while \cdp is a guarantee for a dataset.  As a refresher, Fig. \ref{fig:relation} showcases the relationships between them.  The most salient point is that the \ddp is a stronger guarantee than \cdp -- $\epsilon$-\ddp implies $\epsilon$-\cdp. Thus, owing to the \ddp guarantee of the underlying \name primitive, \nam trivially satisfies both \ddp (Lemma \ref{lem:ddp}) and \cdp (Lemma \ref{lem:cdp} App. \ref{app:thm:main}) guarantees. \par For our discussion in Sec. \ref{sec:security}, we use the \cdp guarantee since the \indfaocpa guarantee of \OPEs is also defined on datasets. 
 In what follows,  we show how to interpret the protection provided by \nam's \ddp guarantee since it is stronger and holds for every data point. We do so with the help of an indistinguishability game, as is traditional for cryptographic security definitions. Let the input be drawn from a discrete domain of size $N$, i.e., $|\mathcal{X}|=N$. The record indistinguishability game, $\mathcal{G}^{\mathcal{A}}_{\beta-RI}$, is characterized by a precision parameter $\beta \in [\frac{1}{N},1]$.  In this game, the adversary has to distinguish among a single record (data point) $x$  and set of values $Q(x)$ that differ from $x$ by at most $\lceil\beta N\rceil$. For instance, for $x=3$, $N=10$ and $\beta=1/5$, the adversary has to distinguish among the values \scalebox{0.9}{$3$} and \scalebox{0.9}{$Q(x)= \{1,2,4,5\}$ ($\lceil\beta N\rceil=2$)}. Let $y_i$ denote the ciphertext for $x_i$ after encryption with \nam. The game is formally defined as follows: 
 \\\textbf{Game} $\mathcal{G}^{\mathcal{A}}_{\beta-RI}(p)$ 
 \squishlistnum\item $x_0\leftarrow \mathcal{A}$
 \item $Q(x)=\{x_1, \cdots, x_q\}$ where $x_i \in \mathcal{X}, i \in [q]$ s.t. $|x_0-x_i|\leq \lceil \beta N\rceil$ and $x_i\neq x_0$  \item Select $p \in \{0,1,\cdots, q\}$ uniformly at random 
\item $p'\leftarrow \mathcal{A}\big(y_p\big)$ 
\squishendnum
$\mathcal{A}$ is said to win the above game if $p'=p$. Let $rand$ be a random variable indicating the output of the baseline strategy  where the adversary just performs random guessing.
\begin{theorem} \vspace{-0.1cm}
For a \nam scheme satisfying $\frac{\epsilon}{2}$-\ddp, we have \vspace{-0.2cm} 
\begin{gather}
\vspace{-0.2cm}
\Big|\mathrm{Pr}\big[p'=p\big]-\mathrm{Pr}\big[rand=p\big]\Big|\leq \frac{e^{\epsilon^{*}}}{q+e^{\epsilon^{*}}}-\frac{1}{q+1} \label{Eq:attack}\vspace{-0.8cm}\end{gather} where $\epsilon^*=\epsilon \lceil\beta N\rceil$ and $q=|Q(x_0)|$ (Step $(2)$ of game $\mathcal{G}^{\mathcal{A}}_{\beta-RI}$).\label{thm:attack:record}
\end{theorem}
From the above theorem, observe that for low values of $\epsilon^{*}$ (i.e., low $\epsilon$ and $\beta$) the R.H.S of the Eq. \eqref{Eq:attack} is low. This means that \textit{for reasonably low values of $\epsilon$ (high privacy), with very high probability an adversary cannot distinguish among input values that are close to each other (small $\beta$) any better than just random guessing}. Now, recall that owing to the \ddp guarantee of the underlying \name primitive, every data point encrypted with \nam is also protected by the \ddp guarantee (Lemma \ref{lem:ddp}). 
This implies that, for any dataset $X$, the above indistinguishability result holds for every individual data point (record) simultaneously. Thus, \textit{the \ddp guarantee rigorously limits the accuracy of any inference attack for every record of a dataset}. The proof follows directly from the \ddp guarantee  (App. \ref{app:attack}). \par 
As a concrete example, let us look at the binomial attack~\cite{Grubbs1} on \OPE schemes satisfying \indfaocpa. The attack uses a biased coin model to locate the range of ciphertexts corresponding to a particular plaintext. Experimental results on a dataset of first names show that the attack can recover records corresponding to certain high frequency plaintexts (such as, first name ``Michael'') with high accuracy. In this context, the implications of the above result is as follows. Consider a dataset with plaintext records corresponding to first names ``Michael'' and ``Michele''.  For \nam, the recovery rate for either would not be better than the random guessing baseline since both the values are close to each other in alphabetic order.

Note that the above result is information-theoretic and holds for \textit{any} adversary --  active or passive, both in the persistent (access to volume/access-pattern/search-pattern leakage) and snapshot attack models (access to a single snapshot of the encrypted data) \cite{sok}.
\begin{tcolorbox}[sharp corners]\vspace{-0.2cm}
\textbf{Remark 4.} In the very least, \nam \textit{rigorously limits the accuracy of any inference attack for every record of a dataset for all adversaries.} (Thm. \ref{thm:attack:record}). 
\vspace{-0.2cm}
\end{tcolorbox}

\vspace{-0.5cm}

 \vspace{-0.2cm}
\section{\nam for Encrypted Databases} \label{sec:application} 
In this section, we describe how to use a \nam scheme in practice for encrypted databases. We discuss how we can leverage the partition parameter, $\Pa$, of the underlying \name primitive for improved utility. 
\\\textbf{Problem Setting.} For encrypted databases, a data owner has access to the entire database in the clear and encrypts it before outsourcing it to an untrusted server. 
 The queriers of the encrypted databases are authorized entities with access to the secret keys. In many practical settings the data owner themselves is the querier \cite{sok}.
 \par The most popular use case for databases encrypted with \OPEs is retrieving the set of records belonging to a queried range. However, due to randomization, encryption with \nam leads to loss in utility. Specifically in the context of range queries, it might miss some of the correct data records and return some incorrect ones. For the former, constraining \nam to maintain only a partial order is found to be helpful. As discussed in Sec. \ref{sec:construction}, the more coarse grained the partition is (the lesser the number of intervals), the larger is the probability for \name to output the correct encoding.  Hence, if any given range $[a,b]$ is covered by a relatively small number of intervals in $\Pa$, then with high probability the set of records corresponding to the encodings $\{\tilde{o}|\tilde{o}\in O \wedge \Pa(a) \leq \tilde{o} \leq \Pa(b)]\}$ will contain most of the correct records. This results in better accuracy for the subsequent \nam scheme since it's accuracy is determined by the underlying \name primitive (Thm. \ref{thm:utility}).   
 \par The problem of returning incorrect records can be mitigated by piggybacking every ciphertext encrypted with \nam with another ciphertext obtained from encrypting the corresponding plaintext with a standard authenticated encryption scheme \cite{encryption}, $\overline{\mathcal{E}}:=\langle \overline{\K},\overline{\E},\overline{\D}\rangle$. We refer to this as the augmented \nam  \footnote{The augmented \nam scheme still upholds the \indfaocpae~ guarantee  owing to the semantic security of the encryption scheme $\overline{\mathcal{E}}$.} scheme:  \\\textbf{Augmented \nam, $\mathcal{E}^{\dagger}$} \squishlist
\item \textit{Key Generation }($\textsf{K}_{\epsilon}^{\dagger}$).  This algorithm generates a pair of keys $(\s,K)$ where $\s\leftarrow\K_{\epsilon}(\kappa)$ and $K\leftarrow \overline{\K}(\kappa)$
\item \textit{Encryption} ($\textsf{E}_{\epsilon}^{\dagger}$). The encryption algorithm generates  $(\s',y_0,y_1)$ where $\tilde{o}\leftarrow\name(x,\Pa,\epsilon/2)$, \scalebox{0.95}{$(\s',y_{0})\leftarrow \E(\s, \tilde{o},\Gamma)$},  \scalebox{0.95}{ $y_{1}\leftarrow \overline{\E}(K,x)$} 
\item \textit{Decryption} ($\textsf{D}_{\epsilon}^{\dagger}$). The decryption algorithm uses $\s$ and $K$ to decrypt both the ciphertexts, $(x,\tilde{o})$ as \scalebox{0.9}{$\tilde{o}\leftarrow \D_{\epsilon}(\s,y_0)$} and \scalebox{0.9}{$x\leftarrow \overline{\D}(K,y_{1})$}.
\squishend 
After receiving the returned records from the server, the querier can decrypt $\{y_{i1}\}$ and discard the irrelevant ones. The cost of this optimization for the querier is the processing overhead for the extra records (see discussion later).
\par For most input distributions, an equi-depth partitioning works well (as demonstrated by our evaluation in Sec. \ref{sec:evaluation:results}). Nevertheless, the partition can be updated dynamically as well (App. \ref{app:discussion}). 
\begin{tcolorbox}[sharp corners]\vspace{-0.2cm}\textbf{Remark 5.} The partitioning of the input domain ($\Pa$)  has no bearing on the formal security guarantee. It is performed completely from an utilitarian perspective in the context of encrypted databases -- it results in an accuracy-overhead trade-off  (accuracy -- number of correct records retrieved; overhead -- number of extra records processed).\vspace{-0.2cm} \end{tcolorbox}

\textbf{Range Query Protocol.} The end-to-end range query protocol is described in Alg. \ref{alg:range}. Before detailing it, we will briefly discuss the protocol for answering range queries for a \OPE scheme, $\mathcal{E}$, that satisfies the $\indfaocpa$ guarantee (see \cite{IND-FAOCPA} for details).  Recall that every ciphertext is unique for such a \OPE scheme.  Hence, a querier has to maintain some state information for every plaintext. Specifically, if $Y=\{y_1,\cdots,y_n\}$ denotes the corresponding ciphertexts for an input set $X=\{x_1,\cdots,x_n\}$, then the querier stores the maximum and minimum ciphertext in $Y$ that corresponds to the plaintext $x_i$, denoted by $\max_\mathcal{E}(x_i)$ and $\min_{\mathcal{E}}(x_i)$, respectively.  For answering a given range query $[a,b]$, the querier asks for all the records in $Y$ that belong to $[\min_{\mathcal{E}}(a),\max_{\mathcal{E}}(b)]$. Recall that in \nam, the \OPE scheme is applied to the output (encodings) of the \name primitive. So now for answering $[a,b]$, the querier has to retreive records corresponding to $[\Pa(a),\Pa(b)]$ instead where $\Pa$ is the partition for the encoding. Hence, the querier first maintains the state information for the encodings (Steps 1-6, Alg. \ref{alg:range}). Note that since the size of the encoding space is smaller than the input domain $\mathcal{X}$, the amount of state information to be stored for a \nam is less than that for a \OPE (see \ifpaper full paper~\cite{anom}\else App. \ref{app:discussion}\fi). Next, the querier asks for all the encrypted records in the set $Y'=\big\{\langle y'_{i0},y'_{i1}\rangle|i\in [n] \mbox{ and } y'_{i0} \in [\min_{\mathcal{E}^{\dagger}}(\Pa(a)),\max_{\mathcal{E}^{\dagger}}(\Pa(b)]\big\}$ from the server (Steps 7-10). On receiving them, the querier only retains those records that fall in the queried range (Steps 11-18). 
\par There are two ways the utility can be further improved. The first is including records from some of the  intervals preceding $\Pa(a)$ and following $\Pa(b)$. The querier can ask for the records in \scalebox{0.95}{$[\max(o_1,o_{a-l}),\min(o_{b+l},o_k))], l \in \mathbb{Z}_{\geq 0}$} where $o_a:=\Pa(a)$ and $o_b:=\Pa(b)$.  However, the cost is an increased number of extra records.   \par Another optimization is to answer a workload of range queries at a time. Under \nam, queries can be made only at the granularity of the partition. Thus, if a queried range $[a,b]$ is much smaller than $[\Pa(a),\Pa(b)]$, then the querier has to pay the overhead of processing extra records. This cost can be reduced in the case of a workload of range queries where multiple queries fall within $[\Pa(a),\Pa(b)]$ (records that are irrelevant for one query might be relevant for some other in the workload). Additionally, the number of missing records for the query $[a,b]$ is also reduced if the records from the neighboring intervals of $[\Pa(a),\Pa(b)]$ are also included in the response (owing to the other queries in the workload).
\\\textbf{Discussion.} As described above, the server side interface for range query protocols is the same for both \nam and a \OPE scheme with the $\indfaocpa$ guarantee (with a nominal change to accommodate the extra ciphertexts $\{y_{i1}\}$). 
The cost is the extra storage for $\{y_{i1}\}$. However, in this age of cloud services, outsourced storage ceases to be a bottleneck \cite{storage:not:bottleneck}. \par The querier, on the other hand, needs to decrypt  all the returned records (specifically, $\{y_{i1}\}$}). However, decryption is in general an efficient operation. For instance, the decryption of $1$ million ciphertexts encrypted with AES-$256$ GCM  requires $<3$ minutes in our experimental setup. Thus, on the overall there is no tangible overhead in adopting \nam.
\vspace{-0.2cm} \begin{tcolorbox}[sharp corners, breakable]\vspace{-0.2cm}\textbf{Remark 6.} \nam could be used for secure data analytics in settings where (1) the $\epsilon$-\cdp guarantee is acceptable, i.e., the main security concern is preventing the distinction between input values close to each other, and $(2)$ the application can tolerate a small loss in utility. 
Specifically in such settings, replacing encrypted databases that are already deploying \OPE schemes (satisfying $\indfaocpa)$
with a \nam scheme would give a strictly stronger security guarantee against all attacks with \textit{nominal change in infrastructure or performance} -- a win-win situation. \vspace{-0.2cm}
\end{tcolorbox}


 \setlength{\textfloatsep}{2pt}
 \begin{algorithm}
\small
\caption{Range Query Protocol}\label{alg:range}
\begin{algorithmic}[1]
\Statex \textbf{Notations:} $Z$ - Input dataset with $n$ records $(r_i, x_i)$ where \scalebox{0.9}{$x_i\in \mathcal{X}$} \Statex \hspace{1.75cm} denotes the sensitive attribute to be encrypted under  \Statex \hspace{1.75cm} \nam  and $r_i$ denotes the rest of associated  data (other  \Statex \hspace{1.75cm} attributes could be encrypted as well);
 \Statex \hspace{1.4cm} $\s$ - Secret key for \nam; $\Pa$- Partition for \nam;
\Statex \hspace{1.4cm} $K$- Secret key for the authenticated encryption  scheme $\overline{\mathcal{E}}$
\Statex \textbf{Input:} Range Query $[a,b], a,b \in \mathcal{X}$
\Statex \textbf{Output:} Set of records $V=\{r_i| ( r_i,x_i ) \in Z, x_i \in [a,b]\}^{10}$

\Statex \textbf{Initialization: Querier}
\State $X=\langle x_1,\cdots,x_n\rangle$ 
\State $Y=\mathcal{E}^{\dagger}(X,\s,K,\Gamma,\Pa,\frac{\epsilon}{2})$ \hfill\textcolor{blue}{$\rhd$} Contains encrypted attributes  $\{( y_{i0},y_{i1})\}$
\State \textbf{for} $o \in \mathcal{O}$
\State \hspace{0.2cm}  $\max_{\mathcal{E}^{\dagger}}(o)=\max\{y_{i0}|(y_{i0},y_{i1}) \in Y \text{ and $y_{i0}$ decrypts to } o\}$
\State \hspace{0.2cm}  $\min_{\mathcal{E}^{\dagger}}(o)=\min\{y_{i0}|(y_{i0},y_{i1}) \in Y \text{ and $y_{i0}$ decrypts to } o \}$
\State \textbf{end for} \hfill \textcolor{blue}{$\rhd$} Querier maintains state information
\Statex \textbf{Range Query Protocol}
\Statex \textbf{Querier}
\State $C=\{\min_{\mathcal{E}^{\dagger}}(\Pa(a)),\max_{\mathcal{E}^{\dagger}}(\Pa(b)\}$ \hfill \textcolor{blue}{$\rhd$} Transformed range query \Statex \hfill  based  on state information 
\State\textit{Querier} $\xrightarrow{C}$ \textit{Server} 
\Statex \textbf{Server}
\State 
 \scalebox{0.95}{$Y'=\big\{( y_{i0},y_{i1})|i\in [n] \mbox{ and } y_{i0} \in \big[\min_{\mathcal{E}^{\dagger}}\big(\Pa(a)\big),\max_{\mathcal{E}^{\dagger}}\big(\Pa(b)\big)\big]\big\}$} 
 \Statex \hfill \textcolor{blue}{$\rhd$} Server returns the set of records matching the query
 \State \textit{Server} $\xrightarrow{Y'}$ \textit{Querier}
 \Statex\textbf{Querier}
 \State $V=\phi$
 \State \textbf{for} $y_{i1} \in Y'$
 \State \hspace{0.51cm} $x_i' \leftarrow \overline{\D}(K,y_{i1}')$
 \State \hspace{0.51cm} \textbf{if} $(x_i' \in [a,b])$  \hfill \textcolor{blue}{$\rhd$} Verifying whether record falls in $[a,b]$ 
 \State \hspace{0.8cm} $V=V\cup r_i$
 \State \hspace{0.51cm} \textbf{end if}
 \State \textbf{end for}
  \State \textbf{Return} $V$ 
 \end{algorithmic}
\label{alg:range}
\end{algorithm}
\section{Experimental Evaluation}\label{sec:evaluation}
 In this section, we answer the following three questions:\textcolor{white}{f\footnote{$V$ has a small utility loss as explained before.}}
\squishlist\item \textbf{Q1:} Does \nam retrieve the queried records with high accuracy? \item \textbf{Q2:} Is the processing overhead of \nam  reasonable? \item \textbf{Q3:} Can \name answer statistical queries in the \ldp setting with high accuracy?
\squishend
\vspace{-0.05cm}\textbf{Evaluation Highlights}
\squishlist\item \nam retrieves almost all the records of the queried range. For instance, \nam only misses around $4$ in every $10$K correct records on average for a dataset of size $\sim 732$K with an attribute of domain size $\sim 18$K and $\epsilon=1$. \vspace{-0.1cm}\item The overhead of processing the extra records for \nam is low. For example, for the above dataset, the  number  of  extra  records  processed  is  just $0.3\%$ of the dataset size for \scalebox{0.9}{$\epsilon=1$}. \item We give an illustration of \nam's protection against inference attacks.  For an age dataset  and an adversary with real-world auxiliary knowledge, no inference attack in the snapshot attack model can distinguish between two age values $(x,x')$ such that $|x-x'|\leq 8$ for $\epsilon=0.1$.      \item 
\name can answer several queries in the \ldp setting with high accuracy. For instance, \name can answer ordinal queries with $94.5\%$ accuracy for a dataset of size $\sim38$K, an attribute of domain size $\sim240$K and $\epsilon=1$. Additionally, \name achieves $6
\times$ lower error than the state-of-the-art $\epsilon$-\ldp technique for frequency estimation for $\epsilon=0.1$. \squishend
\vspace{-0.4cm}\begin{figure*}[hbt!]
    \begin{subfigure}[b]{0.25\linewidth}
        \centering
         \includegraphics[width=0.9\linewidth]{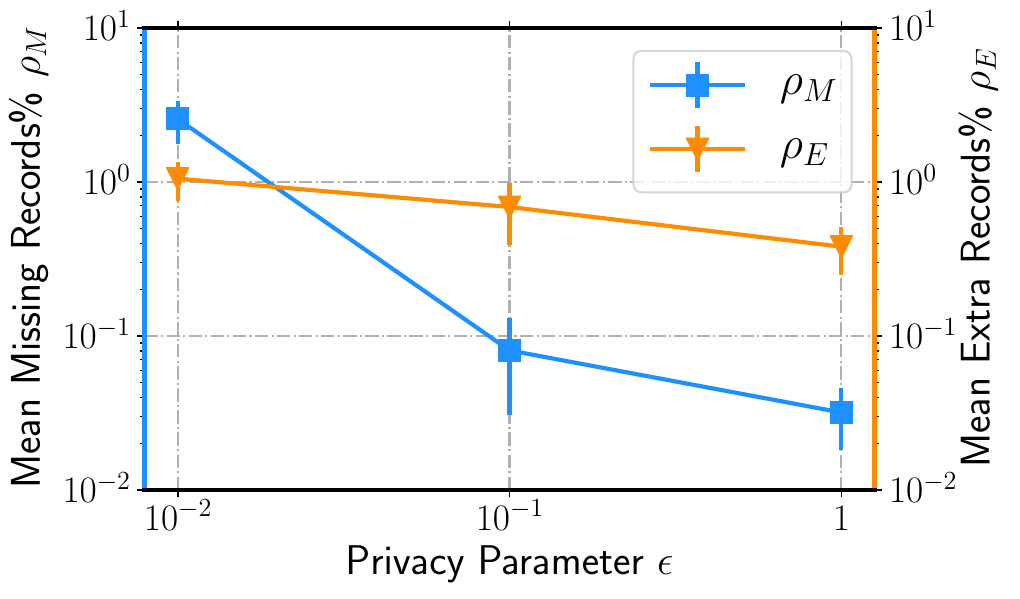}
            \vspace{-0.25cm} 
            \caption{\scalebox{0.8}{\textsf{PUDF}: Effect of $\epsilon$}}
        \label{fig:cdp:eps:1}
    \end{subfigure}
    \begin{subfigure}[b]{0.25\linewidth}
    \centering \includegraphics[width=0.9\linewidth]{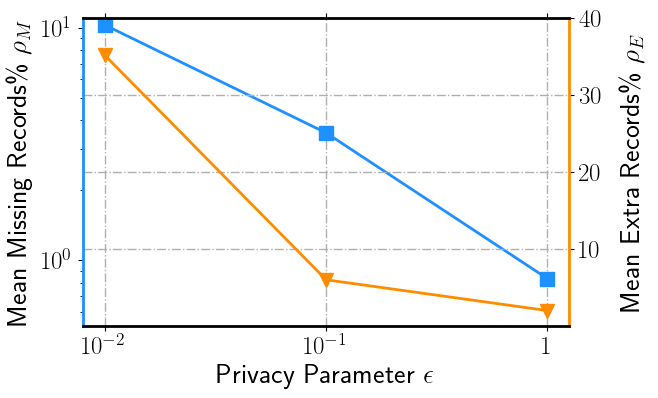}   
\vspace{-0.25cm} 
\caption{\scalebox{0.8}{\textsf{SPARC}: Effect of $\epsilon$}}
        \label{fig:cdp:eps:2}\end{subfigure}
    \begin{subfigure}[b]{0.25\linewidth}
    \centering    \includegraphics[width=0.9\linewidth]{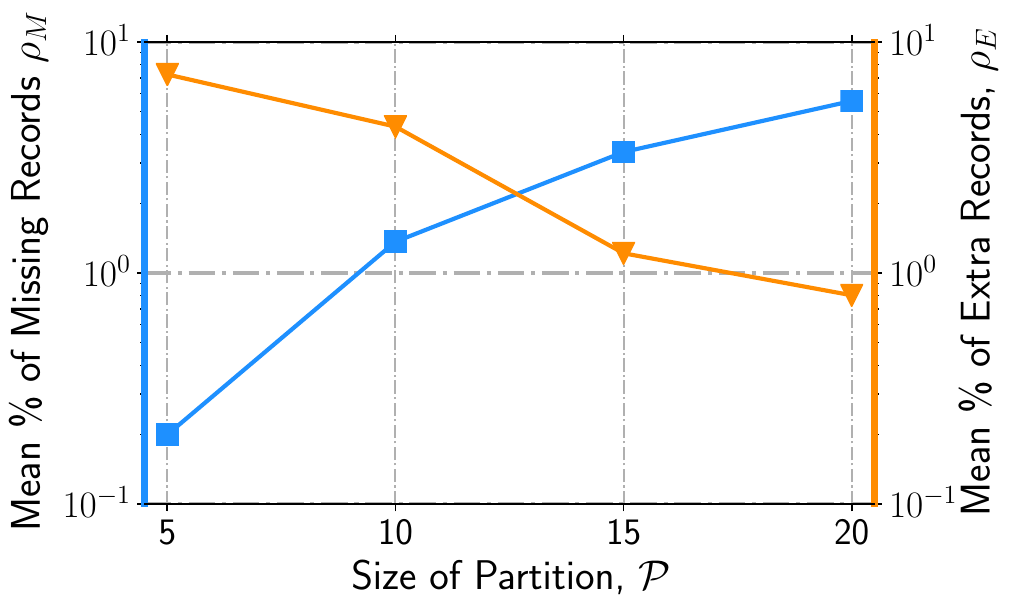}
    \vspace{-0.25cm}    
    \caption{\scalebox{0.8}{\textsf{Adult}: Effect of $\Pa$}}
        \label{fig:cdp:partition:1}\end{subfigure}
      \begin{subfigure}[b]{0.25\linewidth}
    \centering    \includegraphics[width=0.9\linewidth]{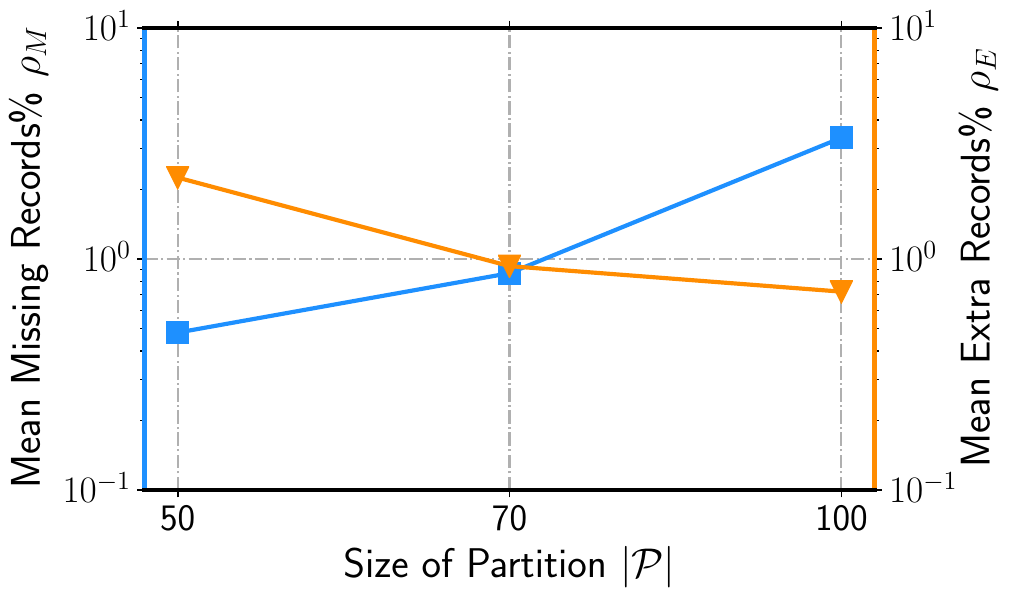}    
         \vspace{-0.25cm}  \caption{\scalebox{0.8}{\textsf{Salary}: Effect of $\Pa$}}
        \label{fig:cdp:partition:2}\end{subfigure}
   \begin{subfigure}[b]{0.25\linewidth}
        \centering
         \includegraphics[width=0.9\linewidth]{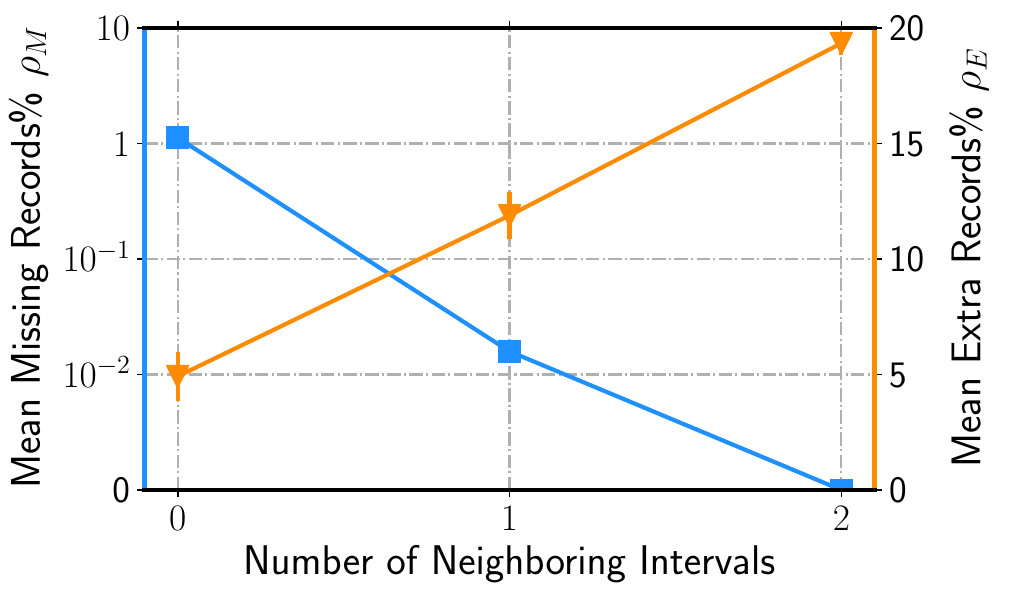}
      \vspace{-0.25cm}  \caption{\scalebox{0.8}{\textsf{Adult}: Effect of Neighboring intervals}}
        \label{fig:cdp:neighbor:1}
    \end{subfigure}
    \begin{subfigure}[b]{0.25\linewidth}
    \centering \includegraphics[width=0.9\linewidth]{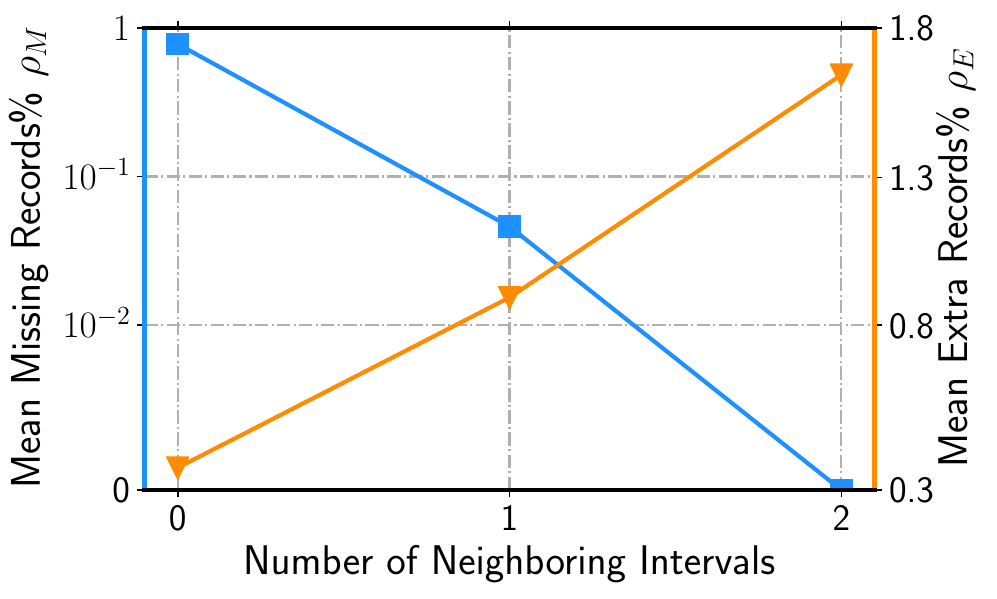}  
    \vspace{-0.25cm} \caption{\scalebox{0.8}{\textsf{Salary}: Effect of neighboring intervals}}
\label{fig:cdp:neighbor:2}\end{subfigure}
    \begin{subfigure}[b]{0.25\linewidth}
    \centering    \includegraphics[width=0.9\linewidth]{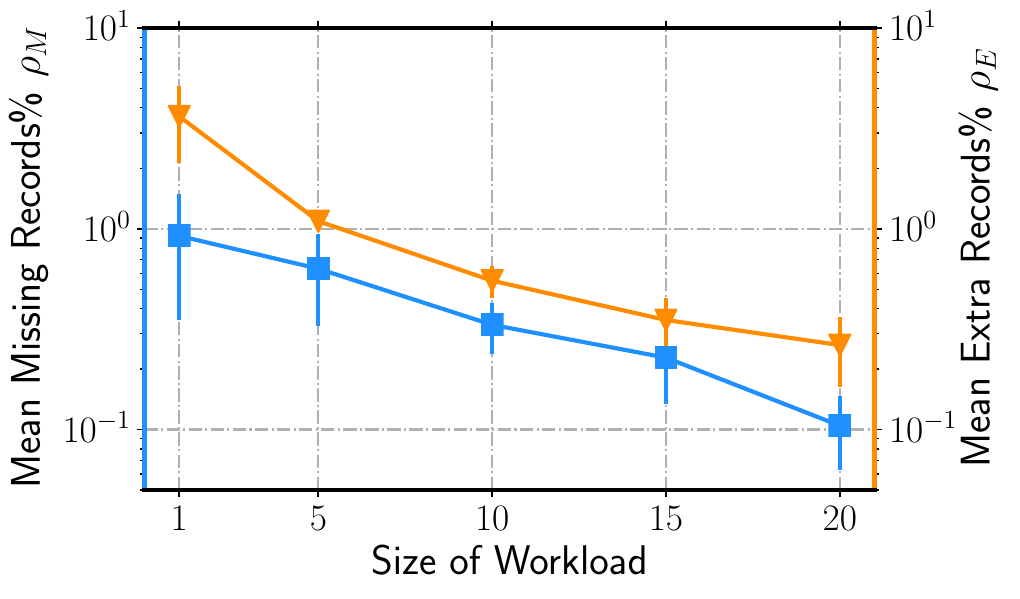}
       \vspace{-0.25cm} \caption{\scalebox{0.8}{\textsf{SPARC}: Effect of Workload}}
        \label{fig:cdp:workload:1}\end{subfigure}
      \begin{subfigure}[b]{0.25\linewidth}
    \centering    \includegraphics[width=0.9\linewidth]{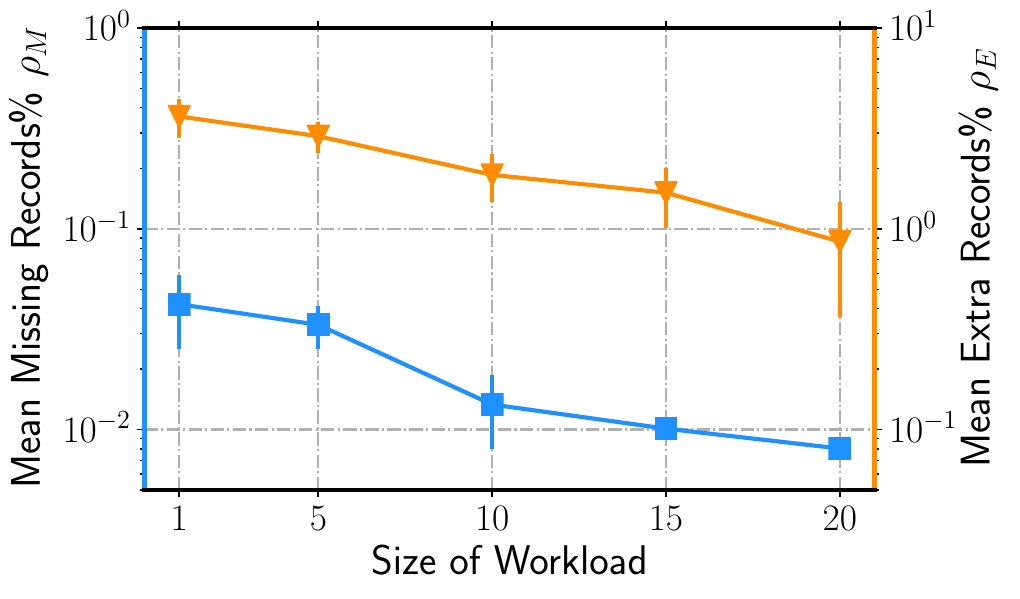}
       \vspace{-0.25cm} \caption{\scalebox{0.8}{\textsf{PUDF}: Effect of Workload}}
        \label{fig:cdp:workload:2} 
    \end{subfigure}
\vspace{-0.7cm}
   \caption{Accuracy Analysis of \nam in the Context of Encrypted Databases}
   \label{fig:cdp}\vspace{-0.5cm}
\end{figure*}
\subsection{Experimental Setup}
\textbf{Datasets.} We use the following datasets: \squishlist\item \textit{PUDF}~\cite{PUDF}. This is a hospital discharge data from Texas. We use the 2013 PUDF data and the attribute \textsf{PAT\_ZIP} ($\sim732$K records of patient's 5-digit zipcode from the domain  $[70601,88415]$). 
\item \textit{Statewide Planning and Research Cooperative System (SPARCS)}~\cite{NYC}. This is a hospital inpatient discharge dataset from the state of New York. This dataset has $\sim2532$K records and we use the \textsf{length\_of\_stay} (domain $[1,120]$) attribute for our experiments.  \item \textit{Salary}~\cite{salaries}. This dataset contains the compensation for San Francisco city employees. We use the attribute \textsf{BasePay} (domain $[1000,230000]$) from the years $2011$ ($\sim36$K records) and $2014$ ($\sim38$K records).  \item \textit{Adult}~\cite{UCI}. This  dataset is derived from the 1994 Census. The dataset has $\sim33$K records and we use the attribute \textsf{Age} (domain $[1,100]$) for our experiments. \item \textit{Population}~\cite{census:age}. This is a US Census dataset of annual estimates of the resident population by age and sex. We use the data for male Puerto Ricans for  $2011$ and $2019$. \squishend
Datasets \textsf{Adult} and \textsf{SPARCS} have small and dense domains while \textsf{PUDF} and \textsf{Salary} have larger and sparse domains.
\\\textbf{Metrics.} 
We evaluate \textbf{Q1} using the relative percentage of missing records, \scalebox{0.9}{$\rho_M=\frac{\# \mbox{missing records}}{\#\mbox{correct records}}\%$}. Note that  $\rho_M$ essentially captures false negatives which is the only type of error encountered -- the querier can remove all cases false positives as discussed in Sec. \ref{sec:application}. We evaluate \textbf{Q2}  via the percentage of extra records processed relative to the dataset size, \scalebox{0.9}{$\rho_E= \frac{\#\mbox{extra records}}
{\#\mbox{records in dataset}}\%$}. A key advantage of outsourcing is that the querier doesn’t have to store/process the entire database. $\rho_E$ 
measures this -- low $\rho_E$ implies that the (relative) count of extra records is low and it is still advantageous to outsource. Thus, low $\rho_E$ implies that the client's processing overhead is low (relative to the alternative of processing the whole dataset). We believe this is a good metric for assessing the overhead because: \squishlist \item For clients, the decryption of extra records doesn't result in a tangible overhead (\scalebox{0.9}{$1$} million records take \scalebox{0.9}{$<3$} minutes, Sec. \ref{sec:application}). \item\nam has no impact on the server since its interface (functionality) is the same as that for \OPEs. \squishend
For evaluating ordinal queries (Fig. \ref{fig:ldp:ordinal}), we use $\sigma_{k}=\%\mbox{ of points with }$  $|\Pa(x)-\tilde{o}_x|=k$ where  $\Pa(x)$ and $\tilde{o}_x$ denote the correct and noisy encoding for $x$, respectively.  For instance, $\sigma_0=90$  
means that 90\% of the input data points were mapped to the correct bins.  
For frequency and mean estimation (Figs. \ref{fig:ldp:freq} and \ref{fig:ldp:mean}),  we measure the absolute  error $|c-\tilde{c}|$ where $c$ is the true answer and $\tilde{c}$ is the noisy output. For Fig. \ref{fig:ldp:range}, we use the error metric $|c-\tilde{c}|/k$ where $k$ is the size of the query. We report the mean and s.t.d of error values over $100$ repetitions for every experiment.
\\\textbf{Configuration.} All experiments were conducted on a Macbook with i5, 8GB RAM and
OS X Mojave (v10.14.6). We used Python 3.7.6. 
The reported privacy parameter $\epsilon$ refers to the $\epsilon$-\indfaocpa guarantee, and implies $\frac{\epsilon}{2}$-\cdp and $\frac{\epsilon}{2}$-\ddp for \nam (Fig. \ref{fig:relation}). We instantiate the \name primitive using Alg. \ref{alg:1}. 
  Due to lack of space, we present the results for all only two datasets, 1 dense (\textsf{Adult}, \textsf{SPARCS}) and 1 sparse (\textsf{PUDF}, \textsf{Salary}), in Fig. \ref{fig:cdp}. The default settings are $\epsilon=1$, equi-depth partitioning (based on a non-private prior) of sizes $|\Pa|=122$  for \textsf{PUDF}, $|\Pa|=8$  for \textsf{SPARC}, $|\Pa|= 10$ for \textsf{Adult}, and $|\Pa|=70$ for \textsf{Salary}. The range queries are chosen uniformly at random. We use the data from \textsf{Salary} for $2011$ as an auxiliary dataset for Fig. \ref{fig:ldp:ordinal}.
  
\vspace{-0.4cm}\subsection{Experimental Results}\label{sec:evaluation:results}
\begin{figure*}[ht]
    \begin{subfigure}[b]{0.25\linewidth}
        \centering
         \includegraphics[width=0.9\linewidth]{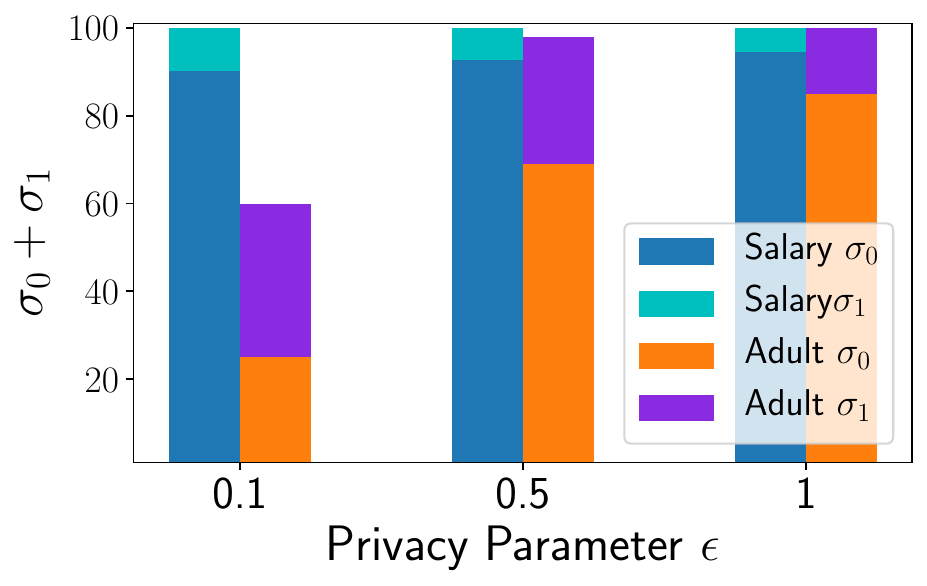}
         \vspace{-0.25cm}
        \caption{Ordinal Queries}
        \label{fig:ldp:ordinal}
    \end{subfigure}
    \begin{subfigure}[b]{0.25\linewidth}
    \centering \includegraphics[width=0.9\linewidth]{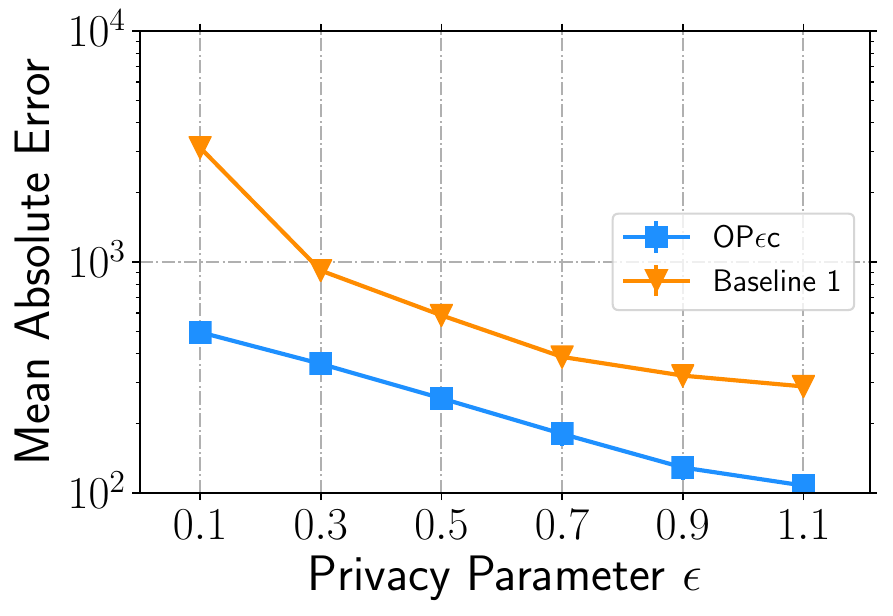}   \vspace{-0.25cm}
 \caption{\textsf{Adult}: Frequency Estimation}
        \label{fig:ldp:freq}\end{subfigure}
    \begin{subfigure}[b]{0.25\linewidth}
    \centering    \includegraphics[width=0.9\linewidth]{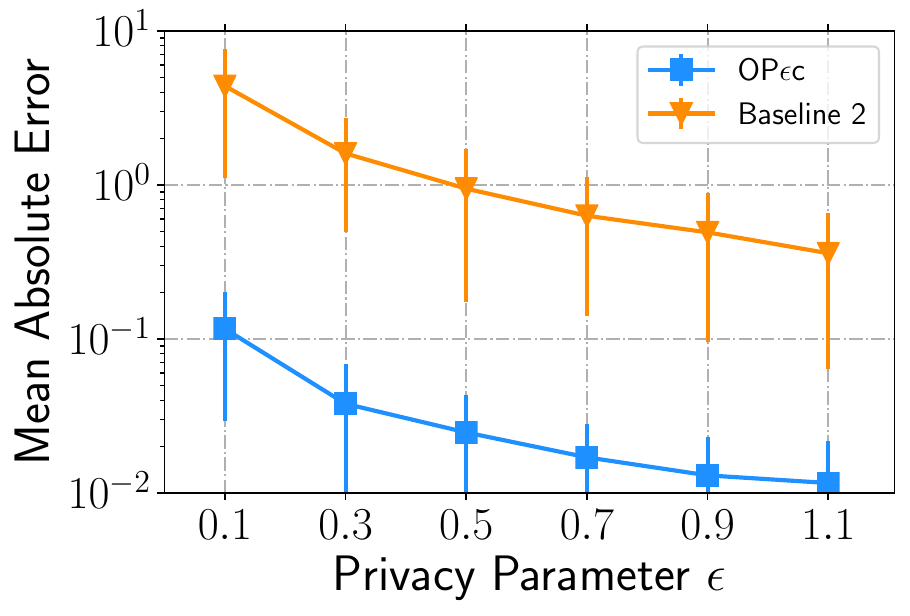}
    \vspace{-0.25cm}
        \caption{\textsf{Adult}: Mean Estimation}
        \label{fig:ldp:mean}\end{subfigure}
      \begin{subfigure}[b]{0.25\linewidth}
    \centering    \includegraphics[width=0.9\linewidth]{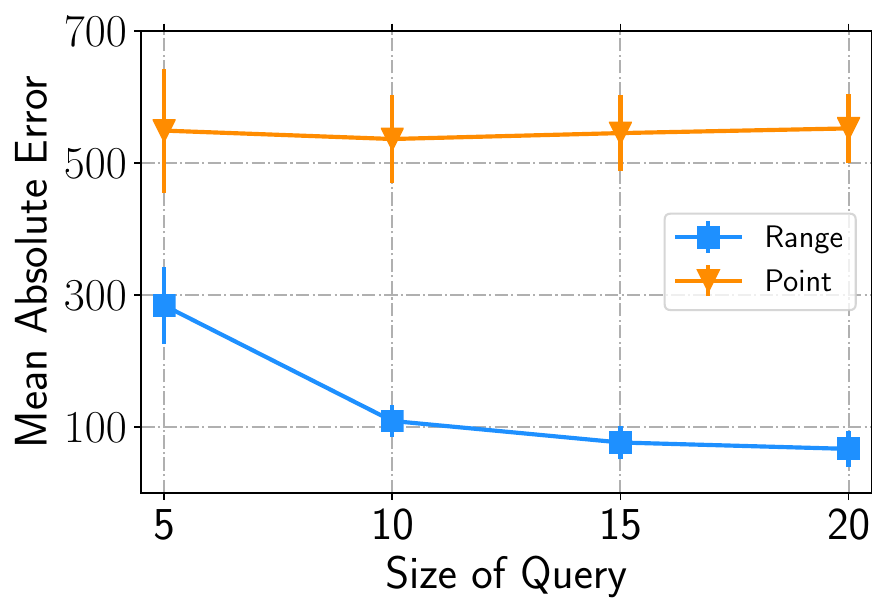}    \vspace{-0.25cm}
        \caption{\textsf{Adult}: Range and Point Queries}
        \label{fig:ldp:range}
    \end{subfigure}\vspace{-0.3cm}
   \caption{Accuracy Analysis of  \name  in the \ldp setting}
   \label{accuracy}\vspace{-0.2cm}
\end{figure*}
\subsubsection{Utility and Overhead of \nam}
In this section, we evaluate \textbf{Q1} and \textbf{Q2} by computing the efficacy of \nam in retrieving the queried records of a range query. Recall, that a \OPE scheme preserves the exact order of the plaintexts. Thus, the loss in accuracy (ordering information) arises solely from \nam's use of the \name primitive. Hence first, we study the effect of the parameters of \name. 
\par We start with the privacy parameter, $\epsilon$  (Figs.  \ref{fig:cdp:eps:1} and \ref{fig:cdp:eps:2}). We observe that $\nam$ achieves high utility even at high levels of privacy.  For example, \nam misses only about $2\%$ of the correct records (i.e., $\rho_M=2\%$) for  \textsf{PUDF} (Fig. \ref{fig:cdp:eps:1}) on average even for a low value of $\epsilon=0.01$ (i.e., the ratio of the output distributions of two datasets that are $100$-adjacent is bounded by $e$). The associated processing overhead is also reasonable -- the size of the extra records retrieved, $\rho_E$, is around $1\%$ of the total dataset size on average. Next, we observe that as the value of $\epsilon$ increases, both the number of missing and extra records drop. For instance, for $\epsilon=1$ we have $\rho_M=0.04\%$, i.e.,  only $4$ in every $10$K correct records are missed on average. Additionally, the number of extra records processed is just $0.3$\% of the total dataset size on average. We observe similar trends for  \textsf{SPARC} (Fig. \ref{fig:cdp:eps:2}) as well. However, the utility for \textsf{SPARC} is lower than that for \textsf{PUDF}. For instance, $\rho_M=10\%$ and $\rho_E=35\%$ for $\epsilon=0.01$ for \textsf{SPARC}. This is so because the ratio of the domain size (120) and partition size (8) for \textsf{SPARC} is smaller than that for \textsf{PUDF} (domain size $\sim 18$K, $|\Pa|=122$). As a result, the individual intervals for \textsf{SPARC} are relatively small which results in lower utility. 
\par Next, we study the effect of the size of the partition (number of intervals) on \nam's utility. As expected, we observe that for \textsf{Adult}, decreasing the partition size from $20$ to $5$ decreases $\rho_M$ from $5\%$ to $0.2\%$ (Fig. \ref{fig:cdp:partition:1}).  However, this increases the number of extra records processed -- $\rho_E$  increases from $0.8\%$ to $7\%$. Similar trends are observed for \textsf{Salary} (Fig. \ref{fig:cdp:partition:2}). \par
Next, we study the effect of including neighboring intervals (Sec. \ref{sec:application})  in Figs. \ref{fig:cdp:neighbor:1} and \ref{fig:cdp:neighbor:2}. For instance, for  \textsf{Salary}, including records from $2$ extra neighboring intervals drops $\rho_M$ from $1\%$ to $0\%$. However, $\rho_E$ increases from $0.4\%$ to $1.7\%$. The increase in $\rho_E$ is more significant for \textsf{Adult} because the domain size for \textsf{Adult} is small and the dataset is dense. On the other hand, \textsf{Salary} has a larger and sparse domain. 
\par Another way for improving utility is to answer a workload of range queries at a time (Sec. \ref{sec:application}).  
We present the empirical results for this in Figs. \ref{fig:cdp:workload:1} and \ref{fig:cdp:workload:2}. For \textsf{SPARC},  we observe that $\rho_M$ and $\rho_E$ drop from $0.9\%$ to $0.1\%$ and $4\%$ to $0.2\%$, respectively as the size of the workload is increased from $1$ to $20$. Further, we note that this effect is more pronounced for \textsf{SPARC}  than for \textsf{PUDF}. This is because, the domain of \textsf{PUDF} is larger and hence, the probability that the queried ranges in the workload are close to each other is reduced.
\vspace{-0.2cm}\subsubsection{ Utility of \name in the \ldp Setting} In this section, we evaluate \textbf{Q3} by studying the utility of the \name primitive in the \ldp setting. \par First, we consider ordinal queries. For \textsf{Adult}, we define an equi-length partition $\Pa=\{[1,10],\cdots,[91,100]\}$ over the domain and our query of interest is --  \textit{Which age group (as defined by $\Pa$) does each data point belong to?} For \textsf{Salary}, we define an equi-depth partition of size $10$ over the domain and our query of interest is -- \textit{Which decile does each data point belong to?} Our results are reported in Fig. \ref{fig:ldp:ordinal}. The first observation is that \name reports the correct encodings with good accuracy. For instance, for $\epsilon=1$, $\sigma_0=94.5\%$ and $\sigma_1=5.5\%$ for the \textsf{Salary} dataset. 
Another interesting observation is that for low values of $\epsilon$, the accuracy for \textsf{Salary} is significantly higher than that for \textsf{Adult}. Specifically, for $\epsilon=0.1$, $\sigma_0=90\%$ and $\sigma_0=35\%$ for \textsf{Salary}  and \textsf{Adult}, respectively. The reason for this is two fold. Firstly, we use an auxiliary dataset for \textsf{Salary} to compute the weighted medians for the central tendencies.  On the other hand, we do not use any auxiliary dataset for \textsf{Adult} and use the median of each interval as the central tendency. Secondly, the domain size of \textsf{Salary} ($230$K) is relatively large compared to the number of intervals $(10)$ which results in higher utility (as explained in Sec. \ref{sec:construction}). \par Fig. \ref{fig:ldp:freq} shows our results for using \name for frequency estimation. \textsf{Baseline1} denotes the state-of-the art $\epsilon$-\ldp frequency oracle \cite{fo}. We observe that \name achieves significantly lower error than \textsf{Baseline1}. For instance, the error of \name is $6\times$ lower than that of \textsf{Baseline1} for $\epsilon=0.1$. This gain in accuracy is due to \name's relaxed $\epsilon$-\ddp guarantee.  From Fig. \ref{fig:ldp:mean}, we  observe that the frequency oracle designed via \name can be used for mean estimation with high utility. Here \textsf{Baseline2} refers to the state-of-the-art protocol \cite{Microsoft} for mean estimation for $\epsilon$-\ldp. We observe that for $\epsilon=0.1$, \name achieves $\sim40\times$ lower error than \textsf{Baseline2}.  Another interesting observation is that \name's frequency oracle gives better accuracy for a range query of size $k$ than $k$ individual point queries (Fig. \ref{fig:ldp:range}). For instance, a range query of size $20$ gives $5\times$ lower error than $20$ point queries. The reason behind this is that the output distribution of \name is the exponential distribution centered at the input $x$. Hence, with high probability $x$ either gets mapped to itself or some other point in its proximity. Thus, the probability for accounting for most copies of $x$ is higher for the case of answering a range query $x \in [a,b]$ than for answering a point estimation for $x$.
\vspace{-0.3cm}
\subsection{An Illustration of \nam's Protection}\vspace{-0.1cm}
Here, we give an illustration of \nam's protection against inference attacks on a real-world dataset.  We use the snapshot\footnote{We use the snapshot attack for the ease of exposition; there are some caveats to its practicality \cite{GRS17}.} attack model (the adversary obtains a one-time copy or snapshot of the encrypted data~\cite{sok}). Our analysis is based on a formal model that captures a generic inference attack in the snapshot model -- we create a bitwise leakage profile for the plaintexts from the revealed order and adversary's auxiliary knowledge as described below.     \\\textbf{Model Description.} 
 We assume the input domain to be discrete and finite, and w.l.o.g denote it as $\mathcal{X}=[0,2^{m-1}]$. Additionally, let $\mathcal{D}$ represent the true input distribution and $X=\{x_1,\cdots,x_n\}$ be a dataset of size $n$ with each data point sampled i.i.d from $\mathcal{D}$. We model our adversary, $\advppt$, to have access to $(1)$ auxiliary knowledge about a distribution, $\mathcal{D}'$, over the input domain, $\mathcal{X}$ and $(2)$ the  ciphertexts, $\mathcal{C}$, corresponding to $X$ which represent the snapshot of the encrypted data store.
The adversary's goal is to recover as many bits of the plaintexts as possible.
Let $X(i), i \in [n]$ represent the plaintext in $X$ with rank~\cite{rank} $i$  and let $X(i,j), j \in [m] $ represent the $j$-th bit for $X(i)$. Additionally, let $b(i,j)$ represent the adversary's guess for $X(i,j)$. Let $\mathcal{L}$ be a $n\times m$ matrix where $\mathcal{L}(i,j)=\mathrm{Pr}\Big[X(i,j)=b(i,j)\big|\mathcal{D},\mathcal{D}'\Big]$ represent the probability that  $\advppt$ correctly identifies the $j$-th bit of the plaintext with rank $i$.
Hence, $\mathcal{L}$ allows analysis of bitwise information leakage from $\mathcal{C}$ to  $\advppt$. The rationale behind using this model is that it \textit{captures a generic inference attack in the snapshot model} allowing analysis at the granularity of bits. \par The exact computation of the bitwise leakage matrix, $\mathcal{L}$, when $X$ is encrypted with \nam and \OPE is given by Thms. \ref{thm:formal:POPE} and \ref{thm:formal:OPE}, respectively,  in \ifpaper full paper~\cite{anom}\else App. \ref{app:snapshot}\fi. Using these theorems, we analytically compute $\mathcal{L}$ for a real-world dataset. We use the age data from the \textsf{Population} dataset for the year $2019$ as the true input distribution, $\mathcal{D}$, and  the data from the year $2011$ is considered to be the adversary's auxiliary distribution, $\mathcal{D}'$. We consider the dataset size to be $200$ and the number of bits considered is $7$ (domain of age is $[1,100]$). Additionally, the partition $\Pa$ for the \name primitive is set to be \scalebox{0.9}{$\Pa=\mathcal{O}=\mathcal{X}=[1,100]$}. The reported privacy parameter $\epsilon$ refers to the $\epsilon$-\indfaocpa guarantee, and implies $\frac{\epsilon}{2}$-\cdp and $\frac{\epsilon}{2}$-\ddp for \nam (Fig. \ref{fig:relation}).  As shown in Fig. \ref{fig:formalmodel}, we observe that the probability of successfully recovering the plaintext bits is significantly lower for  \nam as compared to that of a \OPE. Moreover, the probability of recovering the lower-order bits (bits in the right half) is lower than that of  higher-order bits -- the probability of recovering bits $5$-$7$ is $\approx 0.5$ which is the random guessing baseline. Recall that Thm. \ref{thm:attack:record} implies that the adversary would not be able to distinguish between pairs of inputs that are close to each other. Hence, the above observation is expected since values that are close to each other are most likely to differ only in the lower-order bits.   
Additionally, as expected, the probability of the adversary's success decreases with decreasing value of $\epsilon$. For instance, the average probability of success for the adversary for bit $4$ reduces from $0.77$ in the case of \OPEs (Fig. \ref{fig:formalmodel:L1}) to $0.62$ and $0.51$ for $\epsilon=1$ (Fig. \ref{fig:formalmodel:L2:1}) and $\epsilon=0.1$ (Fig. \ref{fig:formalmodel:L2:2}), respectively, for \nam. Concretely, \textit{no inference attack in the snapshot model that uses the given auxiliary knowledge can distinguish between two age values $(x,x')$ such that \scalebox{0.9}{$|x-x'|\leq 8$} for $\epsilon=0.1$}.
\begin{figure}
  \begin{subfigure}[b]{0.33\linewidth}
    \centering \includegraphics[width=1\linewidth]{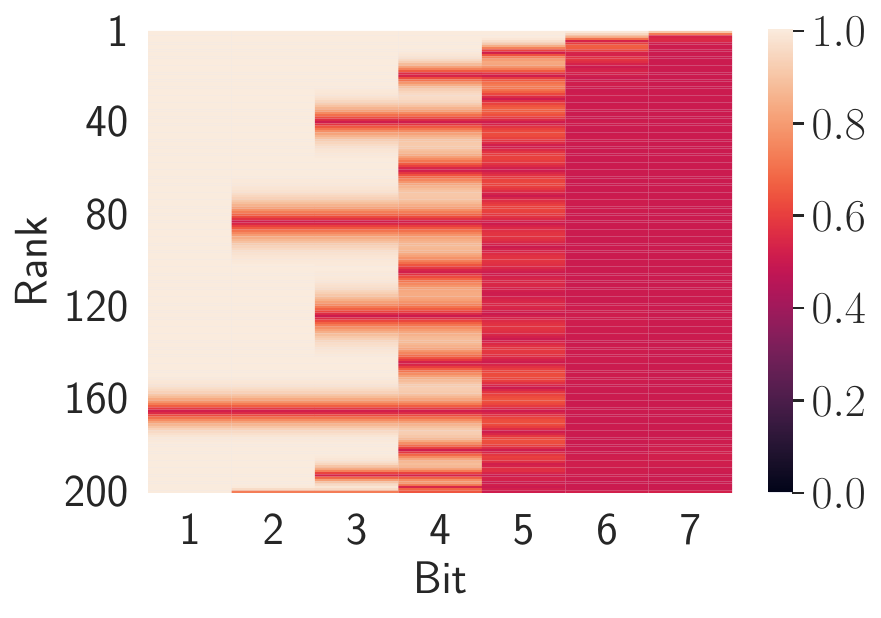}\vspace{-0.3cm}
        \caption{\scalebox{0.8}{$\mathcal{L}$ for \OPE}}
        \label{fig:formalmodel:L1}\end{subfigure}
        \begin{subfigure}[b]{0.33\linewidth}
        \includegraphics[width=1\linewidth]{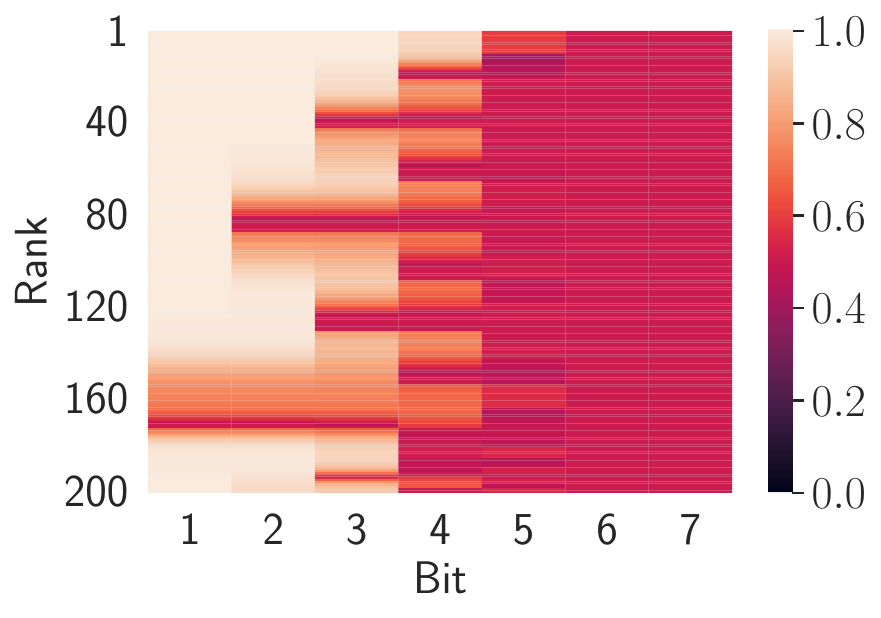}\vspace{-0.3cm}
        \caption{\scalebox{0.8}{$\mathcal{L}$ for \nam for $\epsilon=1$}}
        \label{fig:formalmodel:L2:1}
        \end{subfigure}
        \begin{subfigure}[b]{0.33\linewidth}
        \includegraphics[width=1\linewidth]{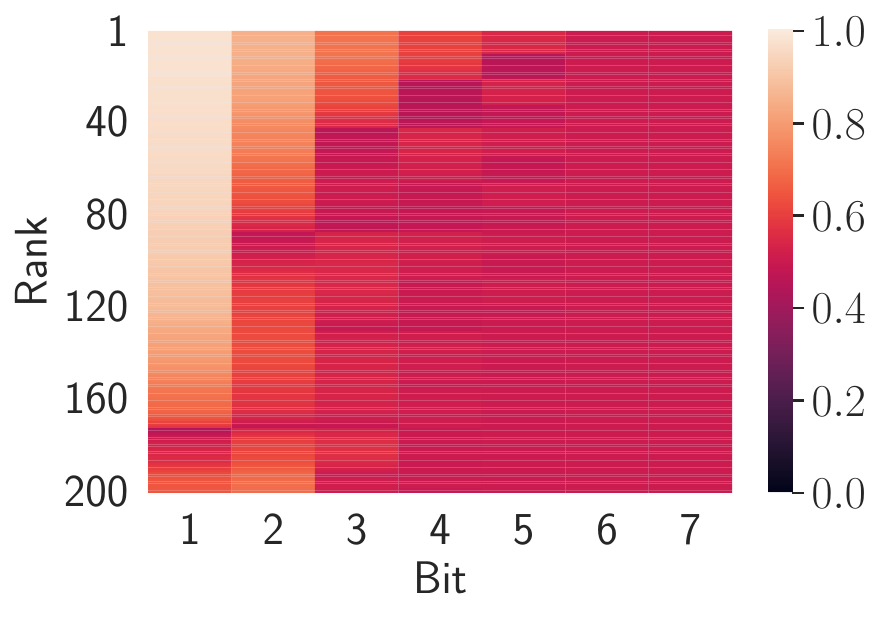}\vspace{-0.3cm}
        \caption{\scalebox{0.8}{$\mathcal{L}$ for \nam for $\epsilon=0.1$}}
        \label{fig:formalmodel:L2:2}
        \end{subfigure}
  \vspace{-0.7cm}
        \caption{\scalebox{0.95}{Numerical Analysis of the Bitwise Leakage Matrix, $\mathcal{L}$}} \label{fig:formalmodel} 
   \end{figure} 
\vspace{-0.7cm}
\section{Related Work}\label{sec:relatedwork}
\textbf{Relaxation of \DP.} \ddp is equivalent to metric-based \ldp \cite{metricDP} where the metric used is $\ell_1$-norm. Further, metric-\ldp is a generic form of Blowfish \cite{Blowfish} and d$_\chi$-privacy \cite{dx} adapted to \ldp. Other works  \cite{linear1,Xiang2019LinearAR, Geo, linear2, Ordinal1, Ordinal2} have also considered metric spaces. 
\\\textbf{\DP and Cryptography.} A growing number of work has been exploring the association between \DP and cryptography. One such line of work proposes to allow a \DP leakage of information for gaining efficiency in cryptographic protocols \cite{DPAP,Shrinkwrap,PSI,DPORam1, DPORam2}. A parallel line of work involves efficient use of cryptographic primitives for differentially private
functionalities \cite{CryptE,EncryptedDP,Prochlo, shuffle2, HH, keyvalue}. Additionally, recent works have combined \DP and cryptography for distributed learning  \cite{kairouz2021distributed, cpSGD, capc} (see \ifpaper full paper \cite{anom}\else App. \ref{app:relatedwork}  \fi).

\vspace{-0.2cm}
\section{Conclusion}\label{sec:extra}
\vspace{-0.1cm}

We have proposed a novel 
$\epsilon$-\cdp 
order preserving encryption scheme, \nam. \nam enjoys a formal guarantee of $\epsilon$-\cdp, in the least, even in the face of inference attacks. To the best of our knowledge, this is the first work to combine \DP with a property-preserving encryption scheme. 

\bibliographystyle{ACM-Reference-Format}
\bibliography{references.bib}
\urlstyle{rm}
\section{Appendix}
\subsection {Background Cntd.} \label{app:background}

\textbf{Notes on \OPE.} Although the notion of \indfaocpa was first introduced by Kerschbaum et al. \cite{IND-FAOCPA}, the proposed definition suffered from a subtle flaw which was subsequently rectified by Maffei et al. \cite{frequencyHiding2}.  The above definition, hence, follows from the one in \cite{frequencyHiding2} (denoted by in $\textsf{IND-FA-OCPA}^*$ in \cite{frequencyHiding2}). Additionally, Defn. \ref{def:OPE} in our paper corresponds to the notion of \textit{augmented} order-preserving encryption scheme (denoted by $\textsf{OPE}^*$ in \cite{frequencyHiding2}) which is crucial for the above security definition. The augmented \OPE scheme is in fact a generalization of the standard \OPE scheme (the only difference being the encryption algorithm $\E$ has an additional input, $\Gamma$). 

\par For a given order $\Gamma = \{\gamma_1, \cdots , \gamma_n\}$, let $\Gamma_{\downarrow i}$ denote the  
denote the order of the sequence $\{\gamma_1, \cdots , \gamma_i\}, i \in [n]$. Note that this order is
unique since $\Gamma$ is already an order. For example, for the randomized sequence
$\Gamma = \{1, 5, 4, 3, 2, 6\}$. Then, $\Gamma_{\downarrow 3}= \{1, 3, 2\}$ which is the order of $\{1, 5, 4\}$. 
Encryption of an input dataset  $X=\{x_1,\cdots,x_n\}$ is carried out as follows:
\squishlistnum 
\item Set $\s_0\leftarrow\K(1^\kappa)$
\item For $\forall i \in [n]$, compute $(\s_{i},y_i)\leftarrow \E(\s_{i-1}, x_i,\Gamma_{\downarrow i})$ \squishendnum
For notational simplicity, we abuse notations and use just $\Gamma$ for step $(2)$ throughout the paper. See Maffei et al. \cite{IND-FAOCPA} for details.\\\\
\textbf{Exponential Mechanism.} A classic \DP mechanism is the exponential mechanism \cite{Dwork} as defined below. 

\begin{definition}[Exponential Mechanism] An $\epsilon$-\DP exponential mechanism   $\mathcal{M}:\mathcal{X}^n\mapsto \mathcal{Y}$ selects an output $y \in \mathcal{Y}$ with probability proportional to  $e^{ \frac{\epsilon u(X,y)}{2\Delta u}}$ where $u:\mathcal{X}^n\times \mathcal{Y}\mapsto \mathbb{R}$ is a utility function mapping input and output pairs with a utility score and $\Delta u$ is the sensitivity defined as \begin{gather*}\Delta u=\max_{y \in \mathcal{Y}} \max_{\mbox{Adjacent pairs} X,X'} \Big|u(X,y)-u(X',y)\Big|   \end{gather*}\end{definition}

\textbf{Composition Theorem for \ddp.} When applied multiple times, the \ddp (and \cdp) guarantee degrades gracefully as follows. 

\begin{theorem}[Sequential Composition] If $\mathcal{M}_1:\mathcal{X}\mapsto \mathcal{Y}$ and $\mathcal{M}_2:\mathcal{X}\mapsto \mathcal{Z}$ are $\epsilon_1$ and $\epsilon_2$-\ddp mechanisms, respectively, then releasing the outputs \scalebox{0.9}{$\langle \mathcal{M}_1(x), \mathcal{M}_2(x) \rangle$} for any input $x \in \mathcal{X}$ satisfies $\epsilon_1+\epsilon_2$-\ddp. \label{thm:seq}
\end{theorem} The proof follows directly from the corresponding proof for standard \DP \cite{Dwork}.
\subsection{Proof of Theorem \ref{thm:faocpae}} \label{app:thm:main}
\begin{proof}\textbf{Intuition.} The intuition of the proof is as follows. Recall that there are four input sequences the adversary has to distinguish among. If the adversary is able to guess bit $b_1$ correctly (with non-trivial probability), it is akin to breaking the \indfaocpa guarantee of \OPEs. Similarly,  if the adversary is able to guess bit $b_2$ correctly, (with non trivial probability) it would imply the violation of the $\epsilon$-\cdp guarantee. \par  The proof is structured as follows. First, we prove that \nam  satisfies $\epsilon/2$-\cdp (or $(\epsilon/2,0)$-\cdp following the notation in Definition \ref{def:cdp}). The rest of the proof follows directly from this result and the $\indfaocpa$ guarantee of the \OPE scheme. \begin{lemma}Let $\mathcal{M}$ be a mechanism that \begin{enumerate}\item inputs a dataset $X \in \mathcal{X}^n$  \item outputs $\widetilde{O}=\{\tilde{o}_1,\cdots,\tilde{o}_n\}$ where for all $i\in [n], \Pa \in \hat{\mathcal{X}}, \tilde{o}_i\leftarrow\name(
x_i,\Pa,\epsilon/2)$  \end{enumerate} Then, $\mathcal{M}$ satisfies $\epsilon/2$-\cdp. \label{lemma:cdp:gen}\end{lemma}
\begin{proof} Let $X, X' \in \mathcal{X}^n$ be $t$-adjacent. Specifically, let $x_i\neq x_i', i \in [n]$. For brevity, we drop $\Pa$ and the privacy parameter $\epsilon/2$ from the notation $\name(\cdot)$. \begin{gather*}\frac{\mathrm{Pr}\big[\mathcal{M}(X)=\widetilde{O}\big]}{\mathrm{Pr}\big[\mathcal{M}(X')=\widetilde{O}\big]}=\frac{\prod_{j=1}^n \mathrm{Pr}\big[\name(x_j)=\tilde{o}_j\big]}{\prod_{j=1}^n\mathrm{Pr}\big[\name(x'_j)=\tilde{o}_j\big]}\\=\frac{\prod_{j=1, j\neq i}^n \mathrm{Pr}\big[\name(x_j)=\tilde{o}_j]}{\prod_{j=1, j\neq i}^n\mathrm{Pr}\big[\name(x_j)=\tilde{o}_j]}\times \frac{\mathrm{Pr}\big[\name(x_i)=\tilde{o}_i]}{\mathrm{Pr}\big[\name(x_i')=\tilde{o}_i\big]}\\\leq e^{\frac{t\epsilon}{2}} [\mbox{ From Eq. \ref{eq:popec:dp} of Definition \ref{def:popec}}]\end{gather*}    This concludes our proof. \end{proof}

\begin{lemma}\nam satisfies $\epsilon/2$-\cdp. \label{lem:cdp}\end{lemma}
This result follows directly from Lemma \ref{lemma:cdp:gen}  from Thm. \ref{thm:post}.\par 
Now, note that $\widetilde{O}_{b_1}\in\mathcal{O}_{b_1}^n$. Thus, $\name(\widetilde{O}_{b_1},\mathcal{O}_{b_1},\infty)=\widetilde{O}_{b_1}$ (Sec. \ref{sec:OPec:def}) . As a result, $\E_{\epsilon}( \textsf{S},\widetilde{O}_{b_1},\Gamma^*,\mathcal{O}_{b_1}, \scalebox{0.8}{$\infty$})$ (Step $4b$) is equivalent to running $\E_{\epsilon}( \textsf{S},X_{b_10},\Gamma^*,\Pa_{b_1}, \epsilon):=\langle \widetilde{O}_{b_1}\leftarrow \name(X_{b_10},\Pa_{b_1},\epsilon/2), \E(\widetilde{O}_{b_1},\s,\Gamma^*)\rangle$. 
Thus, from Defn. \ref{def:faocpa},  if $\Big|\mathrm{Pr}\big[\mathcal{G}^{\advppt}_{\textsf{FA-OCPA}^\epsilon}(\kappa,0,0)=(c_1,c_2)\big] -$\\ $\mathrm{Pr}\big[\mathcal{G}^{\advppt}_{\textsf{FA-OCPA}^\epsilon}(\kappa,1,0)=(c_1,c_2)\big]\Big|\geq \textsf{negl}(\kappa)$, then another \textsf{PPT} adversary, $\advppt'$, can use $\advppt$ to win the  $\mathcal{G}^{\advppt}_{\textsf{IND-FA-OCPA}}(\cdot)$ game which leads to a contradiction.  Hence, we have \begin{gather*}\hspace{-3.5cm}\Big|\mathrm{Pr}\big[\mathcal{G}^{\advppt}_{\textsf{FA-OCPA}^\epsilon}(\kappa,0,0)=(c_1,c_2)\big]\\-\mathrm{Pr}\big[\mathcal{G}^{\advppt}_{\textsf{FA-OCPA}^\epsilon}(\kappa,1,0)=(c_1,c_2)\big]\Big|\leq \textsf{negl}(\kappa)\label{eq:faocpa}\numberthis\end{gather*} Without loss of generality, let us assume \begin{gather*}\hspace{-3.5cm}\mathrm{Pr}\big[\mathcal{G}^{\advppt}_{\textsf{FA-OCPA}^\epsilon}(\kappa,1,0)=(c_1,c_2)\big] \leq\\ \mathrm{Pr}\big[\mathcal{G}^{\advppt}_{\textsf{FA-OCPA}^\epsilon}(\kappa,0,0)=(c_1,c_2)\big]\label{eq:1000}\numberthis\end{gather*} Thus, from Eqs. \eqref{eq:faocpa} and \eqref{eq:1000}, we have
\begin{gather*}\hspace{-3.5cm}\mathrm{Pr}\big[\mathcal{G}^{\advppt}_{\textsf{FA-OCPA}^\epsilon}(\kappa,0,0)=(c_1,c_2)\big]\leq \\\mathrm{Pr}\big[\mathcal{G}^{\advppt}_{\textsf{FA-OCPA}^\epsilon}(\kappa,1,0)=(c_1,c_2)\big]+ \textsf{negl}(\kappa) \label{eq:0010}\numberthis\end{gather*} From Theorem \ref{thm:post} and Lemma \ref{lem:cdp},\begin{gather*}
    \hspace{-3.5cm} \mathrm{Pr}\big[\mathcal{G}^{\advppt}_{\textsf{FA-OCPA}^\epsilon}(\kappa,0,0)=(c_1,c_2)\big]\leq\\\hspace{1.5cm} e^{\frac{t\epsilon}{2}}\mathrm{Pr}\big[\mathcal{G}^{\advppt}_{\textsf{FA-OCPA}^\epsilon}(\kappa,0,1)=(c_1,c_2)\big]\numberthis\label{eq:0001}\\\hspace{-3.5cm}\mathrm{Pr}\big[\mathcal{G}^{\advppt}_{\textsf{FA-OCPA}^\epsilon}(\kappa,0,1)=(c_1,c_2)\big]\leq\\\hspace{1.5cm} e^{\frac{t\epsilon}{2}}\mathrm{Pr}\big[\mathcal{G}^{\advppt}_{\textsf{FA-OCPA}^\epsilon}(\kappa,0,0)=(c_1,c_2)\big]\label{eq:0100}\numberthis
    \\\hspace{-3.5cm}\mathrm{Pr}\big[\mathcal{G}^{\advppt}_{\textsf{FA-OCPA}^\epsilon}(\kappa,1,0)=(c_1,c_2)\big]\leq \\\hspace{1.5cm}e^{\frac{t\epsilon}{2}}\mathrm{Pr}\big[\mathcal{G}^{\advppt}_{\textsf{FA-OCPA}^\epsilon}(\kappa,1,1)=(c_1,c_2)\big]\label{eq:1011}\numberthis\\\hspace{-3.5cm}\mathrm{Pr}\big[\mathcal{G}^{\advppt}_{\textsf{FA-OCPA}^\epsilon}(\kappa,1,1)=(c_1,c_2)\big]\leq\\ \hspace{1.5cm}e^{\frac{t\epsilon}{2}}\mathrm{Pr}\big[\mathcal{G}^{\advppt}_{\textsf{FA-OCPA}^\epsilon}(\kappa,1,0)=(c_1,c_2)\big]\label{eq:1110}\numberthis\end{gather*}
    Now from \eqref{eq:0010} and \eqref{eq:1011}, we have,
    \begin{gather*}\hspace{-3.5cm}\mathrm{Pr}\big[\mathcal{G}^{\advppt}_{\textsf{FA-OCPA}^\epsilon}(\kappa,0,0)=(c_1,c_2)\big]\leq \\e^
    \frac{t\epsilon}{2} \mathrm{Pr}\big[\mathcal{G}^{\advppt}_{\textsf{FA-OCPA}^\epsilon}(\kappa,1,1)=(c_1,c_2)\big]+ \textsf{negl}(\kappa)\label{eq:0011}\numberthis\end{gather*}
    Using Eqs. \eqref{eq:1000} and \eqref{eq:1110}, we have \begin{gather*}\hspace{-3.5cm}\mathrm{Pr}\big[\mathcal{G}^{\advppt}_{\textsf{FA-OCPA}^\epsilon}(\kappa,1,1)=(c_1,c_2)\big]\leq \\e^\frac{t\epsilon}{2}\mathrm{Pr}[\mathcal{G}^{\advppt}_{\textsf{FA-OCPA}^\epsilon}(\kappa,0,0)=(c_1,c_2)\big]\numberthis \label{eq:1100}\end{gather*}
    From Eqs. \eqref{eq:0100} and \eqref{eq:0010}, we have \begin{gather*}\hspace{-3.5cm}\mathrm{Pr}\big[\mathcal{G}^{\advppt}_{\textsf{FA-OCPA}^\epsilon}(\kappa,0,1)=(c_1,c_2)\big]\leq \\e^\frac{t\epsilon}{2}\mathrm{Pr}\big[\mathcal{G}^{\advppt}_{\textsf{FA-OCPA}^\epsilon}(\kappa,1,0)=(c_1,c_2)\big]+ \textsf{negl}'(\kappa) \label{eq:0110}\numberthis \\\mbox{[$\textsf{negl}'(\kappa)=e^{\frac{t\epsilon}{2}} \cdot \textsf{negl}(\kappa)$ which is another negligible function]}\end{gather*}
    Eqs. \eqref{eq:0100} and \eqref{eq:0011} give us \begin{gather*} \hspace{-3.5cm}\mathrm{Pr}\big[\mathcal{G}^{\advppt}_{\textsf{FA-OCPA}^\epsilon}(\kappa,0,1)=(c_1,c_2)\big]|\leq\\ e^{t\epsilon}\mathrm{Pr}[\mathcal{G}^{\advppt}_{\textsf{FA-OCPA}^\epsilon}(\kappa,1,1)=(c_1,c_2)] +\textsf{negl}'(\kappa)\label{eq:0111}\numberthis\end{gather*}
    Using Eqs. \eqref{eq:1000} and \eqref{eq:0001}, we have \begin{gather*}\hspace{-3.5cm}\mathrm{Pr}\big[\mathcal{G}^{\advppt}_{\textsf{FA-OCPA}^\epsilon}(\kappa,1,0)=(c_1,c_2)\big]\\\hspace{1.5cm}\leq e^{\frac{t\epsilon}{2}} \mathrm{Pr}\big[\mathcal{G}^{\advppt}_{\textsf{FA-OCPA}^\epsilon}(\kappa,0,1)=(c_1,c_2)\big] \label{eq:1001}\numberthis\end{gather*}
    Finally, Eqs. \eqref{eq:1110} and \eqref{eq:1001} give us \begin{gather*}\hspace{-3.5cm}\mathrm{Pr}\big[\mathcal{G}^{\advppt}_{\textsf{FA-OCPA}^\epsilon}(\kappa,1,1)=(c_1,c_2)\big]|\leq \\\hspace{1.5cm}e^{t\epsilon} \mathrm{Pr}\big[\mathcal{G}^{\advppt}_{\textsf{FA-OCPA}^\epsilon}(\kappa,0,1)=(c_1,c_2)\big] \label{eq:1101}\numberthis\end{gather*} Note that the $\mathcal{G}^{\advppt}_{\textsf{IND-FA-OCPA}_{\epsilon}}$ game can abort sometimes (Step $4$, when $\widetilde{O}_0$ and $\widetilde{O}_1$ do not share any randomized order). However, this does not lead to any information leakage to $\advppt$ since this step happens before the challenger has chosen any of the bits $\{b_1,b_2\}$. 
Additionally, the condition $1b$ ensures that the event that the game runs to completion happens with non-zero probability. It is so because if $\Pa_{0}(X_{00})$ and $\Pa_{1}(X_{10})$ share a randomized order, then  $\mathrm{Pr}\big[ \widetilde{O}_{0} \mbox{ and } \widetilde{O}_{1} \mbox{ share a randomized order} \big] > 0$.  \par This concludes our proof.
\end{proof}

\subsection{Proof of Thm. \ref{thm:post}}\label{app:thm:post}
\begin{proof}Let $g: \mathcal{Y}\rightarrow \mathcal{Y}'$ be a deterministic function. Now, let us fix two inputs $x,x' \in \mathcal{X}$   and fix an event $\mathcal{Z}\subset \mathcal{Y}'$. Let $\mathcal{W}=\{y \in \mathcal{Y}| f(y) \in Z\}$. We then have,
\begin{align*}\mathrm{Pr}\big[g(\mathcal{M}(x))\in \mathcal{Z}\big] &=  \mathrm{Pr} \big[\mathcal{M}(x)\in \mathcal{W}\big]\\ &\leq e^{\epsilon|x-x'|}\cdot \mathrm{Pr}\big[\mathcal{M}(y)\in \mathcal{W}\big]\\ &=e^{\epsilon|x-x'|}\cdot \mathrm{Pr}\big[g(\mathcal{M}(x))\in \mathcal{Z}
\big]\end{align*}
This concludes our proof because any randomized mapping
can be decomposed into a convex combination of deterministic functions, and a convex combination of differentially private (equivalently \ddp) mechanisms is
differentially private (\ddp) \cite{Dwork}.

The proof for \cdp follows similarly.
 \end{proof}

\subsection{Proof of Thm. \ref{thm:opec}}\label{app:thm:construction}
\begin{proof} Here, we need to prove that Alg. \ref{alg:1} satisfies the Eq. \ref{eq:popec:order} and \ref{eq:popec:dp} ($\epsilon$-\ddp) from Defn. \ref{def:popec}. We do this with the help of the following two lemmas. 
\begin{lemma}Alg. \ref{alg:1} satisfies Eq. \ref{eq:popec:order} from Defn. \ref{def:popec}. \label{lemma:1}\end{lemma}
\begin{proof}
Let $x \in \mathcal{X}_i, i \in [k]$. \\\textbf{Case I.} $d_j, 1\leq j < i-1, i \in [2,k]$
In this case, we have $d_{j}<d_{i-1}$. 
Thus, \begin{gather}\mathrm{Pr}\big[\name(x,\Pa,\epsilon)=o_{i-1}\big]>\mathrm{Pr}\big[\name(x,\Pa,\epsilon)=o_{j}\big]\end{gather}
\\\textbf{Case II.} $d_j \mbox{ s.t. } i+1 < j \leq k, i \in [k-2]$ \\
In this case, we have $d_{i+1}<d_{j}$. Thus, \begin{gather}\mathrm{Pr}\big[\name(x,\Pa,\epsilon)=o_{i+1}\big]>\mathrm{Pr}\big[\name(x,\Pa,\epsilon)=o_{j}\big]\end{gather} Clearly, this concludes our proof. \end{proof}
Next, we prove that Alg. \ref{alg:1} satisfies $\epsilon$-\ddp.
\begin{lemma}Alg. \ref{alg:1} satisfies $\epsilon$-\ddp. \label{lemma:2}
\end{lemma}
\begin{proof} For all $x\in \mathcal{X}$ and $o_i \in \mathcal{O}=\{o_1,\cdots,o_k\}$, we have
\begin{gather*}\hspace{-4cm}\frac{\mathrm{Pr}\big[\name(x,\Pa,\epsilon)=o_i]}{\mathrm{Pr}\big[\name(x+t,\Pa,\epsilon)=o_i\big]}= \\ \hspace{2cm}\Big(e^{|x+t-d_i|-|x-d_i|\cdot \epsilon/2} \cdot\frac{\overset{k}{\underset{j=1}{\sum}}e^{-|x+t-d_j|\cdot \epsilon/2}} {\overset{k}{\underset{j=1}{\sum}}e^{-|x-d_j|\cdot \epsilon/2}} \Big)\\\leq e^{t\epsilon/2}\cdot e^{t\epsilon/2} \\\hspace{0.2cm}\big[\because  |x-d_j|-t\leq|x+t-d_j|\leq |x-d_j|+t\big] \\= e^{t\epsilon }  \label{eq:DP2} \numberthis\end{gather*}
Similarly, \begin{gather*}\frac{\mathrm{Pr}\big[\name(x,\Pa,\epsilon)=o_i]}{\mathrm{Pr}\big[\name(x+t,\Pa,\epsilon)=o_i\big]}\geq e^{-t\epsilon}
\end{gather*}
\end{proof}
Hence, from Lemmas \ref{lemma:1} and  \ref{lemma:2},  we conclude that  Alg. \ref{alg:1} gives a construction for the \name primitive. 
\end{proof}
\subsection{Illustration of Alg. \ref{alg:1}}\label{app:alg:1}
Here, we illustrate Alg. \ref{alg:1} with an example. We illustrate the algorithm with the following example.  Consider a partition $\Pa=\{[1,20],[21,80],[81,100]\}$ for the domain $\mathcal{X}=\langle 1,\cdots,100\rangle$ and let $\mathcal{O}=\{1,2,3\}$ denote the set of its corresponding encodings.  Let us assume the a uniform prior, $\mathcal{D}$ (default value), on $\mathcal{X}$. Thus, in Stage I, median is our measure from central tendency which gives  $d_1=10.5, d_2=50.5$ and $d_3=90.5$. 
\par In Stage II (Steps $4$-$9$), the encoding probability distributions are computed using a variant of the classic exponential mechanism \cite{Ordinal1,Dwork} (Eq. \eqref{eq:pivot1}).  For instance, \scalebox{0.9}{$\mathrm{Pr}\big[\name(40,\Pa,\epsilon)=2\big]=p_{40,2} \propto 1/e^{(40-d_2)\epsilon/2}= 1/e^{5.25\epsilon}$}. The final encoding is then sampled from $p_x$ (Steps $10$-$11$). 
\subsection{Proof for Thm. \ref{thm:attack:record}}\label{app:attack}
\textbf{Theorem 7.} For a \nam scheme satisfying $\frac{\epsilon}{2}$-\ddp, we have  \begin{gather}\Big|\mathrm{Pr}\big[p'=p\big]-\mathrm{Pr}\big[rand=p\big]\Big|\leq \frac{e^{\epsilon^{*}}}{q+e^{\epsilon^{*}}}-\frac{1}{q+1} \label{Eq:attack}\end{gather} where $\epsilon^*=\epsilon\cdot \lceil\beta N\rceil$ and $q=|Q(x_0)|$ (Step $(2)$ of the $\mathcal{G}^{\mathcal{A}}_{\beta-RI}$ game).\\
\begin{proof}Let $y$ denote the output ciphertext (Step $4$) observable to the adversary $\mathcal{A}$.  Note that the game itself satisfies $\epsilon/2$-\ddp.
Let $d$ be the probability that the adversary $\mathcal{A}$ wins the game, i.e., $d:=\mathrm{Pr}\big[p'=p\big]$. Clearly, this cannot be greater than $\mathrm{Pr}\big[\nam(\s,\epsilon/2,x_p)=y\big]$ (we use this shorthand to refer to the encryption as defined in Def. \ref{def:pope}) -- the probability that encrypting $x_p$ under \nam actually outputs $y$.  Let $\mathcal{I}=\{i|i \in \{0,\cdots,q\}, i\neq p\}$. Since $|x_i - x_j| \leq 2\lceil\beta N\rceil, i,j \in \{0,\cdots, q\}$, from the $\epsilon/2$-\ddp guarantee we have \begin{gather*} \forall i \in \mathcal{I}\\ \mathrm{Pr}\big[\nam(\s,\epsilon/2,x_p)=y\big] \leq e^{\epsilon^{*}}\mathrm{Pr}\big[\nam(\s,\epsilon/2,x_i)=y\big] \\ d \leq e^{\epsilon^{*}}\mathrm{Pr}\big[\nam(\s,\epsilon/2,x_i)=y\big] \end{gather*}  Summing the equations for all $i \in \mathcal{I}$, we have \begin{gather*}q\cdot d \leq e^{\epsilon^*}\sum_{i \in \mathcal{I}} \mathrm{Pr}\big[\nam(\s,\epsilon/2,x_i)=y\big] \\\Rightarrow  q\cdot d \leq e^{\epsilon^*}(1-d) \\\Rightarrow  d \leq \frac{e^{\epsilon^*}}{q+e^{\epsilon^*}}\numberthis \label{eq:help1} \end{gather*} Clearly, \begin{gather} \mathrm{Pr}\big[rand=p\big]=\frac{1}{q+1} \label{eq:help2}\end{gather} Hence, from Eqs. \ref{eq:help1} and \ref{eq:help2}, we have \begin{gather*}\Big|\mathrm{Pr}\big[p'=p\big]-\mathrm{Pr}\big[rand=p\big]\Big|\leq \frac{e^{\epsilon^{*}}}{q+e^{\epsilon^{*}}}-\frac{1}{q+1} \end{gather*}\end{proof}

\subsection{Frequency Estimation Using \name} \label{app:freq:utility} 
\begin{algorithm}
\small
\caption{Frequency Estimation}\label{alg:freqestimation}
\begin{algorithmic}[1]
\Statex \textbf{Input:} $X$ -  Input dataset $\langle x_1,\cdots,x_n\rangle$; $\epsilon$ - Privacy parameter
\Statex \textbf{Output:} $\mathbf{X}$ - Estimated frequency 
\Statex \textbf{Data Owner}
\State Set $\mathcal{P}=\mathcal{X}$ 
\State \textbf{for} $i \in [n]$
\State \hspace{0.2cm} $\textsf{DO}_i$ computes \scalebox{0.9}{$\tilde{o}_i=\name(x_i,\mathcal{X},\epsilon)$} and sends it to the aggregator 
\State \textbf{end for}
\Statex \textbf{Data Aggregator}
\State Data aggregator performs \textsf{NNLS} optimization as follows 
\begin{gather*}
\mathbf{A}\cdot \mathbf{X} = \mathbf{Y} \mbox{ where}\\\mathbf{A}(i,j)=\mathrm{Pr}[\name(i,\mathcal{X},\epsilon)=j], i,j \in [m]\\\mathbf{Y}
(i)=\mbox{Count of value }i \mbox{ in }\{\tilde{o}_1,\cdots,\tilde{o}_n\} \end{gather*}
\State \textbf{Return} \textbf{\textsf{X}}
\end{algorithmic}
\end{algorithm}
Here we describe Alg. \ref{alg:freqestimation}.
Given a privacy parameter, $\epsilon$, each data owner, $\textsf{DO}_i, i \in [n]$, reports  $\tilde{o}_i=\name(x_i,\mathcal{X},\epsilon)$  to the untrusted data aggregator (Steps 1-4). Next, the data aggregator  performs non-negative least squares (\textsf{NNLS})  as a post-processing inferencing step on the noisy data to compute the final frequency estimations (Steps 5-6).  \textsf{NNLS} is a type of constrained least squares optimizations problem where the coefficients are not allowed to become negative. That is, given a matrix $\mathbf{A}$ and a (column) vector of response variables $\mathbf{Y}$, the goal is to find $\mathbf{X}$ such that \begin{gather*}
\arg \min_\mathbf{X} \|\mathbf {A\cdot X} -\mathbf {Y} \|_{2}, \mbox{subject to } \mathbf{X} \geq 0
\end{gather*}
where $||\cdot||_2$ denotes Euclidean norm. 
The rationale behind this inferencing step (Step 5) is discussed below. \vspace{-0.1cm}
\begin{lemma}  W.l.o.g let $\mathcal{X}=\{1,\cdots,m\}$ and let $\mathbf{Y}$ be the vector such that $\mathbf{Y}(i), i \in [m]$ indicates the count of value $i$ in the set $\{\tilde{o}_1,\cdots,\tilde{o}_n\}$ where $\tilde{o}_i=\name(i,\mathcal{X},\epsilon)$.  Given,  \begin{gather}
\mathbf{A}(i,j)=\mathrm{Pr}\big[\name(i,\mathcal{X},\epsilon)=j\big], i,j \in [m]\end{gather}   the solution $\mathbf{X}$ of $\mathbf{A}\cdot \mathbf{X} = \mathbf{Y}$ gives an unbiased frequency estimator ($\mathbf{X}(i)$ is the unbiased estimator for value $i$).\label{lem:unbiased}\vspace{-0.1cm}\end{lemma}
The proof of the above lemma is presented in App. \ref{app:lem:unbiased}. Thus by the above lemma, $\mathbf{X}$ is an unbiased frequency estimator. However, it is important to note that the solution $\mathbf{X}$ is not guaranteed to be non-negative.  But, given our problem setting, the count estimates are constrained to be non-negative. Hence, we opt for an $\textsf{NNLS}$ inferencing. When the exact solution $\mathbf{X}=\mathbf{A}^{-1}\cdot \mathbf{Y}$ is itself non-negative, the estimator obtained from the \textsf{NNLS} optimization is identical to the exact solution. Otherwise, the \textsf{NNLS} optimization gives a biased non-negative estimator that results in minimal least square error. 

\subsection*{Utility Analysis for Frequency Estimation Using \name}
Here, we present a formal utility analysis of the frequency oracle. 
Let $p_{ij}, i,j \in \mathcal{X}$ denote the $\mathrm{Pr}\big[\name(i,\mathcal{X},\epsilon)\big]=j$ and let $\mathbf{X}'[i]$ denote the true count for $i$. Additionally, let $\mathbb{I}_{i,j}$ be an indicator variable for the event $\name(i,\mathcal{X},\epsilon)=j$.
\begin{theorem}The variance of count estimation $\mathbf{X}[i]$  is given by  \begin{gather*}\textsf{Var}(\mathbf{X}[i])=\sum_{j=1}^n\sum_{k=1}^{n}\big(p_{k,j}(1-p_{k,j})\cdot (\mathbf{X}'[k]\cdot\mathbf{A}^{-1}[i,j])^2 \big) -
\\\sum_{k=1}^n\sum_{1\leq j_1 < j_2 \leq n} \Big(\mathbf{X}'[k]^2\cdot\mathbf{A}^{-1}[i,j_1]\cdot \mathbf{A}^{-1}[i,j_2]\cdot p_{k,j_1}\cdot p_{k,j_2}\Big)\end{gather*}\end{theorem}
\begin{proof} Variance of the indicator variable is given by,
\begin{gather*}
\textsf{Var}(\mathbb{I}_{j,i}) = p_{j,i}\cdot q_{j,i}
\end{gather*}
Additionally, we so have\begin{gather*}\textsf{Cov}(\mathbb{I}_{j,i},\mathbb{I}_{j,k
})=-p_{j,i}p_{j,k}\cdot \\\textsf{Cov}(\mathbb{I}_{j,i},\mathbb{I}_{k,l})=0\end{gather*}Using this we have, \begin{gather*}\textsf{Var}(\mathbf{X}[i])=\textsf{Var}\Big(\sum_{j=1}^n\mathbf{A}^{-1}[i,j])\cdot \mathbf{Y}[j]\Big)\\=\textsf{Var}\Big(\sum_{j=1}^n\Big(\mathbf{A}^{-1}[i,j]\cdot \sum_{k=1}^n\mathbf{X}'[k]\cdot\mathbb{I}_{k,j}\Big)\Big)\\=\textsf{Var}\Big(\sum_{j=1}^n\sum_{k=1}^{n}\big(\mathbb{I}_{k,j}\cdot \mathbf{X}'[k]\cdot \mathbf{A}^{-1}[i,j] \big)\Big)\\=\sum_{j=1}^n\sum_{k=1}^{n}\big(\textsf{Var}(\mathbb{I}_{k,j})\cdot (\mathbf{X}'[k]\cdot\mathbf{A}^{-1}[i,j])^2 \big) + \\\sum_{k=1}^n\sum_{1\leq j_1 < j_2 \leq n} \Big(\mathbf{X}'[k]^2\cdot\mathbf{A}^{-1}[i,j_1]\cdot \mathbf{A}^{-1}[i,j_2]\cdot \textsf{Cov}(\mathbb{I}_{k,j_1},\mathbb{I}_{k,j_2})\Big) \\=\sum_{j=1}^n\sum_{k=1}^{n}\big(p_{k,j}(1-p_{k,j})\cdot (\mathbf{X}'[k]\cdot\mathbf{A}^{-1}[i,j])^2 \big) -
\\\sum_{k=1}^n\sum_{1\leq j_1 < j_2 \leq n} \Big(\mathbf{X}'[k]^2\cdot\mathbf{A}^{-1}[i,j_1]\cdot \mathbf{A}^{-1}[i,j_2]\cdot p_{k,j_1}\cdot p_{k,j_2}\Big)\end{gather*}

\end{proof}

\subsection{Proof of Lemma \ref{lem:unbiased}}\label{app:lem:unbiased}
\begin{proof} Let $\mathbf{X}'$ be a vector such that $\mathbf{X}'(i)$ represents the true count of the value $i\in [m]$. Thus, we have \begin{align*}\mathbb{E}\big[X(i)\big]&=\mathbb{E}\big[\sum_{j=1}^{n}\mathbf{A}^{-1}(i,j)\cdot \mathbf{Y}(j)\big]\\&=\sum_{j=1}^{n}\mathbf{A}^{-1}(i,j)\cdot\mathbb{E}\big[\mathbf{Y}(j))\big]\\&=\sum_{j=1}^{n}\mathbf{A}^{-1}(i,j)\cdot\big(\sum_{k=1}^n\mathbf{X}'(k)\cdot \mathrm{Pr}\big[\name(k,\mathcal{X},\epsilon)=j\big]\big)\\&=\sum_{j=1}^{n}\mathbf{A}^{-1}(i,j)\cdot\big(\sum_{k=1}^n\mathbf{X}'(k)\cdot\mathbf{A}(j,k)\big)\\&=\sum_{j=1}^{n}\mathbf{X}'(j)\cdot(\sum_{k=1}^{n}\mathbf{A}^{-1}(i,k)\cdot\mathbf{A}(k,j))\\&=\mathbf{X}'(i)\end{align*}This concludes the proof.
\end{proof}

 \subsection{Bitwise Leakage Matrix}\label{app:snapshot}\vspace{-0.1cm}    Most of the inference attacks are inherent to any \OPE scheme -- they do not leverage any weakness in the cryptographic security guarantee of the schemes but instead utilize the ordering information of the plaintexts that is revealed by definition. In other words, these attacks are \textit{beyond the scope of the standard cryptographic security guarantees} (such as \textsf{IND-FA-OCPA}) for \OPEs and hence, their effects are \textit{not} captured directly by the cryptographic security guarantees.   Hence, in this section, we present a formal model to systematically study the effect of such inference attacks at the granularity of the bits of the plaintext. The goal is to provide \textit{an intuitive insight into \nam's improved protection} against these attacks. 
 \par For this, we concentrate on the ``snapshot" attack model (the adversary only obtains a onetime copy or snapshot of the encrypted data store \cite{sok}) for the ease of exposition. The proposed formal model creates a privacy leakage profile for the plaintext bits based on the revealed ordering information and adversary's auxiliary knowledge.
 This model \textit{captures a generic inference attack in the snapshot model}. Through this analysis, we demonstrate \nam's efficacy in thwarting inference attacks on a real-world dataset (Fig. \ref{fig:formalmodel}). 
 \\\\
\textbf{Model Description.}
 We assume the input domain to be discrete, and finite and w.l.o.g denote it as $\mathcal{X}=[0,2^{m-1}]$. Additionally, let $\mathcal{D}$ represent the true input distribution and $X=\{x_1,\cdots,x_n\}$ be a dataset of size $n$ with each data point sampled i.i.d from $\mathcal{D}$. We model our adversary, $\advppt$, to have access to the following data:
\squishlist\item Auxiliary knowledge about a distribution, $\mathcal{D}'$, over the input domain, $\mathcal{X}$. In practice, this can be derived from domain knowledge or auxiliary (public) datasets known to the adversary. \item The  ciphertexts, $\mathcal{C}$, corresponding to $X$ which represent the snapshot of the encrypted data store.
\squishend
The adversary's goal is to recover as many bits of the plaintexts as possible.
Let $X(i), i \in [n]$ represent the plaintext in $X$ with rank~\cite{rank} $i$  and let $X(i,j), j \in [m] $ represent the $j$-th bit for $X(i)$. Additionally, let $b(i,j)$ represent the adversary's guess for $X(i,j)$. Let $\mathcal{L}$ be a $n\times m$ matrix where $\mathcal{L}(i,j)=\mathrm{Pr}\Big[X(i,j)=b(i,j)\big|\mathcal{D},\mathcal{D}'\Big]$ represent the probability that  $\advppt$ correctly identifies the $j$-th bit of the plaintext with rank $i$.
Hence, the matrix $\mathcal{L}$ helps us in analyzing the bitwise information leakage from $\mathcal{C}$ to  $\advppt$.
\\\\\textbf{Adversary's Approach.}  $\advppt$'s goal is to produce their best guess for $X(i,j)$.  Given $\advppt$'s auxiliary knowledge about the distribution, $\mathcal{D}'$, the strategy  with   the least probabilistic error is as follows:
\begin{gather} \hspace{-0.3cm}b(i,j)=\arg\max_{b\in \{0,1\}}\big\{\mathrm{Pr}_{\mathcal{D}'}\big[X(i,j)=b\big]\big\}, i\in[n], j\in[m] \label{eq:bit}
\end{gather} 
Next, we formalize $\mathcal{L}$ when $X$ is encrypted under \nam. 
\begin{theorem} 
If $X$ is encrypted under  \nam, then for all $i\in [n], j\in [m]$ we have 
\begin{gather*}
\hspace{-0.4cm}\mathcal{L}(i,j)=\sum_{s \in S_{b(i,j)}^{j}}\mathrm{Pr}_{\mathcal{D}}\big[x=s\big]\Big(\sum_{v \in \mathcal{O}} \mathrm{Pr}_{\mathcal{D}^*}\big[\widetilde{O}(i)=v\big]\cdot
\\ \hspace{0.5cm}\mathrm{Pr}\big[\name(s,\Pa,\epsilon)=v\big]/\mathrm{Pr}_{\mathcal{D}^*}\big[o=v\big]\ \Big) +\textsf{negl}(\kappa)\numberthis\label{eq:thm:formal:OPe}
\end{gather*} where $\widetilde{O}(r)$ denotes the encoding with rank $r, r \in [n]$, $\Pa \in \hat{\mathcal{X}}, o \sim \mathcal{D}^*$, and $\mathcal{D}^*:\mathcal{X}\mapsto\mathcal{O}$ represents the distribution of the encoding $\name(x,\Pa,\epsilon),~x\sim\mathcal{D}$ which is given as \begin{gather*}\mathrm{Pr}_{\mathcal{D}^*}\big[v\big]=\sum_{x\in\mathcal{X}}\mathrm{Pr}_{\mathcal{D}}\big[x\big]\cdot 
\mathrm{Pr}\big[\name(x,\Pa,\epsilon)=v
\big], v\in \mathcal{O}\end{gather*} \label{thm:formal:POPE}\vspace{-0.3cm} \end{theorem}
Next, we formalize $\mathcal{L}$ when $X$ is encrypted under a \OPE scheme.
\begin{theorem} If $X$ is encrypted under an \OPE scheme that satisfies the \indfaocpa guarantee, then  for all $i\in [n], j\in [m]$ we have \begin{gather*}\hspace{-1.8cm}\mathcal{L}(i,j)=\sum_{ s \in \mathcal{S}^{j}_{b(i,j)}}\mathrm{Pr}_{\mathcal{D}}\big[X(i)=s\big] + \textsf{negl}(\kappa)\label{eq:thm:formal:OPE}\numberthis \end{gather*}   where $x\sim \mathcal{D}$,  and  $\mathcal{S}^j_{b}=\{s| s \in \mathcal{X} \mbox{ and its }j \mbox{-th bit } s^j=b\}$.\label{thm:formal:OPE}\end{theorem}
 The above theorem formalizes what  $\advppt$ can learn from just the order of the plaintexts (that is leaked by $\mathcal{C}$ by definition). The proofs of both the theorems are presented in the following section.

\textbf{Note.} For Fig. \ref{fig:formalmodel}, we omit the $\textsf{negl}(\kappa)$ term from Eqs. \ref{eq:thm:formal:OPe} and \ref{eq:thm:formal:OPE}. 
\begin{tcolorbox}[sharp corners]\vspace{-0.2cm} \textbf{Remark.} The bitwise leakage matrix, $\mathcal{L}$, captures the efficacy of \textit{a generic inference attack in the snapshot model} at the granularity of the plaintext bits. We present this analysis to provide an \textit{intuitive insight} into \nam's improved protection against inference attacks (given formally by Thm. \ref{thm:attack:record}).   \vspace{-0.2cm}\end{tcolorbox} \vspace{-0.1cm}
\subsection{Proofs for Thms. \ref{thm:formal:OPE} and \ref{thm:formal:POPE}}\label{app:thm:fm:OPE}

\textbf{Preliminaries.}\\
Recall, the adversaries strategy of guessing the plaintext bits is given by \begin{gather*} b(i,j)=\arg \max_{b\in \{0,1\}}\big\{\mathrm{Pr}_{\mathcal{D}'}\big[X(i,j)=b\big]\big\}, i\in[n], j\in[m] 
\\= \left\{\begin{array}{ll}
                 0 \mbox{ if } \mathbb{E}_{\mathcal{D}'}\big[X(i,j)\big]\leq 1/2\\
                  1 \mbox{ if } \mathbb{E}_{\mathcal{D}'}\big[X(i,j)\big]>1/2
                \end{array}
              \right. \numberthis   \label{eq:bit}\end{gather*} 
\textbf{Fact 1.} If $\mathcal{D}$ represents an input distribution and $X=\{x_1,\cdots,x_n\}$ denotes a dataset of size $n$ with each data point  sampled i.i.d from $\mathcal{D}$, then we have: \begin{gather*}\mathrm{Pr}_{\mathcal{D}}\big[X(i,j)=b\big]=\sum_{s \in\mathcal{S}^j_b }\mathrm{Pr}_{\mathcal{D}}\big[X(i)=s\big]\end{gather*} where $i\in [n], j \in [m], b \in \{0,1\}$, and $\mathcal{S}^j_{b}=\{s|s\in \mathcal{X} \mbox{ and its } j\mbox{-th bit } s^j=b$\}. 
\subsection*{Proof of Theorem \ref{thm:formal:OPE}}
\begin{proof}
Let $\mathcal{C}(i)$ represent the ciphertext with rank $i$ in $\mathcal{C}$. Additionally, let $X'(i)$ represent the corresponding plaintext for $\mathcal{C}(i)$. From the $\indfaocpa$ guarantee, we observe that the rank of a ciphertext $y \in Y$ is equal to the rank of its corresponding plaintext in $X$, i.e, $X'(i)=X(i)$.  Thus, we have
this, we have
\begin{gather*}\mathcal{L}(i,j)=\mathrm{Pr}_{\mathcal{D}}\big[X'(i,j)=b(i,j)\big]+\textsf{negl}(\kappa) \\\mbox{[The term  \textsf{negl}($\kappa$) accounts for the corresponding term in Eq. \ref{eq:indocpa}}\\\mbox{ for the \indocpa guarantee of the \OPE scheme.]} \\=\mathrm{Pr}_{\mathcal{D}}\big[X(i,j)=b(i,j)\big]\\=\sum_{ s \in \mathcal{S}^{j}_{b(i,j)}}\mathrm{Pr}_{\mathcal{D}}\big[X(i)=s\big]\mbox{ [From Fact 1] }\numberthis \label{eq:lem:formal:OPE} \end{gather*}
 \end{proof}
Eqs. \ref{eq:bit} and \ref{eq:lem:formal:OPE} can be numerically computed using the following lemma.
\begin{lemma}If $\mathcal{D}$ represents an input distribution and $X=\{x_1,\cdots,x_n\}$ denotes a dataset of size $n$ with each data point  sampled i.i.d from $\mathcal{D}$, then we have:  \begin{gather*}\mathrm{Pr}_{\mathcal{D}}\big[X(i)=s\big]=\left\{\begin{array}{ll}\hspace{-0cm
}\scalebox{0.95}{$\overset{n}{\underset{j=n-i+1}{\sum}}\hspace{-0.2cm}\binom{n}{j}\cdot\mathrm{Pr}_{\mathcal{D}}\big[x<s\big]^{n-j}\hspace{-0.2cm}\cdot\mathrm{Pr}_{\mathcal{D}}\big[x=s\big]^{j}$}\\\hspace{2.5cm}\mbox{ \textbf{if} } \mathrm{Pr}_{\mathcal{D}}\big[x>s\big]=0 \\\overset{n}{\underset{j=i}{\sum}}\binom{n}{j}\cdot\mathrm{Pr}_{\mathcal{D}}\big[x=s\big]^j\cdot\mathrm{Pr}_{\mathcal{D}}\big[x>s\big]^{n-j}\\\hspace{2.5cm}\mbox{ \textbf{if} } \mathrm{Pr}_{\mathcal{D}}\big[x<s\big]=0\\\overset{n}{\underset{j=1}{\sum}}\Bigg(\overset{\min\{i,n-j+1\}}{\underset{k=\max\{1,i-j+1\}}{\sum}}\Big(\binom{n}{k-1,j,n-k-i+1}\cdot\\\hspace{1cm}\mathrm{Pr}_{\mathcal{D}}\big[x<s\big]^{k-1}\cdot\mathrm{Pr}_{\mathcal{D}}\big[x=s\big]^j\cdot\\\mathrm{Pr}_{\mathcal{D}}\big[x>s\big]^{n-k-j+1}\Big)\Bigg)\mbox{ \textbf{otherwise}}\end{array}\right.\end{gather*} where $x \sim \mathcal{D}, i \in [n]$ and $s \in \mathcal{X}$.\end{lemma}
\begin{proof}
Let $X_{sort}$ denote the sorted version of $X$. Additionally, 
let $r^f_s$ and $r^l_s$ denote the positions of the first and last occurrences of the value $s$ in $X_{sort}$, respectively. Let $cnt_s$ denote the count of data points with value $s$ in $X$. Thus, clearly $cnt_s=r^l_s-r^f_s+1$  \\
\textbf{Case I:} $\mathrm{Pr}_{\mathcal{D}}\big[x>s\big]=0$\\
In this case, we have \begin{gather*}X(i)=s  \implies X(r)=s, \forall r \mbox{ s.t } i \leq r \leq n\end{gather*} Thus, $r^l_s=n$ and $n-i+1 \leq cnt_s\leq n$ and  \begin{gather*}\mathrm{Pr}_{\mathcal{D}}\big[X(i)=s\big]=\sum_{j=n-i+1}^n\mathrm{Pr}_{\mathcal{D}}\big[X(i)=s|cnt_s=j\big]\\=\sum_{j=n-i+1}^n\binom{n}{j}\cdot\mathrm{Pr}_{\mathcal{D}}\big[x<s\big]^{n-j}\cdot\mathrm{Pr}_{\mathcal{D}}\big[x=s\big]^{j}\end{gather*}\textbf{Case II:} $\mathrm{Pr}_{\mathcal{D}}\big[x<s\big]=0$ \\In this case, we have \begin{gather*}X(i)=s  \implies X(r)=s, \forall r \mbox{ s.t } 1 \leq r \leq i\end{gather*} Thus, $r^f_s=1$ and therefore $i \leq cnt_s \leq n $ and 
\begin{gather*}\mathrm{Pr}_{\mathcal{D}}\big[X(i)=s\big]=\sum_{j=i}^n\mathrm{Pr}_{\mathcal{D}}\big[X(i)=s|cnt_s=j\big]\\=\overset{n}{\underset{j=i}{\sum}}\binom{n}{j}\cdot\mathrm{Pr}_{\mathcal{D}}\big[x=s\big]^j\cdot\mathrm{Pr}_{\mathcal{D}}\big[x>s\big]^{n-j}\end{gather*}
\textbf{Case III:} Otherwise\\
For all other cases, if $cnt_s=j, j \in [n]$, then we must have $\max\{1,i-j+1\}\leq r^l_s\leq\min\{i,n-j+1\}$. Thus, we have \begin{gather*}\mathrm{Pr}_{\mathcal{D}}\big[X(i)=s\big]=\sum_{j=1}^n\mathrm{Pr}_{\mathcal{D}}\big[X(i)=s|cnt_s=j\big]\\=\overset{n}{\underset{j=1}{\sum}}\Bigg(\overset{\min\{i,n-j+1\}}{\underset{k=\max\{1,i-j+1\}}{\sum}}\Big(\binom{n}{k-1,j,n-k-i+1}\cdot\\\mathrm{Pr}_{\mathcal{D}}\big[x<s\big]^{k-1}\cdot\mathrm{Pr}_{\mathcal{D}}\big[x=s\big]^j\cdot\mathrm{Pr}_{\mathcal{D}}\big[x>s\big]^{n-k-j+1}\Big)\Bigg)\end{gather*}
\end{proof}

\subsection*{Proof for Thm. \ref{thm:formal:POPE}}\label{app:thm:fm:POPE}
\begin{proof}
Recall that in  \nam, the \OPE scheme is applied to the encodings obtained from the \name primitive.
Thus, in this case, the ciphertexts $\mathcal{C}$ preserve the rank of the encodings of  $\name(x,\Pa,\epsilon), x\in X$. Let $\widetilde{O}$ represent this set of encodings.  Additionally, let $\widetilde{O}(i)$ be  the encoding in $\widetilde{O}$ with rank $i$. Let $X''(i)$ represent the corresponding plaintext for the encoding $\widetilde{O}(i)$. Thus, for $s\in \mathcal{X}, x \sim \mathcal{D}$ and $o \sim \mathcal{D}^*$, we have
 \begin{gather*}\small \hspace{-6cm}\mathrm{Pr}_{\mathcal{D}}\big[X''(i)=s\big]=\\\hspace{1cm}\sum_{v \in \mathcal{O}}\mathrm{Pr}_{\mathcal{D}^*}\big[\widetilde{O}(i)=v\big]\cdot \mathrm{Pr}_{\mathcal{D}}\big[X''(i)=s|\widetilde{O}(i)=v\big] \\=\sum_{v \in \mathcal{O}} \mathrm{Pr}_{\mathcal{D}^*}\big[\widetilde{O}(i)=v\big]\cdot  \mathrm{Pr}_{\mathcal{D}}\big[x=s|\name(x,\Pa,\epsilon)=v\big]\\\hspace{-4.7cm}=\sum_{v\in \mathcal{O}} \mathrm{Pr}_{\mathcal{D}^*}\big[\widetilde{O}(i)=v\big]\cdot \\\hspace{1cm}\frac{\mathrm{Pr}_{\mathcal{D}}\big[\name(x,\Pa,\epsilon)=v|x=s\big]\cdot \mathrm{Pr}_{\mathcal{D}}\big[x=s\big]}{\mathrm{Pr}_{\mathcal{D}}\big[\name(x,\Pa,\epsilon)=v\big]}\\\hspace{-4.7cm}=\sum_{v\in \mathcal{O}} \mathrm{Pr}_{\mathcal{D}^*}\big[\widetilde{O}(i)=v\big]\cdot \\\hspace{1cm}\frac{\mathrm{Pr}_{\mathcal{D}}\big[\name(x,\Pa,\epsilon)=v|x=s\big]\cdot \mathrm{Pr}_{\mathcal{D}}\big[x=s\big]}{\mathrm{Pr}_{\mathcal{D}^*}\big[o=v\big]}\\=\mathrm{Pr}_{\mathcal{D}}\big[x=s\big]\sum_
 {v\in \mathcal{O}}
 \mathrm{Pr}_{\mathcal{D}^*}\big[\widetilde{O}(i)=v\big] \frac{\mathrm{Pr}\big[\name(s,\Pa,\epsilon)=v\big]}{\mathrm{Pr}_{\mathcal{D}^*}\big[o=v\big]} \end{gather*}
Thus finally,
\begin{gather*}\mathcal{L}(i,j)=\mathrm{Pr}\big[X''(i,j)=b(i,j)\big]+\textsf{negl}(\kappa)\\\mbox{[The term \textsf{negl}($\kappa$) accounts for the corresponding term in Eq. \ref{Eq:main}}\\\mbox{ for the \indfaocpae~ guarantee of the \nam scheme.]} \\=\sum_{s \in S_{b(i,j)}^{j}}\mathrm{Pr}\big[X''(i)=s\big]+\textsf{negl}(\kappa)\\=\sum_{s \in S_{b(i,j)}^{j}}\mathrm{Pr}_{\mathcal{D}}\big[x=s\big]\Big(\sum_{v \in \mathcal{O}} \mathrm{Pr}_{\mathcal{D}^*}\big[\widetilde{O}(i)=v\big]\cdot
\\ \hspace{1.5cm}\mathrm{Pr}\big[\name(s,\Pa,\epsilon)=v\big]/\mathrm{Pr}_{\mathcal{D}^*}\big[o=v\big]\ \Big) + \textsf{negl}(\kappa)\end{gather*} 
\end{proof}

\subsection{Related Work}\label{app:relatedwork}
The \ddp guarantee is equivalent to the notion of metric-based  \ldp \cite{metricDP} where the metric used is $\ell_1$-norm. Further, metric-\ldp is a generic form of Blowfish \cite{Blowfish} and d$_\chi$-privacy \cite{dx} adapted to \ldp. Other works  \cite{linear1,Xiang2019LinearAR, Geo, linear2, Ordinal1, Ordinal2} have also modelled the data domain as a metric space and scaled the privacy parameter between pairs of elements by their distance. A recent work \cite{acharya2019contextaware} propose context-aware framework of \ldp that allows a privacy designer to incorporate the application’s context into the privacy definition.

\par
A growing number of work has been exploring the association between differential privacy and cryptography \cite{DPmarry}. Mironov et al. \cite{CDP} introduced the notion of computational differential privacy where the privacy guarantee holds against a \textsf{PPT} adversary. Roy Chowdhury et al. \cite{CryptE} use cryptographic primitives to bridge the gap between the two settings of differential privacy -- \ldp and \textsf{CDP}. A line of work \cite{Prochlo, shuffle2} has used cryptographic primitives for achieving anonymity  for privacy amplification in the \ldp setting. Mazroom et al. \cite{DPAP} have proposed techniques for secure computation with \DP access pattern leakage.  Bater et al.  \cite{Shrinkwrap} combine differential privacy with secure computation for query performance optimization in private data federations. Groce et al. \cite{PSI} show that allowing differentially private leakage can significantly improve the performance of private set intersection protocols.  Vuvuzela \cite{Vuvuzela} is an anonymous communication system
that uses differential privacy to enable scalability and privacy of the messages. Differential privacy has also been used in the context of ORAMs \cite{DPORam1, DPORam2}.  
A parallel line of work involves efficient use of cryptographic primitives for differentially private
functionalities. Agarwal et al. \cite{EncryptedDP} design encrypted databases that support
differentially-private statistical queries, specifically private histogram queries. Rastogi et al.~\cite{Rastogi} and Shi et al.~\cite{Shi} proposed algorithms that allow an untrusted aggregator to periodically estimate the sum of $n$ users' values in a  privacy preserving fashion. However, both schemes are irresilient to user failures. Chan et al.~\cite{Shi2} tackled this issue by constructing binary interval trees over the users. B{\"o}hler et al. \cite{HH} solves the problem of differentially private heavy hitter estimation in the distributed setting using secure computation. Recently, Humphries at al. \cite{keyvalue} have proposed a solution for computing differentially private statistics over key-value data using secure computation. in the combined Additionally, recent works have combined \DP and cryptography in the setting of distributed learning  \cite{kairouz2021distributed, cpSGD, capc}.

\subsection{Additional Evaluation}

\begin{figure}[ht]
  
    \begin{subfigure}[b]{0.5\linewidth}
    \centering    \includegraphics[width=\linewidth]{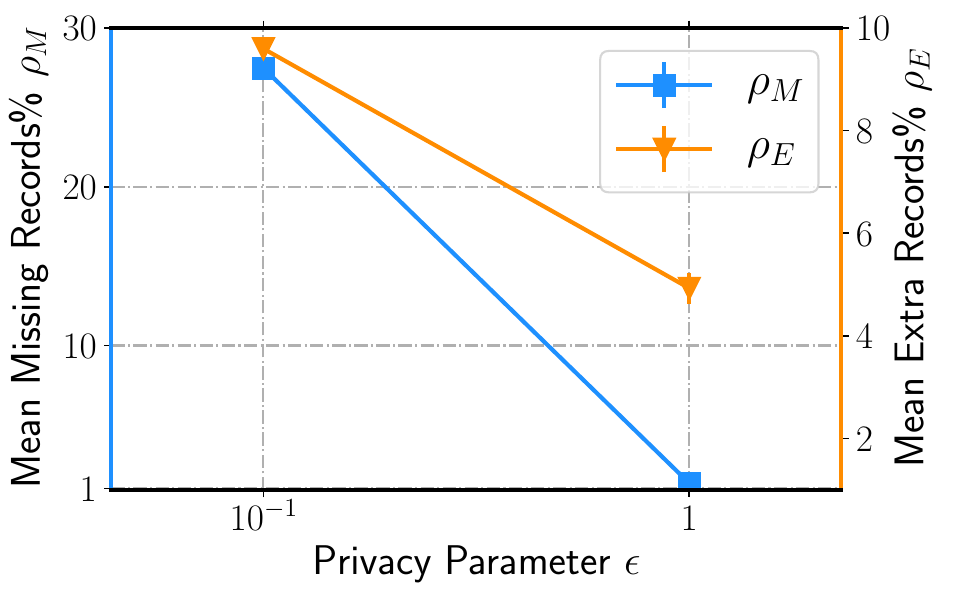}
        \caption{\textsf{Adult}: Effect of $\epsilon$}
        \label{fig:Adult:eps}\end{subfigure}
      \begin{subfigure}[b]{0.5\linewidth}
    \centering    \includegraphics[width=\linewidth]{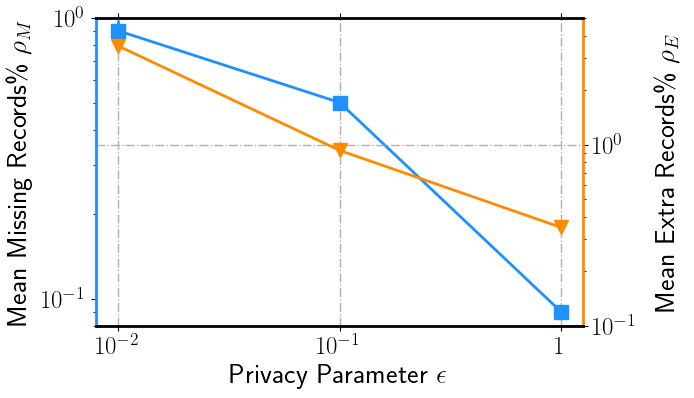}    
        \caption{\textsf{Salary}: Effect of $\epsilon$}
        \label{fig:Sal:eps}
    \end{subfigure}
   \caption{Accuracy Analysis of  \nam Cntd.}
   \label{accuracy}
\end{figure}

\subsection{Discussion}\label{app:discussion}
\nam is the first step towards integrating \OPEs and \DP. Here, we discuss several avenues for future research.\\\\
\textbf{Extension to Other Related Cryptographic Security Guarantees.}  We focused on \OPEs with ideal security (\indfaocpa) since this is the first step in the exploration of combining DP with a property-preserving encryption scheme. Extending this to more practical schemes such as BCLO~\cite{Boldyreva} and CLWW~\cite{ORE3} is a natural and important future direction. In principle, we can follow the same construction strategy as that of OP$\epsilon$, $(i)$ encode the plaintext with OP$\epsilon$c $(ii)$ encrypt the resulting encoding with the BCLO scheme or the CLWW scheme. However, the construction could be improved further as follows. Let $L(\cdot)$ denote the leakage function associated with the OPE scheme. For ideal security, $L(\cdot)$ is just the data order (which complements the order-preserving criteria of OP$\epsilon$c, Eq. 2 Def. 3.6).  The leakage functions are different for BCLO and CLWW (roughly half of the plaintext bits for BCLO and Eq. 3.1 in the CLWW paper). If we replace the first step with a $\epsilon$-DP encoding scheme that is tuned with the corresponding $L(\cdot)$ instead, then this could potentially lead to better utility. Intuitively, we want the encoding  to have differentially private $L(\cdot)$ leakage. Additionally, this could lead to smoother composition for analyzing the formal security guarantee of the resulting scheme. We will expand our discussion in the paper to include this.
 The current scheme can be trivially extended to the \textsf{IND-OCPA} security guarantee 
\cite{Boldyreva,Popa} for \OPEs by replacing Def. \ref{def:pope} with a \OPE scheme that satisfies \textsf{IND-OCPA} guarantee instead. 
Exploring connections with modular \OPEs \cite{ModularOPE1,ModularOPE2} is also an interesting future direction. The property of partial order preserving can  provide protection against certain inference attacks. For example, some attacks require access patterns for uniformly random range queries \cite{Grubbs19} or knowledge about the volume of  every range query \cite{Grubbs2}. This is clearly not possible with \nam as only queries at the granularity of the chosen partition are permitted. Hence, another future direction could be formalizing this security gain parameterized on the choice of the partition. A related path to explore here could be studying connections with the existing notion of partially order preserving encoding \textsf{POPE} proposed by Roche et. al \cite{POPE}. 
A recent line of work has focused on providing formal guarantees against some specific types of attacks in the context of encrypted databases \cite{pancake,Paterson_sol,Kerschbaum19,kamara}. 
Our model is distinct from all the above mentioned approaches. Additionally, since the \cdp guarantee holds regardless of the type of inference attacks, it would be interesting to see if it can be combined with the above approaches for a stronger formal guarantee or better efficiency. \par Beyond \OPEs, secure ordering could be required in a  distributed setting where $n$ mutually untrusting parties, each holding a data point, want to compute a sorted order over their data (generalization of the classic Yao's millionaires' problem \cite{Yao,securesorting}). \OPEs are ill-suited for this setting because $(1)$ currently \OPEs are defined only in private key cryptography which means that a single malicious agent posing as a data owner can compromise the protocol $(2)$ \OPEs (satisfying \indfaocpae and \textsf{IND-OCPA}) are stateful and mutable \cite{Kerschbaum19,Popa} which is not feasible in a distributed setting. This requires the use of multi party computation (\textsf{MPC}) techniques. A straightforward way to extend is to compute over the outputs of the \name primitive. Proposing techniques for improved utility is an important future work.
\\\\
\textbf{Compromised Querier.} In the context of a database encrypted under a \OPE scheme, a querier has access only to the records that precisely belong to the queried range. However, in our setting the querier might know the values of some records that fall outside the queried range (Sec. \ref{sec:application}). This might lead to additional leakage, when compared to the case of a \OPE encrypted database, in the event the querier is compromised. One way to prevent this is to use an attribute-based encryption scheme \cite{abe} for $\overline{\mathcal{E}}$ where the decryption is possible only if the record belongs to the queried range.   \\\\
\textbf{Support for Non-ordinal Data.}
Currently, $\epsilon$-\ddp (equivalently \cdp) provides a semantically useful privacy guarantee only for data domains that have a naturally defined order. A possible future direction can be exploring how to extend this guarantee for non-ordinal domains (like categorical data). One such way could be associating the categories of the non-ordinal domain with some ordinal features like popularity \cite{Ordinal1} and defining the guarantee w.r.t to these ordinal features instead.\\\\
\textbf{Extension of \ldp Mechanisms.} The performance of the algorithms presented in Sec. \ref{sec:opec:ldp} could be improved by borrowing techniques from the existing literature in \ldp. For example, the partition for \name could be learnt from the workload factorization mechanism from  \cite{Workload}. In another example, a B-ary tree could be constructed over the input domain using \name for answering range queries \cite{range}.
\\\\\textbf{Less State Information for Clients.} For \nam, in fact the clients need to store less state information than for \OPEs satisfying \indfaocpa as discussed below. Clients for any \OPE~scheme, $\mathcal{E}$, (satisfying the \indfaocpa guarantee) need to store two pieces of state information for each unique value of $\mathcal{X}$ that appears in the dataset $X$ to be encrypted (see \cite{frequencyHiding2}). For example, for input domain $\mathcal{X}=[100]$ and a dataset $X=\{42,45,45,50,88,67,67,77,90,98,98,98,98\}$ drawn from this domain, the client needs to store two information $\{\max_{\mathcal{E}}(x),\min_{\mathcal{E}}(x)\}$ for $x \in \{42,45,50,88,67,77,90,98\}$.  Recall that \nam~ applies \OPE~ to the output of the \name~ primitive. This implies that for \nam, $\mathcal{E}^{\dagger}$, the client needs to store the state information only for each encoding in $\mathcal{O}$ of the underlying \name~ primitive.
For the above mentioned example, consider a partition $\Pa=\{[1,20],[21,40],[41,60],[61,80],[81,100]\}$ with corresponding encodings $\mathcal{O}=\{10,30,50,70,90\}$. Here, the client needs to store $\{\max_{\mathcal{E}^{\dagger}}(o),\min_{\mathcal{E}^{\dagger}}(o)\}$ only for $o \in \mathcal{O}=\{10,30,50,70,90\}$. This means that clients now need to store less state information for \nam~ than for \OPE. 
\\\\\textbf{Extension of \name.} The \name primitive can be extended to generic metric space along the lines of previous literature \cite{metricDP,dx}. This would support arbitrary partition instead of just non-overlapping intervals. For this, first sort the input domain $\mathcal{X}$ according to the metric $d(\cdot)$.  Then divide the sorted domain, $\mathcal{X}_S$, into  non-overlapping intervals which determines the partition, $\Pa$, for the \name primitive. Alg. \ref{alg:1} can now be defined on $\Pa$ with metric $d(\cdot)$.
\par Recent work in database theory has explored efficient $k$-top query answering mechanisms~\cite{deep2018compressed, deep2020matrix, deep2021enumeration, deep2021space, deep2021ranked}. \name can be used in conjunction with these mechanisms for guaranteeing data privacy.  
\\\\\textbf{Choice of Partition.}  As described in Sec. \ref{sec:application} (Remark 5), the partition $\Pa$ is chosen completely from \textit{an utilitarian perspective} in the context of encrypted databases -- it results in an \textit{accuracy-overhead trade-off}  (accuracy -- number of correct records retrieved; overhead -- number of extra records processed). The data owner can choose $\mathcal{P}$ based on some (non-private) prior on the dataset. One strategy is to use equi-depth partitioning. Our experimental results in Sec. VII show that this strategy works well in practice. Moreover, the partition can be changed dynamically as long as the encoding domain $\mathcal{O}$ of the underlying \name primitive  has enough wiggle room. For instance, for input domain $\mathcal{X}=[100]$, let the initial partition be $\Pa=\{[1,20],[21,40], $\\$[41,60], [61,80], [81,100]\}$ over a input domain $[100]$. Let the corresponding encodings be $\mathcal{O}=\{1, 21, 41, 61, 81\}$. Now, if in the future the interval $[1,40]$ needs to be further partitioned into $\{[1,10],[11,20],[21,30],[31,40]\}$, it can be performed as follows:
\squishlist
\item retrieve and delete all records from the database in the range $[1,40]$ (this step might incur some loss in accuracy)
\item assign the encodings $\{1,11,21,31\}$ for the aforementioned sub-partition
\item insert back the records encrypted under the new encoding 
\squishend
However, the cost here is that every update consumes an additional $\epsilon$-\ddp privacy budget for the updated records.
\\\\\textbf{Additional Advantages of Partitioning.}
Two additional advantages of partitioning are:
\squishlist
\item Clients for any \OPE (satisfying \indfaocpa) need to store some state information. However, \nam requires much less information storage due to partitioning which is advantageous for resource-constrained clients (as illustrated above).

\item Although we don't consider it in this paper, prior work shows that partial-order preservation improves security~\cite{Grubbs19,Grubbs2}. Formalizing this security gain of partitioning is an interesting future direction.
\squishend
\vspace{0.3cm}
\noindent\textbf{Encrypting Multiple Columns.} For encrypting records with multiple columns, we can encrypt each column individually under the \nam scheme (satisfying $\epsilon$-\ddp). Then, from the composition theorem of \ddp (Thm. \ref{thm:seq}), we would still enjoy $c\cdot\epsilon$-\ddp guarantee over the entire dataset where $c$ is the total number of columns. 

\end{document}